\documentclass[10pt,notitlepage,a4paper,longbibliography]{revtex4-2}
\pdfoutput=1
\usepackage[utf8]{inputenc}
\usepackage[T1]{fontenc}
\usepackage{xcolor}
\usepackage{nicematrix}
 \usepackage{longtable}
 \addtocontents{toc}{\setcounter{tocdepth}{1}}
\usepackage[caption=false]{subfig}

\usepackage[left=1.2in,right=1.2in,top=1.2in,bottom=1.2in]{geometry}
\usepackage{bbm}
\usepackage{pifont}
\usepackage{afterpage}
\usepackage{tcolorbox}
\usepackage{amsmath,amssymb}
\usepackage{amsthm}
\usepackage{makecell}
\usepackage{indentfirst}
\usepackage{mathrsfs}
\usepackage{graphicx}
\usepackage{hyperref}
\usepackage{thmtools}
\usepackage{csquotes}
\usepackage{verbatim}
\usepackage{bookmark}
\usepackage{booktabs}
\usepackage{url}
\usepackage{physics}
\usepackage[shortlabels]{enumitem}
\usepackage{dsfont}
\usepackage{lineno}
\usepackage{mathtools}

\usepackage{multirow}
\usepackage{times}
\usepackage{tikz}
\usepackage{nicefrac}
\theoremstyle{plain}
\newtheorem{theorem}{Theorem}
\newtheorem{lemma}[theorem]{Lemma}

\newtheorem{assume}[theorem]{Assumption}

\theoremstyle{definition}

\newtheorem{remark}{Remark}

\newcommand{\thmref}[1]{Theorem~\ref{#1}}

\newcommand{\lemref}[1]{Lemma~\ref{#1}}
\newcommand{\secref}[1]{Section~\ref{#1}}
\newcommand{\appref}[1]{Appendix~\ref{#1}}
\newcommand{\figref}[1]{Figure~\ref{#1}}

\newcommand{\bigH}{\overline{H}}
\newcommand{\bigU}{\overline{U}}
\newcommand{\dt}{\Delta t}
\newcommand{\ds}{\Delta s}

\makeatletter
\newcommand*{\rom}[1]{\expandafter\@slowromancap\romannumeral #1@}
\makeatother

\newcommand{\wt}[1]{\widetilde{#1}}

\graphicspath{{assets}}
\hypersetup{
    colorlinks=true,
    linkcolor=blue!50!black!70!,
    citecolor=red!70!black!70!,
}

\newcommand\myeq[1]{\mathrel{\stackrel{\makebox[0pt]{\mbox{\normalfont\tiny #1}}}{=}}}

\newcommand{\Natural}{\mathbb{N}}

\newcommand{\Real}{\mathbb{R}}
\newcommand{\Complex}{\mathbb{C}}
\newcommand{\ee}{\mathbb{E}}

\newcommand{\eps}{\epsilon}

\newcommand{\exact}{\varphi_{\text{true}}}
\newcommand{\exactorth}{\varphi_{\perp}}

\definecolor{cellbg}{RGB}{255,237,237}
\definecolor{cellbgblue}{RGB}{204,204,255}

\newcommand{\unit}{\mathbb{I}}
\newcommand{\id}{{\normalfont \bf I}}
\newcommand{\hbt}{\mathcal{H}}
\newcommand{\leftmarginnew}{0.75cm}
\newcommand{\h}{{\boldsymbol{h}}}
\newcommand{\DAS}{digitized adiabatic simulation}
\newcommand{\DAScap}{Digitized adiabatic simulation}
\newcommand{\one}{\mathbbm{1}}
\newcommand{\augspace}{\mathscr{X}}
\newcommand{\timedom}{\mathbb{T}}
\newcommand{\exactU}[2]{\mathcal{{U}}(#1,#2)}
\newcommand{\hdu}{\wt{U}}

\newcommand{\eT}[1]{e_{\mathcal{T}}^{#1}}
\newcommand{\expT}[1]{\exp_\mathcal{T}\big(#1\big)}
\newcommand{\HDR}{Huyghebaert and De Raedt}
\newcommand{\SH}{Sambe-Howland}

\newcommand{\dthalf}{\frac{\dt}{2}}

\usepackage{soulutf8}
\usepackage{xcolor}
\sethlcolor{yellow!20}

\begin{document}
\allowdisplaybreaks

\title{Unifying framework for quantum simulation algorithms for time-dependent Hamiltonian dynamics}
\author{Yu Cao$^{1,2,3}$ Shi Jin$^{1,2,3}$ Nana Liu$^{1,2,3,4}$}

\affiliation{$^1$ School of Mathematical Sciences, Shanghai Jiao Tong University, Shanghai, 200240, China}
\affiliation{$^2$ Institute of Natural Sciences, Shanghai Jiao Tong University, Shanghai 200240, China}
\affiliation{$^3$ Ministry of Education Key Laboratory in Scientific and Engineering Computing, Shanghai Jiao Tong University, Shanghai 200240, China}
\affiliation{$^4$ Global College, Shanghai 200240, China.}
\email{Corresponding author: nana.liu@quantumlah.org}

\begin{abstract}
Recently, there has been growing interest in simulating time-dependent Hamiltonians using quantum algorithms, driven by diverse applications, such as quantum adiabatic computing. While techniques for simulating time-independent Hamiltonian dynamics are well-established, time-dependent Hamiltonian dynamics is less explored and it is unclear how to systematically organize existing methods and to find new methods. Sambe-Howland's continuous clock elegantly transforms time-dependent Hamiltonian dynamics into time-independent Hamiltonian dynamics, which means that by taking different discretizations, existing methods for time-independent Hamiltonian dynamics can be exploited for time-dependent dynamics. In this work, we systemically investigate how Sambe-Howland's clock can serve as a unifying framework for simulating time-dependent Hamiltonian dynamics. Firstly, we demonstrate the versatility of this approach by showcasing its compatibility with analog quantum computing and digital quantum computing. Secondly, for digital quantum computers, we illustrate how this framework, combined with time-independent methods (e.g., product formulas, multi-product formulas, qDrift, and LCU-Taylor), can facilitate the development of efficient algorithms for simulating time-dependent dynamics. This framework allows us to (a) resolve the problem of finding minimum-gate time-dependent product formulas; (b) establish a unified picture of both Suzuki's and Huyghebaert and De Raedt's approaches; (c) generalize Huyghebaert and De Raedt's first and second-order formula to arbitrary orders; (d) answer an unsolved question in establishing time-dependent multi-product formulas; (e) and recover continuous qDrift on the same footing as time-independent qDrift. Thirdly, we demonstrate the efficacy of our newly developed higher-order Huyghebaert and De Raedt's algorithm through digital adiabatic simulation.
\end{abstract}

\maketitle
\tableofcontents

\newpage
\section{Introduction}

The concept of quantum computers, which aim to surpass the fundamental limits of classical computers, was proposed by Feynman in 1982 \cite{feynman_simulating_1982}. Since then, this direction  has attracted significant attention. 
As quantum computers are designed using the principles of quantum mechanics, and closed quantum systems inherently follow Schrödinger's equation (i.e., Hamiltonian dynamics with unitary evolution), the problem of simulating Hamiltonian dynamics has emerged as a core challenge across multiple disciplines. The pioneering work of Lloyd \cite{lloyd_universal_1996} introduced a foundational approach of developing quantum algorithms to simulate Hamiltonian dynamics in physics, and this study has inspired research on Hamiltonian simulation using quantum devices over the past three decades.
To date, numerous quantum and quantum-classical algorithms have been proposed for simulating time-independent Hamiltonian dynamics. These include, but are not limited to, product formulas
(such as the Lie-Trotter-Suzuki formula and its variants) \cite{trotter_product_1959,suzuki_generalized_1976,suzuki_fractal_1990,suzuki_general_1991}, multi-product formulas \cite{chin_multi_product_2011,childs_hamiltonian_2012,low_well_conditioned_2019,endo_mitigating_2019},
randomized algorithms \cite{poulin_quantum_2011,zhang_randomized_2012,childs_faster_2019,campbell_random_2019},  Taylor's series \cite{berry_simulating_2015} with Linear Combination of Unitaries (LCU) \cite{childs_hamiltonian_2012}, qubitization and quantum signal processing \cite{low_hamiltonian_2019}. Furthermore, detailed methods or analysis have been adapted for various situations, such as multi-scale Hamiltonians \cite{bosse_efficient_2024,sharma_hamiltonian_2024} and unbounded Hamiltonians \cite{suzuki_hybrid_1995}.

In contrast, time-dependent Hamiltonian simulation has received relatively less attention. There are several notable classical approaches. 
One prominent example is the time-dependent product formula developed by Huyghebaert and De Raedt \cite{huyghebaert_product_1990} and contemporaneously, by Suzuki \cite{suzuki_general_1993}. Another popular  approach is the Magnus-based method \cite{blanes_magnus_2009}, which often requires
combination with product formulas or LCU to be implemented on quantum hardware.
Recently, researchers have, for instance, explored time-dependent multi-product formula \cite{watkins_time_2024}, time-dependent qubitization \cite{watkins_time_2024,mizuta_optimal_2023,mizuta_optimal_multiperiodic_2023}, and time-dependent qDrift \cite{berry_time_dependent_2020}, as well as time-dependent algorithm with minimum gate counts \cite{ikeda_minimum_2023}. 
However, it appears that there is a lack of a systematic approach to leverage the benefits and advances in recent studies on time-independent Hamiltonian simulation problems for time-dependent situations within the community.
Several questions raised by some of these works remain partially addressed, leaving open questions or only provide limited solutions for special instances. 

This motivates us to provide a \emph{systemic way} through \SH{}'s clock to (I)  recover many existing time-dependent algorithms using time-independent algorithms; (II) establish the  connection between some popular algorithms in literatures; (III) provide the mathematical machinery of using \SH{}'s clock; (IV) and more importantly, to develop new high-order \HDR{}'s algorithm, which was discussed as an open question in a recent work \cite{ikeda_minimum_2023}.

To the best of our knowledge, \SH{} clock, an elegant and powerful tool, appears to have been under-appreciated for time-dependent Hamiltonian simulation algorithms.
We also observe that technical challenges arise for discrete variants of \SH{}'s clock in designing quantum algorithms. Since it seamlessly connects time-dependent Hamiltonian dynamics with time-independent Hamiltonian dynamics, it allows us to use the richer set of tools from the latter to find new algorithms for the former case. 
 Thus this motivates us to investigate how the \emph{continuous form} of \SH{}'s clock can be exploited as a unifying framework from which existing time-dependent Hamilton simulation algorithms can be derived, and also to enhance existing algorithms by applying existing methods for time-independent Hamiltonian dynamics to time-dependent Hamiltonian dynamics. 

\subsection{Problem setup}
 
 To properly setup the problem, our goal is to simulate the following time-dependent Schrödinger's equation,
\begin{align*}
\partial_t \ket{\psi(t)} = -i H(t) \ket{\psi(t)}, \qquad t\in [0,1].
\end{align*}
Such equations might arise from many scenarios, e.g., \DAS{} \cite{albash_adiabatic_2018}, dynamics in the interaction picture \cite{low_hamiltonian_2019_paper_b}, or from the resulting equation after applying Schrödingerisation \cite{jin2212quantum,jin2022quantum, jin_analog_2024, time_dilation_2023} for quantum simulation of general linear ordinary differential equations (ODEs) and partial differential equations (PDEs). 
The time-evolution operator for $s\le t$ is defined as follows:
\begin{align}
\label{eqn::exactu}
\exactU{t}{s} := \exp_{\mathcal{T}}\big(-i \int_{s}^{t} H(r)\ \dd r\big).
\end{align}
The notation $\mathcal{T}$ is the time-chronological ordering operator.
The ordering of the non-commutative Hamiltonian is conceived as an inconvenience in simulating the time-dependent Schrödinger's equation in a digital quantum circuit. As a result, this has been a very active research topic in recent years to address this issue.

 \begin{assume}
{We consider Hamiltonians with the following form
 \begin{align}
 \label{eqn::decompose}
H(t) = \sum_{k=1}^{\Lambda} H_k(t).
\end{align}
This structure has received much attention and has been used in numerous studies.
We assume that $H_k(t)$ are smooth with respect to the time variable. The time domain $\timedom$ can be either a 1D torus or simply $\Real$. All operators $H_k(t)$ are defined on a finite-dimensional Hilbert space $\hbt$.}
\end{assume}

{For instance, one may consider that}
\begin{align}
\label{eqn::Hk}
H_k(t) = f_k(t) \h_k, \qquad f_k \in C^\infty(\timedom),
\end{align}
and $\h_k$ are operators on a finite-dimensional Hilbert space $\hbt$ with bounded norm.
{Though most algorithms discussed below only relies on the form \eqref{eqn::decompose} and the regularity assumption above,  we remark that the choice \eqref{eqn::Hk} is also not} very restrictive, e.g., for studying \DAS.
The Hamiltonian $H(t)$ with the form \eqref{eqn::Hk} belongs to the linear combination (LC) model. When we further assume that $\h_k$ are unitary, this becomes the linear combination of unitary model (LCU model). Notably, it is well-established that the LC model also encompasses the sparse 
Hamiltonian case, as shown in e.g., \cite{berry_efficient_2007,berry_time_dependent_2020}.

\subsection{\SH{}’s clock}

The \SH{}'s clock Hamiltonian introduces a conjugate pair of the time variable (commonly referred to as the energy in physical literature) 
into the original time-dependent Hamiltonian, effectively transforming a time-dependent Schrödinger's equation into a time-independent one \cite{sambe_steady_1973,howland1974stationary}. This formalism is also known as the $(t,t')$ method in the literature \cite{peskin_solution_1993}, and has been applied to various situations, including steady-state perturbation theory \cite{sambe_steady_1973}, scattering theory \cite{howland1974stationary}, open quantum system \cite{burgarth_control_2023}, quantum computing \cite{watkins_time_2024} (where a discrete version was employed), simulation of time-dependent PDEs using analog quantum computing \cite{time_dilation_2023}.
The \SH{}'s clock Hamiltonian proceeds as follows \cite{sambe_steady_1973,howland1974stationary,reed1975methods,peskin_solution_1993}:
by augmenting the quantum dynamics with a continuous auxiliary state $s$ and defining the \SH{}'s clock Hamiltonian:
\begin{align}
\label{eqn::bigH}
\bigH &= \hat{p}_s \otimes \id + H(\hat{s}).
\end{align}
Here the $s$ mode - with corresponding quadrature operators $\hat{s}$ and $\hat{p}_s$ where $[\hat{s}, \hat{p}_s]=i\mathbf{I}_s$ - is referred to as the clock mode. 

With the extended Hamiltonian $\bigH$, one can evolve the closed Hamiltonian dynamics in the extended space:
\begin{align}
\label{eqn::full_dynamics}
\begin{aligned}
\partial_t \ket{\Psi(t)} &= -i \bigH \ket{\Psi(t)},\\
\ket{\Psi(0)} &= \Big(\int G(s)\ket{s}\Big) \otimes \ket{\psi(0)},
\end{aligned}
\end{align}
which amounts to solving the following PDE 
(see e.g., \cite[Chapter~X.12]{reed1975methods} and also \cite{sambe_steady_1973,howland1974stationary,peskin_solution_1993,time_dilation_2023})
\begin{align}
\begin{aligned}
\partial_t \Psi(t, s, \cdot) &= -\partial_{s} \Psi(t, s, \cdot) - i H(s) \Psi(t, s, \cdot)\\
\Psi(0, s, \cdot) &= G(s) \psi(0).
\end{aligned}
\end{align}
As the first term $-\partial_{s} \Psi(t, s, \cdot)$ plays the role as the transport, 
one can directly use the method of characteristics and verify that the exact solution takes the form \cite[Theorem 1]{howland1974stationary}:
\begin{align}
\label{eqn::Psi}
\ket{\Psi(t)} = \int\ \dd s\ G(s-t) \ket{s}\otimes \exactU{s}{s-t} \ket{\psi(0)}.
\end{align}

\begin{remark}
This expression has tight connection to linear combination of unitaries (LCU), as it has the form of an infinite linear combination (namely, the integration) of unitary evolution with fixed time span $t$ and with various terminal time $s$.
\end{remark}

Note that the state $G$ could be totally artificial (namely, introduced merely for the mathematical purpose), if we only use this formalism as an intermediate step to derive some algorithms.
If the state $G$ needs to be physical (e.g., used for analog quantum computing), a common choice of $G$ is to use the Gaussian state where 
\begin{align}
\label{eqn::G}
G(s) = \frac{1}{(2\pi\omega^2)^{1/4}}e^{-\frac{s^2}{4\omega^2}}.
\end{align}
This formalism is different from the Feynman's clock, which encodes the full-history Hamiltonian simulation problem into the ground state problem \cite{mcclean_feynmans_2013}, whereas \SH{}'s clock only worries about preparing the quantum state at a particular time.
We summarize possible variants as well as our focus in this paper in Table~\ref{table::summary}.
 
\begin{table}
\caption{This table summarizes how various choices of $\omega$ and how different discretizations of the clock mode connects to different quantum algorithms. We colored our results and focus in red. The limit $\omega\to 0$ means it is a localized state and $\omega\to\infty$ is used to denote that $G \equiv 1$.}
\label{table::summary}
\vspace{\baselineskip}
\begin{center}
\begin{NiceTabular}{cc}[first-row,first-col, hvlines, columns-width=auto, name=A]
 & \makecell{\Block[]{1-1}{\bf Hamiltonian (digital)}}  & \makecell{\Block[]{1-1}{{\bf Hamiltonian (analog)}}}  \\
\toprule
 {\bf clock (analog)} & \makecell{\Block[]{1-1}{hybrid system}}   & \makecell{\Block[draw=red, fill=cellbg]{1-1}{{\bf fully analog} \\(see \secref{sec::aqc}) \\ $\omega = \order{\frac{1}{T \norm{\h_2}}}$}} \\
 \makecell{\Block[]{1-1}{{\bf clock (digital)}}} & \makecell{\Block[draw=blue, fill=cellbgblue, rounded-corners]{1-1}{{\bf fully digital} \\(see below)}} & hybrid system \\
\end{NiceTabular}
\end{center}
\vspace{\baselineskip}
\begin{center}
\begin{NiceTabular}{cccc}[first-row, first-col, hlines, columns-width=0.2\textwidth, name=B]
 & \makecell{\Block[]{1-1}{$\boldsymbol{\omega}$}} & {\bf Discretize $s$ mode}  & {\bf Algorithms} & {\bf Relevant Sections} \\
 & \makecell{\Block[]{1-1}{limit $\omega\to 0$\\ ($G$ is a delta measure)\\ or limit $\omega\to \infty$ \\ ($G=1$)}}  & \makecell{\Block[]{1-1}{{operator splitting}}}  & \makecell{\Block[fill=cellbg]{1-1}{product formula and its variants\\ (e.g., multi-product formula,\\ randomized product formula)}}  &  \makecell{\Block[fill=cellbg]{1-1}{\secref{sec::pointwise}, \secref{section::HDR},\\ \secref{sec::mpf}, \secref{sec::qDrift}}} \\
& \makecell{\Block[]{1-1}{limit $\omega\to 0$\\ ($G$ is a delta measure)\\ or limit $\omega\to \infty$ \\ ($G=1$)}} & \makecell{\Block[]{1-1}{{Taylor's expansion}}}  & \makecell{\Block[fill=cellbg]{1-1}{LCU-Taylor based algorithm}}  &  \makecell{\Block[fill=cellbg]{1-1}{ \secref{sec::taylor}}} \\
& when $0<\omega \ll 1$ & \makecell{\Block[]{1-1}{finite difference}}  & \makecell{\Block[]{1-1}{qubitization-based method;\\ so far limited to the first-order \cite{watkins_time_2024}}}  & / \\
& \makecell{\Block[]{1-1}{limit $\omega\to\infty$\\ ($G=1$)}} & \makecell{\Block[]{1-1}{spectral method}}  &  \makecell{\Block[]{1-1}{\vspace{0.3\baselineskip} qubitization-based method;\\ log scale wrt error tolerance $\eps$ \\ under conditions \cite{mizuta_optimal_2023,mizuta_optimal_multiperiodic_2023}}}  &  / \\
\end{NiceTabular}
\end{center}
\tikzset{A/.style = {name prefix = A-} , B/.style = {name prefix = B-}} 
\begin{tikzpicture}[remember picture, overlay, blue, thick, opacity=0.4, dashed]
\draw [A]  (3-|1)  { [B] -- (1-|1) }
           (3-|2) { [B] -- (1-|5) }
;
\end{tikzpicture}
\end{table}

\subsection{Related works for time-dependent Hamiltonian simulation}

In this subsection, we will provide an overview of existing methods for time-independent Hamiltonian simulations, with a particular focus on related works that explore quantum algorithms for time-dependent Hamiltonian dynamics. The following \emph{non-exhaustive discussion} will introduce and examine methods that may be adaptable to the \SH{}'s clock, as well as the Magnus expansion as a parallel framework. Readers may refer to e.g., \cite{miessen_quantum_2023} for a review and additional discussions.

\bigskip
{\noindent {\bf Product formulas:}} The product formula is arguably the most extensively studied family of quantum algorithms for Hamiltonian simulation \cite{trotter_product_1959,suzuki_generalized_1976,de_raedt_product_1987,suzuki_fractal_1990,suzuki_general_1991,huyghebaert_product_1990}, also known as the Lie-Trotter-Suzuki formula. 
In the time-independent case, this approach has been analyzed e.g., in \cite{berry_efficient_2007} and more recently, the error estimate has been meticulously studied in \cite{childs_theory_2021}. 
This family of schemes boasts benefits from commutator scaling \cite{childs_theory_2021}, a feature that many scaling-optimal algorithms (to be discussed later) may not possess.
Regarding time-dependent Hamiltonian simulations,
\HDR{} generalized the Strang-splitting scheme (a second-order scheme) to the time-dependent case \cite{huyghebaert_product_1990} in 1990s. Concurrently, Suzuki proposed a general framework for deriving time-dependent Hamiltonian simulation algorithms up to any order using low-order product formulas \cite{suzuki_general_1993,hatano_finding_2005}. A detailed error analysis for Suzuki's scheme \cite{suzuki_general_1993} was further developed in \cite{wiebe_higher_2010} and adaptive schemes have been explored in e.g., \cite{wiebe_simulating_2011,zhao_adaptive_2024} to enhance efficacy for practical considerations.

\medskip
{\noindent {\bf Multi-Product formulas:}} 
For product-based formulas, the way to achieve arbitrary high-order algorithm is to use a large number of iterative passes of local Hamiltonian evolution operator (namely, the multiplicative construction). The multi-product formula instead uses the idea of Richardson's extrapolation (namely, additive construction) and the time-independent case was studied  in e.g., \cite{chin_multi_product_2011,childs_hamiltonian_2012,low_well_conditioned_2019,endo_mitigating_2019,aftab_multi_product_2024}. 
For the time-dependent case, Geiser combined product-based formulas with extrapolation techniques \cite{geiser_multi_product_2011} to develop multi-product formulas. A different variant was recently posted as an open technical question in \cite{watkins_time_2024} for well-conditioned multi-product formula. Notably, when combined with LCU techniques, 
multi-product formulas are known to achieve logarithmic error scaling \cite{low_well_conditioned_2019,aftab_multi_product_2024}, whereas the error scaling for product-based formulas is always suboptimal and algebraic.

\medskip
{\noindent {\bf Randomized algorithm:}} 
Various randomization ideas have been found to be practical and useful in quantum algorithms for Hamiltonian simulation. 
Examples include randomizing the ordering in product formulas \cite{zhang_randomized_2012,childs_faster_2019}, randomizing the time step \cite{zhang_randomized_2012,poulin_quantum_2011}, and randomly selecting local unitary evolutions in a full sequence (known as the \enquote{qDrift algorithm}) \cite{campbell_random_2019}. In particular, the latter idea has attracted much attention recently, and has been generalized to stochastic Hamiltonian sparsification \cite{ouyang_compilation_2020}, partially random algorithm \cite{jin_partially_2023,hagan_composite_2023}, qSWIFT (a high-order generalization of qDrift) \cite{nakaji_high_order_2024}, cost-informed choice for the importance weight \cite{kiss_importance_2023}, and others. 
As for time-dependent Hamiltonian simulation, a time-dependent analog of qDrift, also known as the \enquote{continuous qDrift}, was proposed in \cite{berry_time_dependent_2020}. 
Although this scheme is insightful, there appears to be a lack of deeper structural connection 
between the time-dependent case and time-independent case, which we will discuss and address later.

\medskip
{\noindent {\bf Taylor's series (TS) method and Dyson's series (DS) method:}} Taylor's expansion for time-independent Hamiltonian simulation algorithm in LCU form was initially studied in \cite{berry_simulating_2015}, which demonstrated that the complexity scales logarithmically with respect to the error precision $\order{T \frac{\log(\nicefrac{T}{\eps})}{\log \log(\nicefrac{T}{\eps})}}$ and this scaling cannot be achieved using product formulas with finite orders. As a remark, for the 1D Heisenberg model, it has been demonstrated in  \cite{childs_toward_2018} that Taylor's series method is not competitive empirically, compared with both product formula and quantum signal processing. The time-dependent generalization of Taylor series method was developed in \cite{kieferova_simulating_2019,low_hamiltonian_2019_paper_b} using Dyson's expansion, further improved in \cite{berry_time_dependent_2020} to obtain $L^1$ norm scaling estimate, and also studied in \cite{chen_quantum_2021} using permutation expansion.

\medskip
{\noindent {\bf Quantum Signal Processing/Qubitization:}} 
The qubitization technique for time-independent Hamiltonian simulation has been shown to 
achieve optimal complexity scaling, matching the theoretical lower bound \cite{low_hamiltonian_2019}. This technique utilized the Jacobi-Anger expansion, differing from other related works that employ Taylor's type expansions \cite{berry_simulating_2015}. 
The initial study on its time-dependent generalization was presented in \cite{watkins_time_2024}, which is limited to a first-order algorithm. Another approach has been presented recently in \cite{mizuta_optimal_2023} for the case where time schedules are periodic (also known as the Sambe's space). A further generalization to multi-periodic time-dependent Hamiltonians has been explored in \cite{mizuta_optimal_multiperiodic_2023}. 
For time-periodic Hamiltonians, their query complexity in \cite{mizuta_optimal_2023} matches the optimal one provided by time-independent qubitization. 
This result has also demonstrated, from a different perspective, the value of \SH{}'s formalism in designing quantum algorithms at least for periodic case. Consequently, \SH{}'s 
clock for quantum algorithms warrants a more systematic study and presentation.

\medskip
{\noindent {\bf Magnus expansion:}}
Magnus expansion has been a widely studied tool for time-dependent dynamics \cite{blanes_magnus_2009}, as well as Hamiltonian simulation problems recently, as seen in works such as \cite{an_time_dependent_2022,fang_time_dependent_2024,bandrauk_exponential_2013,ikeda_minimum_2023,chen_quantum_2023,casares_quantum_2024,bosse_efficient_2024,sharma_hamiltonian_2024}.
 Typically, after performing a finite truncation of the Magnus expansion, it is necessary to additionally employ a quadrature scheme to estimate the integration if required, 
and then apply Hamiltonian simulation techniques like product formulas 
\cite{bandrauk_exponential_2013,ikeda_minimum_2023,chen_quantum_2023,casares_quantum_2024,bosse_efficient_2024,sharma_hamiltonian_2024}, or LCU and quantum signal processing (see below) \cite{an_time_dependent_2022,fang_time_dependent_2024}. 
A key challenge in applying the Magnus expansion to quantum algorithms is the need to implement nested commutators. The community has proposed at least two approaches to address this issue.
One is to use the commutator-free quasi-Magnus (CFQM) operators \cite{blanes_fourth_2006,alvermann_high_order_2011,casares_quantum_2024}. 
As remarked above, while the CFQM method effectively handles time-dependence, it requires additional Hamiltonian simulation algorithms, such as product formulas, to be used in conjunction with it.
Another commutator-free approach involves leveraging 
Magnus-based ideas to directly develop asymmetric product formulas, as demonstrated in
\cite{ikeda_minimum_2023,chen_quantum_2023}. We remark that Magnus' expansion is a parallel framework in the same way as \SH{}'s clock, and it would be interesting to see if there is any chance to further unify them.

\bigskip\bigskip
\begin{table}[h!]
\caption{This Table summaries how \SH{}'s clock (with $\omega \to 0$) can be used as a mathematical tool to uncover existing quantum algorithms and also unravel new quantum algorithms for time-dependent Hamiltonian simulation. For the ordered decomposition, we only list the case where the number of local Hamiltonian is $\Lambda = 2$ for simplicity; \enquote{Gate (general)} means the number of gates for one-step approximation without any assumption on original Hamiltonian, whereas \enquote{Gate (DAS)} means the number of gates  for digital adiabatic simulation where $H_k(s) = f_k(s) \h_k$. We assume that time-independent decomposition has $q$-cycles as in \eqref{eqn::product_2_body} below.}
\medskip
\label{table::digital}
\renewcommand{\arraystretch}{1.5}
\begin{longtable}{p{0.22\textwidth}p{0.2\textwidth}p{0.23\textwidth}p{0.15\textwidth}p{0.15\textwidth}}
\multicolumn{5}{c}{{\bf \SH{}'s Clock with Product Formula $\boldsymbol{\bigH = A_1 + A_2 + A_3}$}} \medskip\\
\toprule
Ordered Decomposition & Time-independent & Time-dependent & Gate (general) & Gate (DAS)\\
\hline
\multirow{2}{*}{\makecell{$A_1 = \hat{p}_s\otimes \id$\\ $A_2 = H_1(\hat{s})$\\ $A_3 = H_2(\hat{s})$}} & Suzuki \cite{suzuki_general_1991} & Suzuki \cite{suzuki_general_1993} & $2\Lambda q - q$ & $2\Lambda q - 2q+1$ \\
\cline{2-5}
& \makecell[l]{any product formula \\ with $q$-cycles} & \thmref{thm::td_pointwise} & $2\Lambda q - q$ & $2\Lambda q - 2q+1$\\
\hline
\makecell[l]{$A_1 = H_1(\hat{s})$ \\ $A_2 = \hat{p}_s\otimes \id$\\ $A_3 = H_2(\hat{s})$} & \makecell[l]{any product formula \\ with $q$-cycles} & \thmref{thm::td_pointwise} & $2\Lambda q - 2q+1$ & $2\Lambda q - 2q+1$ \\
\hline{}
\multirow{3}{*}{\makecell[l]{$A_1 = \hat{p}_s \otimes \id + H_1(\hat{s})$\\
 $A_2 = - \hat{p}_s \otimes \id$\\
 $A_3 = \hat{p}_s \otimes \id + H_2(\hat{s})$}} & first-order Trotter & \makecell[l]{Huyghebaert, \\ De Raedt \cite{huyghebaert_product_1990}} & $\Lambda$ & $\Lambda$\\
 \cline{2-5}
 & Strang splitting ($q=1$) & \makecell[l]{Huyghebaert, \\ De Raedt \cite{huyghebaert_product_1990}} & $2\Lambda-1$ & $2\Lambda -1$ \\
 \cline{2-5}
 & \makecell[l]{any product formula \\ with $q$-cycles} & \thmref{thm::product_td_v2} & $2\Lambda q-2q+1$ &  $2\Lambda q-2q+1$\\
\bottomrule
\end{longtable}

\begin{longtable}{p{0.22\textwidth}p{0.2\textwidth}p{0.2\textwidth}p{0.33\textwidth}}
\multicolumn{4}{c}{{\bf \SH{}'s Clock with Multi-Product Formula $\boldsymbol{\bigH = A_1 + A_2 + A_3}$}} \medskip \\
\toprule
Ordered Decomposition & Base & Time-independent MPF & Time-dependent MPF \\
\hline{}
\makecell[l]{$A_1 = \hat{p}_s\otimes \id$\\ $A_2 =  H_1(\hat{s})$ \\ $A_3 = H_2(\hat{s})$} 
 & \multirow{2}{*}{\makecell[l]{Strang splitting}} & \multirow{2}{*}{\cite{low_well_conditioned_2019}} & \makecell[l]{\thmref{thm::mpf_v1} (using continuous-clock) \\[1.5mm] (Conjectured in \cite{watkins_time_2024} using discrete clock) \\
 (cf. \cite{geiser_multi_product_2011} using Suzuki's time operator)}\\
\cline{1-1}\cline{4-4}
\makecell[l]{$A_1 = \hat{p}_s \otimes \id + H_1(\hat{s})$\\ $A_2 = - \hat{p}_s \otimes \id$\\
$A_3 = \hat{p}_s \otimes \id + H_2(\hat{s})$} & & &  see \thmref{thm::mpf_v2}\\
\bottomrule
\end{longtable}
\begin{longtable}{p{0.2\textwidth}p{0.2\textwidth}p{0.17\textwidth}p{0.15\textwidth}p{0.18\textwidth}}
\multicolumn{5}{c}{\makecell{{\bf\SH{}'s Clock with qDrift $\boldsymbol{\bigH = \sum_{k=1}^{\Lambda}  \int_{0}^{1}\dd r\ \mu(k,r) \Big(\hat{p}_s\otimes \id + \frac{H_k(\hat{s})}{\mu(k,r)}\Big)}$}\\ 
where $\mu$ is any probability measure for $(k,r)\in \{1, 2,\cdots,\Lambda\}\times[0,1]$}} \medskip\\
\toprule{}
Measures & Time-independent qDrift & Time-dependent qDrift & Error Scaling & Notes \\
\hline{}
\makecell[l]{$\mu(k,r) = \mu_k$\\ is independent of $r$} & \multirow{2}{*}{\cite{campbell_random_2019}} & see \eqref{eqn::cqDrift_v1} & see \lemref{lem::qdrift_err_1} & / \\
\cline{1-1}\cline{3-5}
general measure $\mu$ & & see \eqref{eqn::cqDrift_v2} & see \lemref{lem::qdrift_err_2} & \makecell[l]{recover ``continuous \\ qDrift'' \cite{berry_time_dependent_2020};\\
 see \lemref{lem::c_qDrift_equiv}}\\
\bottomrule
\end{longtable}
\end{table}

\newpage

\subsection{Summary of results}

We take \SH{}'s clock as our starting point, and systemically investigate how to leverage the advances in time-independent Hamiltonian simulation problems to tackle the time-dependent case.

\medskip
{\noindent{\bf Various computational models:}} We will demonstrate that this framework 
can be applied to analog quantum computing with detailed analysis on how to choose $\omega$. The primary focus of this work is to develop algorithms suitable for digital quantum computers, which we will elaborate on in details below.

\medskip
{\noindent {\bf Product formula:}} As mentioned above, there are two main approaches for directly handing time-dependence, \HDR{}'s approach (which utilizes the access to time-integration of local Hamiltonians) and Suzuki's approach (which utilizes access to pointwise local Hamiltonians). Specifically, \HDR{} approached the time-dependent problem by drawing on insights and natural generalizations; Suzuki introduced his renowned time-operator to address the problem of time-ordering. These two parallel approaches can be, in fact, unified within the framework of \SH{}'s clock, which offers a straightforward way to design product formulas for time-dependent case. In fact, Suzuki had briefly commented on the possibility of \SH{}'s clock approach in \cite[Sect.~6]{suzuki_general_1993}; unfortunately, it seems that this path has not yet been seriously explored so far.

\begin{itemize}

\item In this work, we show that \SH{}'s clock approach provides a {unified framework} that can lead to  both \HDR{}'s  and Suzuki's approaches, just via different decomposition of $\bigH$ in \eqref{eqn::bigH}.

\item By applying time-independent scheme directly to \SH{}'s clock and with appropriate decomposition, we generalize \HDR{}'s approach to arbitrary high-order, which was discussed still as an unknown problem in recent literatures \cite{ikeda_minimum_2023}. 

\item This approach also enables us to, from a different perspective, resolve the problem of finding quantum algorithms for time-dependent problems with minimum number of gates, as studied recently in  \cite{ikeda_minimum_2023}.

\item Since we need to introduce an unbounded operator $\hat{p}_s = -i \partial_s$ into a perhaps bounded Hamiltonian $H(s)$, it brings a concern about its mathematical soundness when $G$ in \eqref{eqn::G} is chosen as a delta measure. This concern actually can be straightforwardly answered by choosing $G \equiv 1$ and using \SH{}'s clock as an intermediate step rather than as an actual quantum protocol.

\end{itemize}

The top panel of Table~\ref{table::digital} gives  a visual summary of our main results, and see \secref{sec::pointwise} and \secref{section::HDR} for technical details. As a remark, we only list the case $\Lambda=2$ for simplicity of illustration, but main results are stated for a general $\Lambda \in \Natural$ later; the decomposition order of $\bigH$ in the table matters as these operators are non-commutative. 

\medskip
{\noindent{\bf Multi-product formula:}} The multi-product formula (MPF) for time-dependent case was recently studied in \cite{watkins_time_2024}, however, one important component is only stated as a conjecture, which we will explore using the continuous \SH{}'s clock. Due to the tight connection between Suzuki's formula and \HDR{}'s approach within the framework of \SH{}'s clock, we can immediately replace the base formula in MPF to \HDR{}'s second-order algorithm, and similarly to other higher-order base schemes. The middle panel of Table~\ref{table::digital} gives  a visual summary.

\medskip
{\noindent{\bf Randomized algorithm:}} We also investigate the resulting formula by applying the time-independent qDrift into the \SH{}'s clock formalism, and under suitable decompositions, we can recover the continuous qDrift proposed by \cite{berry_time_dependent_2020}. The benefit of \SH{}'s formalism enables us to analyze e.g., the bias from time-dependent qDrift almost in the same way as time-independent qDrift.
This is summarized in the bottom panel of Table~\ref{table::digital}.

\medskip
{\noindent{\bf Finite-order LCU-Taylor algorithm:}} We also explore how Taylor's expansion may be directly applied to \SH{}'s clock to derive quantum algorithms with finite order. 
Specifically, we will show how to directly derive a second-order LCU-Taylor algorithm using \SH{}'s clock, as a possible alternative candidate to Dyson's expansion.
Its generalization to other finite-order algorithms are straightforward and is not pursued herein. 
Although these algorithms do not achieve optimal scaling, they are relatively easier for understanding and have straightforward implementation. Moreover, such algorithms serve as a 
testament to the versatility of \SH{}'s clock in providing a systematic approach to developing quantum algorithms.

\medskip
{\noindent{\bf Numerical experiments:}} We perform numerical experiments on our newly derived  high-order \HDR{}'s algorithm. Leveraging recent advances in optimizing time-independent product formulas (see e.g., \cite{barthel_optimized_2020,ostmeyer_optimised_2023}), we can seamlessly extend these improvements to the time-dependent case, yielding more practically efficient product formulas for time-dependent simulations. We consider four different fourth-order decompositions summarized in \cite{ostmeyer_optimised_2023} for demonstration, and compare with a Magnus-motivated product formula \cite{ikeda_minimum_2023} equipped with the same time-independent scheme.  This new family of algorithms exhibits better performance for digitized adiabatic simulation. We remark that wider benchmark experiments are surely needed for further validation, but these experiments have already showcased how \SH{}'s clock can assist with the design of efficient quantum algorithms.

\medskip
{\noindent{\bf Remark:}} We acknowledge a related contemporary independent work \cite{mizuta_explicit_2024}, which employs Floquet theory to derive explicit error bounds for time-dependent product formulas and multi-product formulas. While there is a thematic connection between our work (particularly \secref{section::HDR} and \secref{sec::mpf}) and \cite{mizuta_explicit_2024}, we emphasize that the technical approaches are distinct. Unlike \cite{mizuta_explicit_2024}, our approach does not need to rely on the periodic embedding of time-dependent Hamiltonians. Instead, we present a more general and flexible methodology applicable to product formulas utilizing both pointwise-query and integral-based query models, as well as other types of quantum algorithms. On the contrary, \cite{mizuta_explicit_2024} provided an error bound with commutator scaling, a result not yet pursued in this paper. Consequently, while both studies draw inspiration from the Sambe-Howland's clock formalism, they are complementary, differing in their primary findings and focuses, as well as technical details.

\subsection{Notations}

We denote the Hilbert space of the original Hamiltonian simulation problem as $\hbt$.
Suppose $A_k$ are operators for $k = 1, 2, \cdots, \Lambda$ where $\Lambda$ is an arbitrary positive integer, then 
\begin{align*}
\prod_{k=1}^{\Lambda} e^{A_k} &:= e^{A_1} e^{A_2} \cdots e^{A_{\Lambda}},\qquad
\prod_{k=\Lambda}^{1} e^{A_k} := e^{A_{\Lambda}}\cdots  e^{A_2} e^{A_1}.
\end{align*}
As operators are generally non-commutative, these two notations are distinct.
We denote $\ket{\one}$ as the constant function with value one in the domain $\timedom$ for the $s$ mode. 
We remark that this state $\ket{\one}$ could be un-normalized, which is introduced merely for mathematical purposes. 
For any operator $A$ on a Hilbert space $\hbt$, the notation $\norm{A}_{p}$ means the Schatten-$p$ norm for $p \in [1,\infty]$; for any super-operator $\mathcal{A}$, the notation $\norm{\mathcal{A}}_{p}$ means the induced $p$-norm acting on the density matrices of $\hbt$, namely, $\norm{\mathcal{A}}_{p} := \sup_{\rho} \norm{\mathcal{A}(\rho)}_p$ where the supremum is taken over all possible density matrices $\rho$.
The notation $\mathcal{U}(t,s)$ denotes the target time-dependent unitary evolution; 
$U(t)$ denotes the time-independent unitary evolution;
${U}(t,s)$ denotes the time-dependent unitary evolution (used for approximation);
$\widetilde{U}(t,s)$ denotes the Huyghebaert and De Raedt’s unitary evolution (used for approximation);
and $e_{\mathcal{T}} \  \equiv\ \exp_{\mathcal{T}}$ denotes the time-ordered exponential.

\section{Preliminaries: computational models, applications, and more on \SH{}'s clock}

\subsection{Computational models}
\label{sec::models}

\SH{}'s clock provides a generic  approach to transform a time-dependent  problem into time-independent one. However, depending on how one chooses the localized state $G$ (namely, $\omega$ in \eqref{eqn::G}) and whether one intends to use auxiliary $G$ as a true physical ancilla mode or merely an artificial one  purely for mathematical purposes, one could develop different quantum algorithms suitable for various quantum architecture starting from \SH{}'s clock.

\subsubsection{The idealized case ($\omega\to 0$)}

For the idealized case where one  chooses the Dirac delta function (as a mathematical tool only), there are two equivalent ways to recover the original state:

\begin{itemize}[leftmargin=\leftmarginnew{}]
\item  {\bf (Localized version).} If the auxiliary state $G = \delta(\cdot-t)$ is the Dirac  delta function centered at $t$ and $Q\equiv 1$, then 
\begin{align}
\label{eqn::v1}
\bra{Q}\bra{\varphi} e^{-i \bigH t' } \ket{G}\ket{\psi} = \bra{\varphi} \exactU{t+t'}{t} \ket{\psi}.
\end{align}
It means one can prepare a localized state, and then trace out all the auxiliary degree of freedom to recover the desired state $\exactU{t+t'}{t} \ket{\psi}$.

\item {\bf (Delocalized version).} If the auxiliary state $G\equiv 1$ and $Q$ is the Dirac delta function centered at $t+t'$, then 
\begin{align}
\label{eqn::v2}
\begin{aligned}
\big(\bra{G}\bra{\psi} e^{i \bigH t'} \ket{Q}\ket{\varphi}\big)^* =& \bra{Q}\bra{\varphi} e^{-i \bigH t'} \ket{G}\ket{\psi} \\
=& \bra{\varphi}\exactU{t+t'}{t} \ket{\psi}.
\end{aligned}
\end{align}

In later applications, we will typically choose $t' = \dt$ to be small to develop product formulas for one time-step. It can be observed that these two versions \eqref{eqn::v1} and \eqref{eqn::v2} are essentially the \emph{conjugate pair} to each other.
\end{itemize}

\subsubsection{For analog quantum computing ($\omega>0$)}

For analog quantum computing where one can only perform measurement rather than direct integration, it is easy to  show that for the full dynamics \eqref{eqn::full_dynamics} on the augmented space, 
\begin{align}
\label{eqn::trace_Psi}
\begin{aligned}
& \rho_\omega(t)\ :=\ \tr_s\big(\ketbra{\Psi(t)}\big) \\
=& \int\ \dd s\ \abs{G(s-t)}^2\  \exactU{s}{s-t} \ketbra{\psi(0)} \exactU{s}{s-t}^{\dagger}.
\end{aligned}
\end{align}
Again, when $G$ is localized at $0$, then 
\begin{align*}
\tr_s\big(\ketbra{\Psi(t)}\big)  \approx \exactU{t}{0} \ketbra{\psi(0)} \exactU{t}{0}^{\dagger}.
\end{align*}
As was remarked in \cite{time_dilation_2023}, one  does not  necessarily need to prepare the auxiliary state $G$ as a pure state. Due to the linearity in the auxiliary state, it can be a mixed state whose diagonal components are localized at zero. For simplicity, we shall stick with the pure state formalism. Experimentally, one could approximately prepare a squeezed Gaussian state $\ket{G}$ with variance $\order{\omega^2}$ as in \eqref{eqn::G}.

\subsection{\DAScap}

Given two Hamiltonians $\h_1$ and $\h_2$, together with a schedule $f : [0,1]\to [0,1]$ satisfying $f(0) = 0$ and $f(1) = 1$, and the time scale $T$, one can control a wave function to slowly move from the ground state of $\h_1$ to the ground state of $\h_2$ \cite{albash_adiabatic_2018}:
\begin{align}
\label{eqn::AQC}
\left\{\ 
\begin{aligned}
\partial_t \ket{\psi(t)} &= -i H(t) \ket{\psi(t)};\\
H(t) &= H_1(t) + H_2(t),\ \\
H_1(t) &= f_1(t) \h_1, \qquad f_1(t) = T\big(1-f(t)\big),\\
H_2(t) &= f_2(t)\h_2, \qquad f_2(t) = Tf(t). 
\end{aligned}\right.
\end{align}
The initial state $\ket{\psi(0)}$ is chosen as the ground state of $\h_1$. Typically, one takes the time $T\gg 1$ to ensure accuracy. The choices of optimal schedule $f$ and $T$ are important to achieve high efficiency. In this work, as our goal is to simulate time-dependent Hamiltonians rather than optimizing its schedule path, we assume that these choices are already made. 

\begin{assume}
\label{A2}
Let us denote the exact ground state of $\h_2$ as $\ket{\exact}$.
Assume that the time $T$ and the schedule $f$ have been chosen such that
\begin{align}
\label{eqn::err_aqc}
\abs\big{\braket{\psi(1)}{\exact} - 1} &= \order{\eps}
\end{align}
or equivalently
\begin{align}
\norm\big{\ketbra{\psi(1)} - \ketbra{\varphi}}_{\tr} = \order{\sqrt{\eps}}.
\end{align}
\end{assume}

This assumption means that in follow-up applications, we won't optimize the schedule nor time span $T$ as in e.g., \cite{zeng_schedule_2016,an_quantum_2022}, but rather only focus on simulating such a dynamics using quantum architecture.

\subsection{The functional space of \SH{}'s clock}
\label{subsec::funct_space}

When one considers the \DAS{}, one could always restrict the time domain to $[0,1]$; therefore, one can always find a smooth extension of this Hamiltonian simulation problem to a torus with periodic boundary conditions on $H_k$. In case one considers unbounded time domain, one may simply choose $\timedom = \mathbb{R}$. In what follows, we shall only consider these two situations.
For the augmented Hamiltonian dynamics \eqref{eqn::full_dynamics}, we consider the following functional space 
\begin{align*}
\augspace := \Big\{f \in C^\infty(\timedom;\Complex): \norm{f^{(m)}}_{L^\infty(\timedom)} < \infty \Big\} \otimes \hbt,
\end{align*}
where $f^{(m)}$ means the $m^{\text{th}}$ order derivative of $f$.
For problems with lower regularity conditions, one could possibly resolve the regularity issues e.g., by first approximating the problem with smooth enough Hamiltonians. Dealing with such technicalities is beyond the scope of this work, and we shall exclusively focus on smooth functions throughout  this work. Thus, the above functional space $\augspace$ is sufficient and natural to work on.

\subsection{\SH{}'s clock and the perspective of numerical PDEs}

In what follows, we shall elaborate on the details in Table~\ref{table::summary} from the numerical PDE's perspective. 
Recall that \SH{}'s clock essentially transforms the time-dependent Hamiltonian dynamics into a PDE.
In particular, the first term in the augmented Hamiltonian $\hat{p}_s\otimes \id$ contributes as a transport term.
For instance, when $H(s) = -\frac{1}{2 M}\partial_{xx} + f(s) V(x)$ is a time-dependent particle Hamiltonian under the potential $V$ and with particle mass $M$, the above construction transforms an $s$-dependent elliptic-type operator $H(s)$ into a parabolic-type for $\bigH$, and the Schr{\"o}dinger's equation on the augmented space is 
\begin{align*}
\partial_{t} \Psi(t, s,x) =\ -\partial_s \Psi(t, s,x) + \frac{i}{2M} \partial_{xx} \Psi(t, s,x) + (-i) f(s) V(x) \Psi(t, s,x).
\end{align*}
As one is dealing with PDEs, the most natural approaches to consider are finite difference method, spectral method, and operator-splitting. Of course, the finite-element method is also possible to consider and it might be useful for developing other types of quantum algorithms; however, we will not consider it for the time being.
Firstly, we will discuss a challenge of using finite-difference method and discuss why we should consider the limiting situation $\omega \to 0$ or $\omega \to \infty$ for digital quantum computers. As for spectral methods and operator-splitting methods, we will discuss their advantages and limitations; the first one was studied in \cite{mizuta_optimal_2023,mizuta_optimal_multiperiodic_2023} for periodic Hamiltonians, and we will study operator-splitting methods later as well as many variants.

\bigskip
{\noindent \bf Finite difference method for $0 < \omega \ll 1$:}
When we use a Gaussian state with small but non-zero $\omega$ as an approximate in \eqref{eqn::G}, then the most natural approach is the finite-difference method. 
A challenge (also mentioned in \cite{watkins_time_2024}) is that as the momentum operator has very large norm $\order{\frac{1}{\ds}}$ after discretization (where $\ds$ is the mesh size in the $s$ mode), it appears to be difficult to go beyond the first-order algorithm in terms of gate complexity, even using theoretically optimal algorithm like Qubitization.

\bigskip
{\noindent \bf Spectral method for $\omega\to \infty$:} In \cite{mizuta_optimal_2023,mizuta_optimal_multiperiodic_2023}, the spectral method was employed using the Fourier basis with periodic boundary condition. One can discretize the $s$ variable using the Fourier basis and project the whole dynamics into a finite number of basis functions. This approach is beneficial if the error of spectral method for the Hamiltonian dynamics \eqref{eqn::full_dynamics} decays exponentially fast with respect to the number of basis functions chosen. This is the case when $H(t)$ is sufficiently smooth.  As discussed above, choosing such a Fourier basis will very likely necessitate the projection into a localized state (namely using the delocalized version), and therefore, the amplitude amplification methods  are necessary and perhaps unavoidable as were used in \cite{mizuta_optimal_2023}. This suggests that quantum algorithms built upon this approach is unlikely to be implementable without the fault-tolerance quantum computer. 

\bigskip
{\noindent \bf Operator-splitting methods as $\omega$ approaches $0$ or $\infty$:} To develop practical algorithms from the \SH{}'s formalism, another natural way is to develop product-based formula (namely, operator-splitting) for such a formalism. It turns out that for both $\omega = 0$ and $\infty$, one will end up with the same final time-dependent product formula. The case $\omega = 0$ is more physically intuitive in terms of the derivation as the $s$ mode in this case aligns better with the \enquote{time} variable, whereas the case $\omega = \infty$ can more easily provide mathematical convenience. However, we remark that both formalisms are technically equivalent.

\bigskip
{\noindent \bf Summary:} From the above non-exhaustive discussion, it appears that to make the \SH{}'s clock formalism suitable for digital quantum platform in practice, it might be desirable that we  choose $\omega \to 0$ and $\omega \to \infty$ in order to avoid the order restriction, and also choose operator-splitting or many variants to avoid the possible implementation challenge. This leads into  \secref{sec::pointwise}, \secref{section::HDR}, \secref{sec::mpf}, and \secref{sec::qDrift} below.

\section{Analog quantum computing}
\label{sec::aqc}

When the Hamiltonian itself is inherently defined for qumodes (namely, continuous-variable quantum state), the most straightforward way is to adopt \emph{analog quantum computing} \cite{kendon_quantum_2010,jin_analog_2024, time_dilation_2023} to realize the Hamiltonian simulation. In \cite{time_dilation_2023}, the continuous form of the Sambe-Howland model was used to turn non-autonomous linear PDEs into time-independent Hamiltonian dynamics which can be simulated on an analog quantum device.
In this section, 
we will first review this formalism and discuss how to achieve high-order approximation with respect to the Gaussian width $\omega>0$ for general problems.
For \DAS{}, as the computational model is fully analog, we will also discuss how to pick $\omega$ for the squeezed Gaussian state.

\subsection{High-order approximation via Richardson's extrapolation}

Recall that the \SH{}'s clock Hamiltonian is given as $\bigH = \hat{p}_s \otimes \unit + H(\hat{s})$. One evolves the closed Hamiltonian dynamics
\begin{align*}
& \partial_t \sigma(t) = -i \comm\Big{\bigH}{\sigma(t)},\\
& \sigma(0) = \rho_0 \otimes \ketbra{\psi(0)},\  \rho_0 = \ketbra{G},
\end{align*}
where $G$ is chosen as the squeezed Gaussian state \eqref{eqn::G}.
It has been shown in \cite{time_dilation_2023} that 
if $\omega \ll 1$, one has:
\begin{align*}
\norm\Big{\ \ketbra{\psi(t)} - \rho_\omega(t) \ }_{\tr} \le \sqrt{C}\omega,
\end{align*}
where $C$ is some constant.
This estimate yields a first-order algorithm. The scaling of $\omega$ can be improved via e.g., Richardson's extrapolation if one mainly focuses on the estimates of an observable $\widehat{O}$ on $\hbt$. This technique has been used for error mitigation \cite{endo_mitigating_2019} and multi-product formulas \cite{chin_multi_product_2011,childs_hamiltonian_2012,low_well_conditioned_2019} for Hamiltonian simulation.

\begin{theorem}
\label{thm::mpf_analog}
Recall $\rho_\omega$ from \eqref{eqn::trace_Psi} and suppose $G$ is given in \eqref{eqn::G}. For any observable $\widehat{O}$ with $\norm{\widehat{O} }_{\infty} \le 1$, and suppose
coefficients $\{\alpha_j\}_{j=1}^{M} \subseteq \Real$ and integer-valued time steps $\{k_j\}_{j=1}^{M} \subseteq \Natural$ are well chosen to satisfy 
\begin{align}
\label{eqn::mpf_order}
\begin{bmatrix}
1 & 1 & \cdots & 1 \\
k_1^{-2} & k_2^{-2} & \cdots & k_M^{-2} \\
\vdots & \vdots & \vdots & \vdots\\
k_{1}^{-2(m-1)} & k_{2}^{-2(m-1)} & \cdots & k_{M}^{-2(m-1)} 
\end{bmatrix}
\begin{bmatrix} \alpha_1\\ \alpha_2\\ \vdots\\ \alpha_M \end{bmatrix}
= \begin{bmatrix} 1\\ 0\\ \vdots\\ 0 \end{bmatrix}.
\end{align}
Then
\begin{align*}
&\ \abs\Big{\ \sum_{j=1}^{M} \alpha_j \tr\big(\widehat{O} \rho_{\omega/k_j}(t)\big) - \bra{\psi(t)}\widehat{O} \ket{\psi(t)}\ } 
\le \  2\sqrt{C_m} \omega^m + \order{\omega^{m+2}},
\end{align*}
where $C_m$ is a constant whose expression can be found in \eqref{eqn::Cm}.
\end{theorem}
The proof is postponed to \appref{proof::thm::mpf_analog}.

\subsection{Gaussian width estimate for \DAS{}}

Next we will specialize the above error estimation for \DAS{}.
For \DAS \, problem in \eqref{eqn::AQC}, the Hamiltonian takes the following more specific form:
\begin{align}
\begin{aligned}
\bigH &= \hat{p}_s \otimes \unit + f_1(\hat{s}) \otimes \h_1 + f_2(\hat{s}) \otimes \h_2.
\end{aligned}
\end{align}

Under Assumption \ref{A2}, the adiabatic path gives the error scaling $\order{\sqrt{\eps}}$ in trace norm in \eqref{eqn::err_aqc}. 
In the following, we provide an estimate of the choice of $\omega$ in order to maintain the same scaling of error.

\begin{theorem}[Choice of $\omega$ for fully analog \DAS]
\label{thm::omega_analog}
Suppose that $\eps\ll 1$ and $\omega \ll 1$. 
By further imposing the constraint
\begin{align}
\label{eqn::omega}
\omega = \order{\frac{1}{T \norm{\h_2}}},
\end{align}
then one can ensure that 
\begin{align}
\norm\big{\tr_s\big(\ketbra{\Psi(t)}\big) - \ketbra{\exact} }_{\tr} = \order{\sqrt{\eps}},
\end{align}
which has the same asymptotic error scaling as \DAS\, formalism \eqref{eqn::err_aqc}.
\end{theorem}

The proof is postponed to  \appref{proof::aqc}. We remark that the scaling of $T$ (implicitly) depends on the precision $\eps$, and the precise dependence is determined by the \DAS{} problem itself and the schedule $f$ chosen.

\begin{remark}
This formalism is beneficial if $\norm{\h_2} = \order{1}$ as in the case of the adiabatic Grover's search algorithm; see \secref{sec::grover} below. The above scaling also linearly depends on $T^{-1}$, which is very likely to be optimal, as the Hamiltonian in the augmented system $\norm{\bigH}= \order{T}$.
\end{remark}

\section{Product formula (\rom{1}): point-wise query to $H_k(s)$}
\label{sec::pointwise}

Product formula is arguably the most widely studied and practical algorithm in Hamiltonian simulation and we will first discuss how \SH{}'s clock can be used to develop time-dependent product formulas.
This perspective could date back to a comment in Sect.~6 of Suzuki's celebrated work \cite{suzuki_general_1993}, but as far as we know, it has not been formally adopted and studied in details in literature. Suzuki's $\nicefrac{\overset{\leftarrow}{\partial}}{\partial t}$ operator in \cite{suzuki_general_1993} is not rigorously specified in his original work, and \SH{}'s clock Hamiltonian precisely provides the right mathematical object.

We will first illustrate the idea of \SH{}'s clock for simple examples in \secref{sec::pointwise::example} for first-order and second-order cases, and then generalize this idea to arbitrary high-order algorithms in \secref{sec::pointwise::general}, given any time-independent product formula. Then we will demonstrate and discuss in \secref{sec::pointwise::gate} that using this approach, the time-dependent product formulas {could} only use the same number of gates as the corresponding time-independent counterparts.
In particular, this suggests a route to design a time-dependent scheme with minimum number of gates for arbitrary high-order, which generalizes \cite{ikeda_minimum_2023}. 
We will explain its connection with Suzuki's formalism in \secref{sec::pointwise::suzuki}. Concerning the mathematical rigorousness, we shall comment on why \SH{}'s clock offers the right tool in \secref{sec::pointwise::math} without the need to deal with Gaussian approximation of delta measure as was probably needed in discrete clock analog \cite{watkins_time_2024}.

\subsection{Illustrative examples: first-order and second-order}
\label{sec::pointwise::example}

We begin with heuristic  explanation on  how  \SH{}'s clock enables one to develop product formulas for time-dependent Hamiltonian simulation. The augmented Hamiltonian takes the following form:
\begin{align*}
\bigH = \hat{p}_s\otimes \id + \sum_{k=1}^{\Lambda} H_k(\hat{s}).
\end{align*}
This operator is time-independent, and one can apply the following Lie's formula (first-order product formula),
$
e^{-i \dt \bigH} \approx  e^{-i \dt \hat{p}_s\otimes \id}\ \prod_{k=1}^{\Lambda} e^{-i\dt H_k(\hat{s})} + \order{\dt^2}.
$
Therefore,
\begin{align*}
 &\ e^{-i \dt \bigH} \ket{t} \otimes \ket{\psi(t)} \\
\approx &\  e^{-i \dt\ \hat{p}_s\otimes \id}\ \prod_{k=1}^{\Lambda} e^{-i\dt H_k(\hat{s})} \ket{t} \otimes \ket{\psi(t)} + \order{\dt^2}\\
\approx &\  e^{-i \dt\ \hat{p}_s}\ \ket{t}\ \otimes \prod_{k=1}^{\Lambda} e^{-i\dt H_k(t)} \ket{\psi(t)} +\order{\dt^2}\\
= &\ \ket{t+\dt} \otimes \prod_{k=1}^{\Lambda} e^{-i\dt H_k(t)} \ket{\psi(t)}.
\end{align*}
By the time-localized version \eqref{eqn::v1}, one knows that $\bra{\one}_s e^{-i \dt \bigH} \ket{t} \otimes \ket{\psi(t)} = \exactU{t+\dt}{t} \ket{\psi(t)}$.
This implies that one can approximate the time-dependent Hamiltonian evolution via
\begin{align*}
\exactU{t+\dt}{t} \approx \prod_{k=1}^{\Lambda} e^{-i\dt H_k(t)} + \order{\dt^2}.
\end{align*}
To get this, we have used the fact that $\bra{\one} \ket{t+\dt} \equiv \int_{\timedom} 1 \cdot \delta\big(s-(t+\dt)\big)\ \dd s  = 1$.
This formula is essentially the analog of Euler's method for ODEs.
If we apply the Strang splitting scheme (namely, a second-order Trotter, or also known as the midpoint scheme), then similarly,
\begin{align*}
&\ \bra{\one}_s e^{-i \dt \bigH} \ket{t} \otimes \ket{\psi(t)} \\
\approx&\ 
\bra{\one}_s e^{-i\hat{p}_s \otimes \id\ \dt/2} 
\prod_{k=1}^{\Lambda} e^{-i \frac{\dt}{2} H_k(\hat{s})} \prod_{k=\Lambda}^{1} e^{-i \frac{\dt}{2} H_k(\hat{s})}
e^{-i\hat{p}_s\otimes \id\ \dt/2}
\ket{t} \otimes \ket{\psi(t)}\\
=& \prod_{k=1}^{\Lambda} e^{-i \frac{\dt}{2} H_k(t+\dt/2)} \prod_{k=\Lambda}^{1} e^{-i \frac{\dt}{2} H_k(t+\dt/2)} \ket{\psi(t)}.
\end{align*}
This suggests the following approximation:
\begin{align}
\exactU{t+\dt}{t}\approx\prod_{k=1}^{\Lambda} e^{-i \frac{\dt}{2} H_k(t+\dt/2)} \prod_{k=\Lambda}^{1} e^{-i \frac{\dt}{2} H_k(t+\dt/2)}.
\end{align}
This scheme is also known as the second-order Trotter for the time-dependent Hamiltonian, and in some literature, it is known as the midpoint scheme. This scheme also has the minimum number of operator exponentials when $\Lambda = 2$ for the second-order product formulas; see e.g., \cite{ikeda_minimum_2023}.
When the Hamiltonian is time-independent, one recovers the midpoint scheme for time-independent Hamiltonians:
\begin{align}
\label{eqn::midpoint}
e^{-i  H \dt} \approx \prod_{k=1}^{\Lambda} e^{-i \frac{\dt}{2} H_k} \prod_{k=\Lambda}^{1} e^{-i \frac{\dt}{2} H_k}.
\end{align}

\subsection{A general scheme}
\label{sec::pointwise::general}

In a recent work, Ostmeyer \cite{ostmeyer_optimised_2023} proved that a product formula for two-local operators with order $n$ can yield a general scheme with the same order for $\Lambda$ local operators. The conclusion was stated for time-independent operator splitting methods. 

\begin{theorem}[{\cite[Theorem 1]{ostmeyer_optimised_2023}}]
\label{thm::ostmeyer}
Suppose there is a product formula with order $n$ such that for arbitrary two operators $A$ and $B$,
\begin{align}
\label{eqn::product_2_body}
\begin{aligned}
e^{(A+B) \dt} =& e^{A a_1 \dt} e^{B b_1 \dt} \cdots e^{B b_q \dt} e^{A a_{q+1} \dt} + \order{\dt^{n+1}},
\end{aligned}
\end{align}
then for any $\Lambda\ge 1$ and a collection of operators $\{A_k\}_{k=1}^{\Lambda}$, 
\begin{align}
\label{eqn::ostmeyer}
\begin{aligned}
&\ e^{\dt\sum_{k=1}^{\Lambda} A_k + \order{\dt^{n+1}}} = \Big(\prod_{k=1}^{\Lambda} e^{A_k c_1 \dt}\Big)\Big(\prod_{k=\Lambda}^{1} e^{A_k d_1 \dt}\Big) \cdots \Big(\prod_{k=1}^{\Lambda} e^{A_k c_q \dt}\Big)\Big(\prod_{k=\Lambda}^{1} e^{A_k d_q \dt}\Big),
\end{aligned}
\end{align}
where
\begin{equation*}
 \begin{cases}
& c_1 = a_1,\ \  c_k = a_k - d_{k-1}, \qquad k = 2, \cdots, q\\
& d_k = b_k - c_k, \qquad  k = 1, 2, \cdots, q.\\
\end{cases}
\end{equation*}
\end{theorem}

In what follows, we generalize such a result to the time-dependent case, \emph{without increasing the number of operator exponentials} in the product formula and \emph{without affecting the order of convergence}.

\begin{theorem}
\label{thm::td_pointwise}
Suppose that the time-independent Hamiltonian can be decomposed as in \eqref{eqn::decompose}. Under the same assumption of Theorem \ref{thm::ostmeyer},
the following scheme gives an order $n$ approximation of the time-ordered unitary evolution $\exactU{t+\dt}{t} \equiv \exp_\mathcal{T}\big(-i \int_{t}^{t+\dt} H(s)\ \dd s\big)$:
\begin{align}
\label{eqn::product_td}
\begin{aligned}
{U}_n(t+\dt, t) := &\  U_{F}\big(t+L_1, t+R_1\big) U_{B}\big(t+R_1, t+L_{2}\big) \cdots \\
&\ \  \cdots U_{F}\big(t+L_{q-1}, t+R_{q-1}\big) U_{B}\big(t+R_{q-1}, t+L_q\big) \\
&\ \  U_{F}\big(t+ L_{q}, t+R_q\big) U_{B}\big(t+R_q, t + 0\big),
\end{aligned}
\end{align}
where
\begin{align*}
\left\{
\begin{aligned}
U_{F}\big(t', t\big) &= \prod_{k=1}^{\Lambda'} e^{-i (t'-t)H_k(t')} \prod_{k=\Lambda'+1}^{\Lambda} e^{-i (t'-t) H_k(t)}, \\
U_{B}\big(t', t\big) &=  \prod_{k=\Lambda}^{\Lambda'+1} e^{-i (t'-t) H_k(t')} \prod_{k=\Lambda'}^{1} e^{-i (t'-t) H_k(t)},
\end{aligned}\right.
\end{align*}
and $\Lambda' \in \{0, 1, 2, \cdots, \Lambda\}$ is arbitrary, with the  coefficients satisfying  the following relations:
for $k = 1, 2, \cdots, q$, 
\begin{equation*}
 \begin{cases}
 & L_k = \dt \sum_{j=k}^{q} (c_j + d_j), \qquad L_{q+1}\equiv 0\\
 & R_k = \dt \Big(d_{k} + \sum_{j=k+1}^{q} (c_j + d_j)\Big).
 \end{cases}
\end{equation*}

Suppose one is given query access to matrix exponentials of the form $e^{-i\dt H_k(\alpha)}$ for any index $k$, time $\alpha$, and time step $\dt$. Then  the following resources are needed:
\begin{table}[h!]
\centering
\begin{tabular}{b{0.3\textwidth}b{0.2\textwidth}b{0.2\textwidth}}
\toprule
Method & Parameter & Number of Gates \\
\toprule
Time-independent case & n/a & {$\boldsymbol{2\Lambda q - (2q-1)}$ }\\
Time-dependent case & $\Lambda' \neq 0$ and $\Lambda' \neq \Lambda$ &  {$\boldsymbol{2\Lambda q - (2q-1)}$} \\
 & $\Lambda' = 0$ & $2 \Lambda q - q$ \\
 & $\Lambda' = \Lambda$ & $2\Lambda q - (q-1)$\\
\bottomrule
\end{tabular}
\caption{This table lists the number of (local) gates for time-dependent schemes for various $\Lambda'$. In this estimate, we consider the most general form, and there is no need to assume $\hat{H}_k(s)$ to be commutative in time (namely, the above estimate is valid even if $\comm{H_k(s)}{H_k(t)}\neq 0$ for $s \neq t$). When we consider the \DAS{}, the number of gates is \emph{always} $2\Lambda q - (2q-1)$ for any such time-dependent schemes, no matter how to choose $\Lambda'$.}
\end{table}
\end{theorem}

\begin{proof} We first present a proof based on the localized version \eqref{eqn::v1} to illustrate the ideas. A second proof which is more friendly for mathematical rigorousness will be given in later subsections. For the augmented Hamiltonian, we can choose
\begin{align}
\label{eqn::Ak}
\left\{\ 
\begin{aligned}
A_k &= -i H_k(\hat{s}), \qquad k \le \Lambda', \\
A_{\Lambda'+1} &= -i \hat{p}_s\otimes \id,\\
A_{k} &= -i H_{k-1}(\hat{s}), \qquad k > \Lambda'+1.
\end{aligned}\right.
\end{align}
Then by the localized version \eqref{eqn::v1} and by \eqref{eqn::product_2_body},
\begin{align*}
\begin{aligned}
&\  U(t+\dt,t) \ket{\psi}\\
\myeq{\eqref{eqn::v1}}&\ \bra{\one}_s e^{-i \dt \bigH} \ket{t}\otimes \ket{\psi} \\
\myeq{\eqref{eqn::product_2_body}}&\ \bra{\one}_s \Big(\prod_{k=1}^{\Lambda+1} e^{A_k c_1 \dt}\Big)\Big(\prod_{k=\Lambda+1}^{1} e^{A_k d_1 \dt}\Big) \cdots \\
&\ \cdots \Big(\prod_{k=1}^{\Lambda+1} e^{A_k c_q \dt}\Big)\Big(\prod_{k=\Lambda+1}^{1} e^{A_k d_q \dt}\Big) \ket{t}\otimes \ket{\psi} \\
&\qquad + \order{\dt^{n+1}} .
\end{aligned}
\end{align*}
For instance, one can easily verify that
\begin{align*}
&\  \prod_{k=\Lambda+1}^{1} e^{A_k d_q \dt} \ket{t}\otimes \ket{\psi}  \\
=&\   \prod_{k=\Lambda+1}^{\Lambda'+2} e^{A_k d_q\dt} e^{A_{\Lambda'+1}\dt }\prod_{k=\Lambda'}^{1} e^{A_k d_q\dt}  \ket{t}\otimes \ket{\psi}\\
=&\ \prod_{k=\Lambda+1}^{\Lambda'+2} e^{A_k d_q\dt} e^{A_{\Lambda'+1}\dt } \Big(\ket{t}\otimes \prod_{k=\Lambda'}^{1} e^{-i H_k(t) d_q\dt}  \ket{\psi}\Big) \\
=&\ \prod_{k=\Lambda+1}^{\Lambda'+2} e^{A_k d_q\dt} \Big(\ket{t+d_q\dt}\otimes  \prod_{k=\Lambda'}^{1}  e^{-i H_k(t) d_q\dt}  \ket{\psi}\Big) \\
=&\ \ket{t+d_q\dt} \otimes U_{B}\big(t+d_q\dt, t\big) \ket{\psi}.
\end{align*}
Consequently, after applying all previous $2q-1$ terms and finally applying $\bra{\one}_s$, we can straightforwardly verify the formula \eqref{eqn::product_td}. 

After going through all multiplicative construction of operator exponential, there seems to be $2 \Lambda q$ gates in \eqref{eqn::product_td}. In fact, one can remove some overcounted gates from $H_1$ and $H_{\Lambda}$, e.g., one can merge operators with the form $e^{-i c_1 H_{\Lambda}(t)} e^{-i d_1 H_{\Lambda}(t)}$  into a single operator exponential $e^{-i (c_1 + d_1) H_{\Lambda}(t)}$.
\begin{itemize}[leftmargin=\leftmarginnew{}]
\item When $\Lambda' = 0$, the gate corresponding to $H_{\Lambda}$ are overcounted but the gates for $H_1$ cannot be combined, which leads into the estimate $2 \Lambda q - (q-1)$. This can be more explicitly seen by the following: consider
\begin{align*}
&\ U_F(t + L_k, t + R_k) U_B(t+R_k, t + L_{k+1}) \\
=&\ \prod_{k=1}^{\Lambda} e^{-i (L_k - R_k) H_k(t+R_k)} \prod_{k=\Lambda}^{1} e^{-i (R_k - L_{k+1}) H_k(t + R_k)},
\end{align*}
and therefore, the gate for $H_\Lambda$ is overcounted; however, one cannot merge gates for $H_1$ if $H_1(t+R_k)$ don't commute with $H_1(t+R_{k+1})$ for the following
\begin{align*}
&\ U_B(t+R_k, t + L_{k+1}) U_F(t+L_{k+1}, t+R_{k+1}) \\
=&\ \prod_{k=\Lambda}^{1} e^{-i (R_k - L_{k+1}) H_k(t + R_k)} \prod_{k=1}^{\Lambda} e^{-i (L_{k+1} - R_{k+1}) H_k(t+R_{k+1})}.
\end{align*}
Overall, for the most general case, the number of gates is $2 \Lambda q - q$, whereas for \DAS{}, one only need $2\Lambda q - (2q-1)$ gates.
\item When $\Lambda' \neq 0$ and $\Lambda' \neq \Lambda$, at each junction of $U_F$ and $U_B$ operators, we can similarly verify that there is one overcounted gate and this leads into the estimate $2 \Lambda q - (2q-1)$. 
\item When $\Lambda' = \Lambda$, the gate estimate is similar to the above cases.
\end{itemize}
\end{proof}

\subsection{Discussion on the minimum-gate implementation}
\label{sec::pointwise::gate}
The above \thmref{thm::td_pointwise} suggests that as long as we pick $\Lambda' \neq 0, \Lambda$, then the time-dependent product formula above always use the same number of gates as the time-independent product formula under a very general condition. Suppose that we already have a time-independent product formula with the minimum number of gates for a certain fixed order (where the order can be arbitrary), then the time-dependent formula given by  \thmref{thm::td_pointwise} must also achieve the minimal number of gates. This generalizes the result in \cite{ikeda_minimum_2023}. Moreover, their algorithm in \cite{ikeda_minimum_2023} is not directly applicable for Hamiltonians with more than $2$ local terms; a similar algorithm is only stated in the context of qubit systems in \cite{chen_quantum_2023} (though not aiming at achieving the minimal number of gates). The above scheme \eqref{eqn::product_td}, however, is capable of handing a general $\Lambda$.

{For example}, for fourth-order product formulas, the Forest-Ruth-Suzuki formula \cite{forest_fourth-order_1990,suzuki_fractal_1990,ostmeyer_optimised_2023} achieves the minimum number of gates:
\begin{align}
\label{eqn::FRS}
e^{X+Y} \approx e^{\frac{\gamma}{2} X} e^{\gamma Y} e^{\frac{1-\gamma}{2} X} e^{(1-2\gamma) Y} e^{\frac{1-\gamma}{2} X} e^{\gamma Y} e^{\frac{\gamma}{2} X},
\end{align}
where $\gamma = (2-2^{1/3})^{-1} \approx 1.351207$. This product formula has the minimum number of gates for two local operators, which makes it suitable for \DAS.
\cite{ikeda_minimum_2023} used this formula to derive a product formula for time-dependent Hamiltonian simulation based on Magnus expansion and it admits the following form:
\begin{align}
\label{eqn::FRS_Magnus}
\begin{aligned}
&\ \mathcal{T}e^{-i\int_{t}^{t+\dt} H_1(s) + H_2(s)\ \dd s}\\
\approx &\ 
e^{-i (\frac{\gamma\beta_1}{2}+u) H_1} 
e^{-i\gamma\beta_2 H_2} 
e^{-i\frac{1-\gamma}{2} \beta_1 H_1} 
e^{-i(1-2\gamma) \beta_2 H_2} \\
& \qquad \times e^{-i\frac{1-\gamma}{2} \beta_1 H_1} 
e^{-i\gamma\beta_2 H_2} 
e^{-i (\frac{\gamma\beta_1}{2}-u) H_1}
\end{aligned}
\end{align}
where 
\begin{align*}
\beta_1 &= \int_{t}^{t+\dt} f_1(s)\dd s,\\
\beta_2 &= \int_{t}^{t+\dt} f_2(s)\dd s, \\
u &= \frac{1}{2 \beta_2} \int_{t}^{t+\dt}\dd s_1 \int_{t}^{s_1}\dd s_2\ f_1(s_1) f_2(s_2) - f_2(s_1) f_1(s_2).
\end{align*}
By \eqref{eqn::product_td} and when $\Lambda' = 1$, the time-dependent product formula is
\begin{align}
\label{eqn::FRS_new}
\begin{aligned}
&\ \mathcal{T}e^{-i\int_{t}^{t+\dt} H_1(s) + H_2(s)\ \dd s} \\
\approx &\  e^{-i \frac{\gamma\dt}{2} \hat{H}_1(t+\dt)} e^{-i \gamma\dt \hat{H}_2(t+(1-\frac{\gamma}{2})\dt)} \\
&\ \ e^{-i \frac{(1-\gamma)\dt}{2} \hat{H}_1(t+(1-\gamma)\dt)}
 e^{-i (1-2\gamma)\dt \hat{H}_2(t+\frac{\dt}{2})} \\
&\ \ e^{-i \frac{(1-\gamma)\dt}{2} \hat{H}_1(t+\gamma\dt)} e^{-i \gamma\dt \hat{H}_2(t+\frac{\gamma\dt}{2})} \\
&\ \ e^{-i \frac{\gamma\dt}{2} \hat{H}_1(t)}.
 \end{aligned}
\end{align}
This formula has the same number of gates compared with \eqref{eqn::FRS_Magnus} derived from \cite{ikeda_minimum_2023}, which means it also has the minimum possible gates for a time-dependent Hamiltonian simulation problem. 

\subsection{Connection to Suzuki's approach}
\label{sec::pointwise::suzuki}

In what follows, we shall discuss how the above newly developed scheme \eqref{eqn::product_td} connects to Suzuki's original approach and a few potential advantage of the approach that we take here.\\

{\noindent\emph{Recovery of Suzuki's formula.}} When one takes $\Lambda' = 0$ and considers the case $c_j = d_j$ for $1\le j \le q$, 
one knows that for the cycle $j$ in \eqref{eqn::product_td}, one has $L_j - R_j  = c_j \dt = d_j \dt = R_j - L_{j+1}$, and the building block in \eqref{eqn::product_td} for the $j^{\text{th}}$ cycle:
\begin{align*}
& U_F(t+ L_j, t + R_j) U_B(t+R_j, t+L_{j+1}) \\
=& \prod_{k=1}^{\Lambda} e^{-i c_k\dt {H}_k(t+R_j)} \prod_{k=\Lambda}^{1} e^{-i d_k \dt H_k(t+R_j)}
\end{align*}
is essentially the midpoint scheme in Eq.~(24) in \cite{suzuki_general_1993}, or Eq.~(9) in  \cite{wiebe_higher_2010} (an analytical work for Suzuki's time-dependent scheme), and the final overall scheme is basically the time-dependent formula by Suzuki.\\

\noindent{\emph{Possible advantage of \SH{}'s clock.}}
This approach is mathematically similar to Suzuki's approach whereas the approach from \SH{}'s clock {offers a new perspective to view Suzuki's time operator with a more explicit mathematical structure, as well as possible physical interpretations. More specifically,} (i) it avoids the necessity to handle or justify the operator $ \nicefrac{\overset{\leftarrow}{\partial}}{\partial t}$ \cite{suzuki_general_1993} whose mathematical justification is not well discussed; {if we view Suzuki's approach in the language of Sambe-Howland's clock, the functional space is more clear as explained in Section~\ref{subsec::funct_space} and Suzuki's time operator is essentially the momentum operator acting on the augmented space.} (ii) \SH{}'s clock Hamiltonian provides a physical interpretation for Suzuki's formulas, and in particular, $\nicefrac{\overset{\leftarrow}{\partial}}{\partial t}$ can be exactly interpreted as the extra auxiliary state as a \enquote{clock}. We  acknowledge that his original operator exactly plays the role as a physical clock but it may not be as explicitly stated as \SH{}'s clock.

\subsection{On the mathematical concerns and the second proof}
\label{sec::pointwise::math}

\subsubsection{Discussion on the mathematical concern}

The above derivation of product formulas uses the transport equation of a delta measure $\ket{t} = \delta(s-t)$, which is mathematically well-defined for the space of generalized distributions. 
However, a reasonable concern is that the operator $\hat{p}_s$ is not well-defined when acting on a delta measure $\ket{t}$, which leads into a concern that when we estimate the error or perform Taylor's expansion of $e^{-i \hat{p}_s \dt} = \sum_{k=0}^{\infty} \nicefrac{\big(-i \hat{p}_s \dt\big)^k}{k!}$, the series is not well-defined when  acting on $\ket{t}\ket{\psi}$ where $\ket{\psi}\in \hbt$. This perspective was discussed for the discrete clock case recently in \cite[Sect.~IV]{watkins_time_2024}. However, the delocalized formalism in \eqref{eqn::v2}, the conjugate formalism of the localized one, can exactly avoid this concern.

The reason is the following: for a smooth function $\ket{\Psi}\in \augspace$, $\hat{p}_s\otimes \id$ is surely well-defined acting on $\ket{\Psi}$ and thus the translation operator $e^{-i \dt \hat{p}_s\otimes \id}$ is also well-defined; the inner product $\bra{t+\dt}_s \ket{\Psi} \equiv \Psi(t+\dt, \cdot)$ is also bounded for $\ket{\Psi}\in \augspace$, and we don't need to interpret $\bra{t+\dt}$ as a delta measure but equivalently as retrieving the function value of a smooth function on the augmented space. Therefore, it becomes clear that the delocalized version could be made mathematically rigorous as long as the input state $\ket{\Psi}$ is smooth with respect to the $s$ variable.

In what follows, we will provide a second proof of \thmref{thm::td_pointwise}, showing that the delocalized formalism \eqref{eqn::v2} gives the same approximation (and thus automatically the same error term) as the localized formalism \eqref{eqn::v1}. This suggests that both approaches, the localized version \eqref{eqn::v1} and the delocalized version \eqref{eqn::v2}, could be made mathematically rigorous with appropriate interpretations.

\subsubsection{The second proof of \thmref{thm::td_pointwise}}
Let us consider the delocalized version \eqref{eqn::v2}. 
Before getting into the detailed proofs, we shall list some facts in \lemref{lem::product_1}, whose proof is postponed to \appref{subsec::lem::product_1}.

\begin{lemma}
\label{lem::product_1}
If one adopts the same notation as in \eqref{eqn::Ak}, then  for any $s'$, $\theta\in\Real$, $\ket{\varphi}\in \hbt$ and any $\ket{\Psi}\in \augspace$, one has
\begin{align*}
&\ \bra{s'} \bra{\varphi}\prod_{k=1}^{\Lambda+1} e^{A_k \theta \dt}\ket{\Psi} 
= \bra{s'-\theta\dt} \bra{\varphi} \Big(\id \otimes U_F(s', s' - \theta\dt)\Big) \ket{\Psi},
\end{align*}
and
\begin{align*}
&\ \bra{s'} \bra{\varphi}\prod_{k=\Lambda+1}^{1} e^{A_k \theta \dt}\ket{\Psi} 
= \bra{s'-\theta\dt} \bra{\varphi} \Big(\id \otimes U_B(s', s' - \theta\dt)\Big) \ket{\Psi}. 
\end{align*}
\end{lemma}

Recall that the delocalized version \eqref{eqn::v2} gives  the following: for any $\ket{\psi} \in \hbt$
\begin{align*}
 &\ \bra{\varphi} \exactU{t+\dt}{t} \ket{\psi} \\
=&\ \bra{t+\dt} \bra{\varphi} e^{-i \bigH \dt} \ket{\one}\ket{\psi}\\
=&\ \bra{t+\dt}\bra{\varphi} \Big(\prod_{k=1}^{\Lambda+1} e^{A_k c_1 \dt}\Big)\Big(\prod_{k=\Lambda+1}^{1} e^{A_k d_1 \dt}\Big) \cdots \\
&\cdots \Big(\prod_{k=1}^{\Lambda+1} e^{A_k c_q \dt}\Big)\Big(\prod_{k=\Lambda+1}^{1} e^{A_k d_q \dt}\Big) \ket{\one} \ket{\psi} +\order{\dt^{n+1}}\\
=&\ \bra{t+\dt-c_1\dt}  \Big(\bra{\varphi} U_F(t+L_1, t+R_1)\Big) \Big(\prod_{k=\Lambda+1}^{1} e^{A_k d_1 \dt}\Big) \\
&\ \cdots \Big(\prod_{k=1}^{\Lambda+1} e^{A_k c_q \dt}\Big)\Big(\prod_{k=\Lambda+1}^{1} e^{A_k d_q \dt}\Big) \ket{\one} \ket{\psi} +\order{\dt^{n+1}}\\
=&\ \cdots \\
\myeq{\eqref{eqn::product_td}}&\ \bra{t} \bra{\varphi} \id\otimes {U}_n(t+\dt, t) \ket{\one} \ket{\psi} +\order{\dt^{n+1}}\\
=&\ \bra{\varphi} {U}_n(t+\dt, t) \ket{\psi} +\order{\dt^{n+1}}.
\end{align*}
To get the third equality, we used the fact that $L_1 = \dt$ and $R_1 = \dt - c_1 \dt$.
This provides a second proof of Theorem \ref{thm::td_pointwise}. 

\begin{remark}
We emphasize that the state $\ket{\Psi}\in \augspace$ in the augmented space \emph{does not need to be a physical state}, and we merely use the bra-ket notation for consistency. 
In particular, for the case $H_k(s) = f(s) \h_k$, with appropriate assumptions on the high-order derivatives of $f\in C^\infty(\timedom)$, one can expect that all operations are well-defined.
However, we notice that such a nice property does not  seem to hold for the discretized \SH{}'s clock with $\omega > 0$, which leads to unsolved technicalities in \cite{watkins_time_2024} when they applied the firstly discretized \SH{}'s clock to multi-product formulas. We will elaborate further how \SH{}'s clock in the \emph{continuous form} can assist with resolving this problem in \secref{sec::mpf} below.
\end{remark}

\section{Product formula (\rom{2}): query to the time integration of $H_k(s)$}
\label{section::HDR}

In this section, we will demonstrate how \SH{}'s clock can establish generalized schemes of \HDR{}'s approach \cite{huyghebaert_product_1990} to arbitrary high-order, by assuming queries to $\exp_{\mathcal{T}}\big(-i \int_{t}^{t+\dt} H_k(s)\ \dd s\big)$ for any $t$ and small $\dt$. Product-based formulas along this line appear to be much less studied, compared to Suzuki's operator $\nicefrac{\overset{\leftarrow}{\partial}}{\partial t}$. Such a generalization seems to be unknown in literature, to the best of our knowledge, and was also discussed as an open question in a recent work \cite{ikeda_minimum_2023}. 
In this section, we will derive novel high-order \HDR{} schemes (or HDR scheme in short) based on \SH{}'s clock, and these schemes turn out to be very efficient in handling time for \DAS.  

Firstly, we will illustrate the main idea and show how Strang splitting scheme together with appropriate decomposition of \SH{}'s clock Hamiltonian can recover the second-order scheme of \HDR{} in \secref{sec::HDR::example}. Then in \secref{sec::HDR::theory}, given any time-independent product formula, we will show how to find a corresponding HDR scheme with the same order. We emphasize that such a family of newly developed high-order scheme does not increase the number of unitary gates required, which addresses the problem of finding minimum-gate implementation in \cite{ikeda_minimum_2023} from yet another different perspective, besides the approach in \secref{sec::pointwise::gate}. 

\subsection{An illustrative example}
\label{sec::HDR::example}
When $\Lambda = 2$, a time-dependent generalization of Strang splitting scheme (midpoint scheme) was studied in \cite{huyghebaert_product_1990}:
\begin{align*}
&\ \exactU{t+\dt}{t} \\
\approx &\  \eT{-i \int_{t+\dthalf}^{t+\dt} H_1(s)\ \dd s} \eT{-i \int_{t+\dthalf}^{t+\dt} H_2(s)\ \dd s}  \eT{-i \int_{t}^{t+\dthalf} H_2(s)\ \dd s} \eT{-i \int_{t}^{t+\dthalf} H_1(s)\ \dd s} \\
\equiv &\ \eT{-i \int_{t+\dthalf}^{t+\dt} H_1(s)\ \dd s} \eT{-i \int_{t}^{t+\dt} H_2(s)\ \dd s} \eT{-i \int_{t}^{t+\dthalf} H_1(s)\ \dd s}.
\end{align*}
If one considers \DAS{}, then one simply has e.g., 
\begin{align*}
\eT{-i \int_{t+\dthalf}^{t+\dt} H_1(s)\ \dd s}  = e^{-i \int_{t+\dthalf}^{t+\dt} H_1(s)\ \dd s}.
\end{align*}

To establish the connection between their scheme and \SH{}'s clock, let us divide the Hamiltonian $\bigH$ into
\begin{align*}
\bigH = \underbrace{\big(\hat{p}_s \otimes \id + H_1(\hat{s})\big)}_{=:A_1} +  \underbrace{- \hat{p}_s \otimes \id}_{=:A_2}  + \underbrace{\big(\hat{p}_s \otimes \id + H_2(\hat{s})\big)}_{=:A_3},
\end{align*}
and apply the Strang splitting (2nd order) for these three (time-independent) operators:
\begin{align*}
&\ \ e^{-i\bigH \dt} \ket{t}\ket{\phi} \\
\approx &\ \ e^{-iA_1 \dthalf} e^{-i A_2 \dthalf} e^{-i A_3 \dthalf} e^{-i A_3 \dthalf} e^{-i A_2 \dthalf}  \Big(e^{-iA_1 \dthalf} \ket{t}\ket{\phi}\Big) \\
\myeq{\eqref{eqn::v1}}&\ \ e^{-iA_1 \dthalf} e^{-i A_2 \dthalf}  e^{-i A_3 \dthalf}  e^{-i A_3 \dthalf} e^{-i A_2 \dthalf}  \ket{t+\nicefrac{\dt}{2}} \big(\eT{-i \int_{t}^{t+\dthalf} H_1(s)\ \dd s} \ket{\phi}\big) \\
=&\ \ e^{-iA_1 \dthalf} e^{-i A_2 \dthalf}  e^{-i A_3 \dthalf} e^{-i A_3 \dthalf} \ket{t} \big(\eT{-i \int_{t}^{t+\dthalf} H_1(s)\ \dd s} \ket{\phi}\big) \\
=& e^{-iA_1 \dthalf} e^{-i A_2 \dthalf}  e^{-i A_3 \dthalf}  \ket{t+\dt/2} \big(\eT{-i \int_{t}^{t+\dthalf} H_2(s)\ \dd s}\ \eT{-i \int_{t}^{t+\dthalf} H_1(s)\ \dd s} \ket{\phi}\big) \\
=&\ \vdots\\
=&\ \ket{t+\dt} \otimes \Big(\eT{-i \int_{t+\dthalf}^{t+\dt} H_1(s)\ \dd s} \eT{-i \int_{t+\dthalf}^{t+\dt} H_2(s)\ \dd s} \eT{-i \int_{t}^{t+\dthalf} H_2(s)\ \dd s} \eT{-i \int_{t}^{t+\dthalf} H_1(s)\ \dd s} \ket{\phi}\Big).
\end{align*}

This is essentially the midpoint scheme developed above after we trace out the additional degree of freedom in the $s$ variable  (or say ignore the $s$-variable).

Unlike Suzuki's approach which uses point-wise access to $H_k(s)$, the above scheme uses the time-integrated Hamiltonian, whereas the whole derivation is essentially the same as the time-independent case. Moreover, one only needs to apply time-independent schemes in the \SH{}'s clock to derive both Suzuki's formula and \HDR{}'s second-order algorithm.

\subsection{A general scheme}
\label{sec::HDR::theory}

Due to the above connection of \HDR{}'s second-order scheme with \SH{}'s clock, we can derive a whole family of time-dependent high-order schemes along the line of \HDR{}'s approach, in a very systemic way.
When given the  $\Lambda$ operators, we can split the Hamiltonian $\bigH$ via the following:
\begin{align}
\label{eqn::split_HDR}
\begin{aligned}
\bigH  &= \underbrace{\big(\hat{p}_s \otimes \id + H_1(\hat{s})\big)}_{=:A_1} + \underbrace{(-\hat{p}_s\otimes \id)}_{=:A_{2}} + \underbrace{\big(\hat{p}_s \otimes \id + H_2(\hat{s})\big)}_{=:A_3}  + \cdots \\
& \qquad + \cdots + \underbrace{(-\hat{p}_s\otimes \id)}_{=:A_{2\Lambda-2}} + \underbrace{\big(\hat{p}_s \otimes \id + H_\Lambda(\hat{s})\big)}_{=:A_{2\Lambda-1}}.
\end{aligned}
\end{align}
In total, there are  $2\Lambda-1$ operators. By applying any time-independent product-based formula for this, we immediately end up with the following approximation:
\begin{theorem}
\label{thm::product_td_v2}
Suppose that the time-independent Hamiltonian can be decomposed as in \eqref{eqn::decompose}. Under the same assumption of Theorem \ref{thm::ostmeyer},
the following scheme gives an order $n$ approximation of the time-ordered unitary evolution $\exactU{t+\dt}{t}$:
\begin{align}
\label{eqn::product_td_v2}
\begin{aligned}
 \hdu_n(t+\dt, t) :=  &\ \hdu_{F} \big(t+L_1, t+R_1\big) \hdu_{B}\big(t+R_1, t+L_{2}\big) \cdots \\
&\ \ \cdots \hdu_{F}\big(t+L_{q-1}, t+R_{q-1}\big) \hdu_{B}\big(t+R_{q-1}, t+L_q\big) \\
&\ \  \hdu_{F}\big(t+ L_{q}, t+R_q\big) \hdu_{B}\big(t+R_q, t + 0\big),
\end{aligned}
\end{align}
where 
\begin{align*}
\left\{
\begin{aligned}
\hdu_{F}\big(t', t\big) &= \prod_{k=1}^{\Lambda} \eT{-i \int_{t}^{t'} H_k(s) \dd s}, \\
\hdu_{B}\big(t', t\big) &= \prod_{k=\Lambda}^{1} \eT{-i \int_{t}^{t'} H_k(s) \dd s}.
\end{aligned}\right.
\end{align*}
Suppose we are given query access to $\eT{-i \int_{t}^{t'} H_k(s)\ \dd s}$ for any $t, t'$ and the index $k$,  this generalized HDR{} scheme always only uses $2\Lambda q - (2q-1)$ gates for each time step, and therefore, it {\bf costs no extra gates} compared with the time-independent product formula.
\end{theorem}

\begin{proof}
The proof is essentially the same as the above illustrative example except that we apply Theorem~\ref{thm::ostmeyer} for high-order splitting in replace of the second order Strang scheme. The detailed derivation is no different from Theorem \ref{thm::td_pointwise}.
\end{proof}

\begin{remark}
When we choose the Forest-Ruth-Suzuki formula \eqref{eqn::FRS} in Theorem \ref{thm::product_td_v2}, the explicit expression is given in \eqref{eqn::HDR_FRS} in Appendix. Such a scheme is different from \eqref{eqn::FRS_Magnus} derived by Magnus' expansion in Ref.~\cite{ikeda_minimum_2023}. Similarly, one could derive the relevant sixth order algorithm with minimum gate complexity by applying Yoshida’s sixth-order formula, and theoretically to any order.
\end{remark}

\section{Time-dependent Multi-product Formula}
\label{sec::mpf}

In this section, we will show how the \SH{}'s continuous clock formalism can help to establish the multi-product formula (MPF) used as quantum algorithms for time-dependent Hamiltonian simulation. The MPF for time-dependent cases was proposed in \cite{watkins_time_2024} using discretized \SH{}'s clock. However, one important foundational result was only proposed as a conjecture due to a technical challenge of using the discrete clock, and thus it requires further study. We will complete this missing component via \SH{}'s continuous clock, which is technically different from \cite{watkins_time_2024}; the key advantage of the continuous clock is that we can avoid the technicality of using localized Gaussian wave function to approximate the delta function. 

This section is organized as follows. We will first review the time-independent MPF in \secref{sec::MPF::review} and then show how the continuous \SH{}'s clock can help with the conjecture from \cite{watkins_time_2024} in \thmref{thm::mpf_v1} in \secref{sec::MPF::v1}. Since we had shown in precious sections that Suzuki's approach and \HDR{}'s approach are essentially the same in the augmented clock space, the base product formula in MPF can be straightforwardly replaced by the midpoint scheme of \HDR{}, without extra effort; see \thmref{thm::mpf_v2}. Though there might be other ways to establish the time-dependent MPF, the \SH{}'s clock offers a simple systematic approach to achieve this.

\subsection{MPF for time-independent Hamiltonian simulation}
\label{sec::MPF::review}

The multi-product formula for operator splitting was studied by Chin and Geiser \cite{chin_multi_product_2011}, and contemporarily for Hamiltonian dynamics by Childs and Wiebe \cite{childs_hamiltonian_2012}, and later further improved in \cite{low_well_conditioned_2019,endo_mitigating_2019}. The main idea behind is to construct a higher-order scheme using the additive construction, namely a linear combination of lower order scheme with multiple time steps \cite{blanes_extrapolation_1999} in the way similar to the Richardson's extrapolation.
For the time-independent case, one has
\begin{align}
\label{eqn::mpf}
e^{-i H \dt} = \sum_{j=1}^{M} \alpha_j \Big(U_2(\dt/k_j)\Big)^{k_j} + \order{\dt^{2m+1}},
\end{align}
where $\alpha_j$ and $k_j$ are well-chosen to satisfy conditions in \eqref{eqn::mpf_order} so as to ensure that the truncation error is $\order{\dt^{2m+1}}$.
The right-hand side takes the form of a linear combination of unitaries (LCU), which can be implemented using LCU techniques with controlled-$U_2$ query and additional qubits. The above $U_2(s)$ is typically chosen as the Strang splitting scheme or say the midpoint scheme \eqref{eqn::midpoint}. Of course, $U_2$ can be theoretically replaced by other product formulas like any high-order symmetric product formula \cite[Sect.~5]{aftab_multi_product_2024}, but for simplicity, we shall stick with this particular choice. There is a very recent progress to provide a detailed mathematical treatment for time-independent case \cite{aftab_multi_product_2024}, which showed both logarithmic error dependence and commutator scaling of MPF for \emph{time-independent} Hamiltonian simulation.

\subsection{MPF for time-dependent Hamiltonian simulation}
\label{sec::MPF::v1}

The time-dependent MPF is comparatively much less studied. Suzuki's operator $\nicefrac{\overset{\leftarrow}{\partial}}{\partial t}$ was adopted in Chin and Geiser \cite{chin_multi_product_2011} to formally provide a time-dependent MPF scheme, and was further explored by Geiser in \cite{geiser_multi_product_2011}. They considered the operator-splitting methods for general linear dynamics without fully tailored results for quantum algorithms. Recently, \cite{watkins_time_2024} proposed a well-conditioned multi-product formula by replacing the consecutive concatenation of $U_2$ operator via its time-dependent analog, and discuss how to use LCU quantum algorithm to implement it. However, their result relies on Conjecture 1 in \cite{watkins_time_2024}. The technicality mentioned in their work seem to arise from adopting the discretized clock.
In what follows, we will discuss how the continuous-form of \SH{}'s clock can help with the challenge of discrete clock in \cite{watkins_time_2024}, and also recover the scheme in \cite{geiser_multi_product_2011} (which used Suzuki's time operator) with slight generalization in \thmref{thm::mpf_v2}.
Since the follow-up application of LCU algorithm has been well treated and explained in various literatures (e.g., \cite{low_well_conditioned_2019,watkins_time_2024,aftab_multi_product_2024}), we shall exclusively focus on the MPF itself, which is  stated in the following theorem.

\begin{theorem}
\label{thm::mpf_v1}
Suppose $\alpha_j$ and $k_j$ are chosen such that \eqref{eqn::mpf} holds (namely, we choose parameters for well-conditioned time-independent case). Then 
\begin{align*}
\exactU{t+\dt}{t} =& \sum_{j=1}^{M} \alpha_j \prod_{\ell={k_j-1}}^{0} {U}_2\Big(t+\frac{(\ell+1) \dt}{k_j}, t+\frac{\ell \dt}{k_j}\Big) + \order{\dt^{2m+1}},
\end{align*}
where for any $t,s\in \timedom$,
\begin{align*}
{U}_2(t,s) := \prod_{k=1}^{\Lambda} e^{-i \frac{t-s}{2} {H}_k(\frac{t+s}{2})}\prod_{k=\Lambda}^{1} e^{-i \frac{t-s}{2} {H}_k(\frac{t+s}{2})}.
\end{align*}
\end{theorem}

When one has access to $\exp_{\mathcal{T}}\big(-i\int_{t_1}^{t_2} H_k(s)\dd s\big)$ for any $t_1$, $t_2$, and $k$, one can use the splitting choice in \eqref{eqn::split_HDR} and end up with the following scheme:
\begin{theorem}
\label{thm::mpf_v2}
Suppose $\alpha_j$ and $k_j$ are chosen such that \eqref{eqn::mpf} holds (namely, parameters for time-independent case). Then
\begin{align*}
\exactU{t+\dt}{t} =& \sum_{j=1}^{M} \alpha_j \prod_{\ell=k_j-1}^0 \hdu_2(t+\nicefrac{(\ell+1) \dt}{k_j}, t+\nicefrac{\ell \dt}{k_j}) + \order{\dt^{2m+1}}.
\end{align*}
where 
\begin{align*}
\hdu_2(t,s) := \prod_{k=1}^{\Lambda} \expT{-i \int_{s+\frac{t-s}{2}}^{t} {H}_k(r)\dd r}\ \prod_{k=\Lambda}^{1} \expT{-i \int_{s}^{s+\frac{t-s}{2}} {H}_k(r)\dd r}.
\end{align*}
\end{theorem}

\begin{remark}
The base scheme can be replaced via higher-order time-independent product formulas as shown in \cite{ostmeyer_optimised_2023} (cf. \thmref{thm::ostmeyer}). As the derivations are essentially the same as the above two Theorems, and our main purpose is to demonstrate the effectiveness and versatility of continuous \SH{}'s clock in developing quantum algorithms, we shall not pursue high-order base cases here for simplicity.
\end{remark}

The proof of \thmref{thm::mpf_v2} is the same as Theorem \ref{thm::mpf_v1}. In the following subsection, we will show how to use both localized and delocalized formalism to prove \thmref{thm::mpf_v1}.

\subsection{Proof of \thmref{thm::mpf_v1}}
The proof essentially follows the presentation of \cite{low_well_conditioned_2019,aftab_multi_product_2024} with an application of the formalism discussed in \eqref{eqn::v1} (or equivalently \eqref{eqn::v2})  in the end.
Let us apply the midpoint formula \eqref{eqn::midpoint} for the time-independent Hamiltonian $\bigH$ and obtain $e^{-i \bigH \dt} \approx \bigU_2(\dt)$ where 
\begin{align*}
&\ \bigU_2(\dt) 
:=\ e^{-i \frac{\dt}{2} \hat{p}_s \otimes \id }\prod_{j=1}^{\Lambda} e^{-i \frac{\dt}{2} H_k(\hat{s})} \prod_{j=\Lambda}^{1} e^{-i \frac{\dt}{2} H_k(\hat{s})} e^{-i \frac{\dt}{2} \hat{p}_s \otimes \id}.
\end{align*}
By following the derivation in \cite{aftab_multi_product_2024} or simply by applying the time-independent MPF for augmented space \eqref{eqn::mpf}, we have the following:
\begin{align*}
&\ \sum_{j=1}^{M} \alpha_j \bigU_2\big(\nicefrac{\dt}{k_j}\big)^{k_j} =\ e^{-i \bigH \dt} + \order{\dt^{2m+1}}.
\end{align*}

{\noindent \emph{Localized $G$ state:}}
By acting them on the state $\ket{t}\ket{\phi}$, one knows that
\begin{align*}
\sum_{j=1}^{M} \alpha_j \bigU_2\big(\nicefrac{\dt}{k_j}\big)^{k_j} \ket{t}\ket{\phi} &= e^{-i \bigH \dt} \ket{t}\ket{\phi} + \order{\dt^{2m+1}}.
\end{align*}
By \eqref{eqn::v1} (with $\varphi$ being arbitrary), 
\begin{align*}
&\ \bra{\varphi} \exactU{t+\dt}{t} \ket{\phi}\\
\myeq{\eqref{eqn::v1}}&\ \bra{\one}\bra{\varphi} e^{-i \bigH \dt} \ket{t}\ket{\phi} \\
=&\ \sum_{j=1}^{M} \alpha_j \bra{\one}\bra{\varphi} \bigU_2\big(\nicefrac{\dt}{k_j}\big)^{k_j} \ket{t}\ket{\phi}  + \order{\dt^{2m+1}}\\
=&\ \sum_{j=1}^{M} \alpha_j  \bra{\varphi} \prod_{\ell=k_j-1}^{0} {U}_2(t+\nicefrac{(\ell+1) \dt}{k_j}, t+\nicefrac{\ell \dt}{k_j}) \ket{\phi} + \order{\dt^{2m+1}}.
\end{align*}
As $\varphi$ and $\phi$ are arbitrary, then one easily gets the final result in this theorem.
We remark that to get the last line in the last equation, we used
\begin{align*}
&\ \bra{\one}\bra{\varphi} \bigU_2\big(\nicefrac{\dt}{k_j}\big)^{k_j} \ket{t}\ket{\phi} \\
=&\ \bra{\one}\bra{\varphi} \bigU_2\big(\frac{\dt}{k_j}\big)^{k_j-1} \ket{t+\nicefrac{\dt}{k_j}}\ket{ {U}_2\big(t+\nicefrac{\dt}{k_j}, t\big)\phi}\\
=&\ \cdots \\
=&\ \bra{\one}\ket{t+\dt} \bra{\varphi}\ket{\prod_{\ell=k_j-1}^0 {U}_2(t+\nicefrac{(\ell+1) \dt}{k_j},t+\nicefrac{\ell t}{k_j})\phi}\\
=&\ \bra{\varphi} \prod_{\ell=k_j-1}^{0} {U}_2(t+\nicefrac{(\ell+1) \dt}{k_j}, t+\nicefrac{\ell \dt}{k_j}) \ket{\phi}.
\end{align*}

{\noindent \emph{Delocalized $G$ state}.} By the delocalized version \eqref{eqn::v2}, one has 
\begin{align*}
&\ \bra{\varphi} \exactU{t+\dt}{t} \ket{\phi}\\
\myeq{\eqref{eqn::v2}}&\ \bra{t+\dt}\bra{\varphi} e^{-i \bigH \dt} \ket{\one}\ket{\phi} \\
=&\ \sum_{j=1}^{M} \alpha_j \bra{t+\dt}\bra{\varphi} \bigU_2\big(\nicefrac{\dt}{k_j}\big)^{k_j} \ket{\one}\ket{\phi}  + \order{\dt^{2m+1}}.
\end{align*}
For each term, one has
\begin{align*}
&\ \bra{t+\dt}\bra{\varphi} \bigU_2\big(\nicefrac{\dt}{k_j}\big)^{k_j} \ket{\one}\ket{\phi} \\
=&\ \bra{t+\dt}\bra{\varphi} \bigU_2\big(\nicefrac{\dt}{k_j}\big) \bigU_2\big(\nicefrac{\dt}{k_j}\big)^{k_j-1} \ket{\one}\ket{\phi}\\
=&\ \bra{t+\dt-\nicefrac{\dt}{k_j}}\Big(\bra{\varphi} U_2\big(t+\dt, t+\dt-\nicefrac{\dt}{k_j})\Big) \bigU_2\big(\nicefrac{\dt}{k_j}\big)^{k_j-1} \ket{\one}\ket{\phi}\\ 
=&\ \cdots \\
=&\ \bra{t} \Big(\bra{\varphi} \prod_{\ell=k_j-1}^0 U_2\big(t+\nicefrac{(\ell+1)\dt}{k_j}, t+\nicefrac{\ell \dt}{k_j})\Big) \ket{\one}\ket{\phi}\\
=&\ \bra{\varphi} \prod_{\ell=k_j-1}^0 U_2\big(t+\nicefrac{(\ell+1)\dt}{k_j}, t+\nicefrac{\ell \dt}{k_j}) \ket{\phi}.
\end{align*}
By combining the last two questions, one arrives at the same conclusion as the localized version \eqref{eqn::v1}.

\section{Time-dependent qDrift}
\label{sec::qDrift}

In this section, we will demonstrate that by applying qDrift \cite{campbell_random_2019}, initially developed by Campbell for time-independent Hamiltonian simulation, to the augmented \SH{}'s clock space, we can readily develop time-dependent qDrift algorithms.
We would like to emphasize that similar calculations can also been carried on to randomized permutation \cite{childs_faster_2019}, and partially randomized algorithm \cite{jin_partially_2023,hagan_composite_2023} without theoretical challenge (which we shall not further pursue here for clarity).
 
This section is organized as follows. We will first review the idea of qDrift in \secref{sec::qDrift_indep} for time-independent Hamiltonian simulation. With \SH{}'s clock, we can immediately generalize qDrift to the time-dependent case in \secref{sec::td_qDrift} and in \secref{sec::td_qDrift_v2}. We remark that even though such a generalization can also been conjectured and validated without  \SH{}'s clock, this generic framework, nevertheless, provides \emph{a unifying and systemic approach} for such algorithms built on random quantum circuit. More importantly, estimating the error is no different from the time-independent case. Finally in \secref{sec::td_qDrift_v3}, we will establish the connection between time-dependent qDrift developed using \SH{}'s clock and the continuous qDrift in \cite{berry_time_dependent_2020}. For the setup of \DAS{}, we will show that \enquote{continuous qDrift} proposed in \cite{berry_time_dependent_2020} is {almost} equivalent to applying the time-independent qDrift in the augmented \SH{}'s clock Hamiltonian (see \lemref{lem::c_qDrift_equiv}). 

\subsection{Review of qDrift for the time-independent case}
\label{sec::qDrift_indep}

We first recall the qDrift for the time-independent Hamiltonian with the form given in \eqref{eqn::decompose} and \eqref{eqn::Hk}, where $f_k(s)$ is time-independent (namely, $f_k(s) \equiv f_k$ is simply a constant). \cite{campbell_random_2019} proposed the qDrift algorithm, which rewrites the unitary evolution (in the density matrix formalism) as the expectation of local unitary evolution, which can then be realized via random unitary circuits. More specifically, we consider the following decomposition
\begin{align*}
H = \sum_{k} \lambda_k \Big(\frac{H_k}{\lambda_k}\Big),
\end{align*}
and the qDrift algorithm uses the following approximation:
\begin{align}
\label{eqn::qDrift}
\begin{aligned}
e^{-i \dt H} \rho\ e^{i \dt H} 
=&\ \sum_{k} \lambda_k e^{-i \dt \frac{H_k}{\lambda_k}} \rho e^{i \dt \frac{H_k}{\lambda_k}} + \order{\dt^2},
\end{aligned}
\end{align}
where $\lambda_j$ is an arbitrary non-degenerate discrete probability distribution. The key advantage of this qDrift algorithm comes from a well-chosen probability distribution to balance the magnitude of $f_j$. For instance, when $\norm{\h_k} =\order{1}$, one can choose $\lambda_j := \abs{f_j}/\sum_{k} \abs{f_k}$ so that the final approximation error scales with respect to $\sum_{k} \abs{f_k}$, rather than directly depending on the number of local Hamiltonian terms $\Lambda$.

This idea can be straightforwardly generalized to infinite many terms. Suppose that there is a probability space with event $\omega$ such that the Hamiltonian can be written as the expectation of a collection of \enquote{local Hamiltonians} $H_{\omega}$ as follows:
 \begin{align*}
H = \ee_{\omega} \big[ H_{\omega}\big] \equiv \int\dd \omega\ \lambda(\omega) H_{\omega},
\end{align*}
where $H_{\omega}$ is a Hermitian operator and $\lambda(\omega)$ is a non-degenerate probability distribution on a certain set of indices. Then the above estimate \eqref{eqn::qDrift} can be essentially rewritten as 
\begin{align}
\label{eqn::qDrift_inf}
\underbrace{e^{-i \dt H} \rho\ e^{i \dt H}}_{=:\ \mathcal{U}_{\dt}(\rho)}  = \underbrace{\int \dd \omega\ \lambda(\omega) e^{-i \dt H_\omega}\ \rho\  e^{i \dt H_\omega}}_{=:\  \mathcal{E}_{\lambda,\dt}(\rho)} + \order{\dt^2}.
\end{align}
The detailed error analysis has been investigated in \cite{campbell_random_2019} and more comprehensive analysis in \cite{chen_concentration_2021}, which can be applied to analyze the error in the last equation.

We collect some facts:
\begin{lemma}
\label{lem::qdrift_indep}
Denote the measure $\eta(\omega_1, \omega_2) = \lambda(\omega_1) \big(\lambda(\omega_2) - \delta(\omega_1 - \omega_2)\big)$.
The distance between the above two channels  is bounded above by 
\begin{align*}
&\ \norm{\mathcal{U}_{\dt} - \mathcal{E}_{\lambda,\dt}}_{1} \le \frac{\dt^2}{2} C + \order{\dt^3},
\end{align*}
where 
\begin{align*}
C = \iint \dd\omega_1\dd\omega_2\ \abs{\eta}(\omega_1,\omega_2)\ \norm{\comm\big{H_{\omega_1}}{\comm{H_{\omega_2}}{\cdot}}}_{1}.
\end{align*}
\end{lemma}
The proof only involves direct computation and is postponed to \appref{proof::qDrift}.
Given estimates about right nested commutator of $\comm\big{H_{\omega_1}}{\comm{H_{\omega_2}}{\cdot}}$, one could choose the probability measure $\lambda$ such that the right-hand side is minimized.

\subsection{Time-dependent qDrift using \SH{}'s clock (\rom{1})}
\label{sec::td_qDrift}

For the general time-dependent Hamiltonian $H$ as in \eqref{eqn::decompose} and a general discrete probability distribution $\{\lambda_k\}_{k=1}^{\Lambda}$, let us consider the following decomposition 
\begin{align}
\label{eqn::H_qdrift}
\bigH = \sum_{k=1}^{\Lambda} \lambda_k \underbrace{\Big(\hat{p}_s \otimes \id + \frac{1}{\lambda_k} H_k(\hat{s})\Big)}_{=:\wt{H}_k}.
\end{align}
By applying the above time-independent qDrift algorithm \eqref{eqn::qDrift_inf} for the augmented Hamiltonian $\bigH$, one has
\begin{align*}
&\ \  \exactU{t+\dt}{t}\ \rho\ \exactU{t+\dt}{t}^\dagger\\
 \myeq{\eqref{eqn::v2}}&\ \bra{t+\dt}_s e^{-i \bigH \dt}\ \big(\ket{\one}\bra{\one}\otimes \rho\big)\ e^{i \bigH \dt} \ket{t+\dt}_s \\
\myeq{\eqref{eqn::qDrift_inf}}&\ \sum_{k} \lambda_k\ \bra{t+\dt}_s  e^{-i \dt \wt{H}_k} \big(\ket{\one}\bra{\one}\otimes \rho\big)\ e^{i \dt \wt{H}_k} \ket{t+\dt}_s + \order{\dt^2}\\
=&\ \sum_{k}\lambda_k\ \bra{t+\dt}_s e^{-i \dt\big(\hat{p}_s\otimes \id + \frac{1}{\lambda_k} H_k(\hat{s})\big)} \big(\ket{\one}\bra{\one}\otimes \rho\big)\ e^{i \dt\big(\hat{p}_s\otimes \id + \frac{1}{\lambda_k} H_k(\hat{s})\big)}  \ket{t+\dt}_s + \order{\dt^2}\\
\myeq{\eqref{eqn::v2}}&\ \sum_{k}\lambda_k\ \eT{-i \int_{t}^{t+\dt} \lambda_k^{-1} H_k(s)\dd s}\ \rho\ \eT{i \int_{t}^{t+\dt} \lambda_k^{-1} H_k(s)\dd s} + \order{\dt^2}.
\end{align*}
 
This suggests the approximation
\begin{align}
\exactU{t+\dt}{t}\ \rho\ \exactU{t+\dt}{t}^\dagger \approx \mathcal{E}_{\lambda, t, \dt}(\rho),
\end{align}
where the latter quantum channel is defined as 
\begin{align}
\label{eqn::cqDrift_v1}
\begin{aligned}
\mathcal{E}_{\lambda, t, \dt}(\rho) 
:=&\ \sum_{k}\lambda_k\ \eT{-i \int_{t}^{t+\dt} \lambda_k^{-1} H_k(s)\dd s}\ \rho\ \eT{i \int_{t}^{t+\dt} \lambda_k^{-1} H_k(s)\dd s}.
\end{aligned}
\end{align}
This immediately generalizes the time-independent qDrift \eqref{eqn::qDrift} to the time-dependent Hamiltonian simulation.
One can immediately check that the error scales like the following:
\begin{lemma}
\label{lem::qdrift_err_1}
Suppose the time-dependent Hamiltonian has the form in \eqref{eqn::decompose} and \eqref{eqn::Hk}.
The error between approximated state and the exact state is 
\begin{align*}
 \norm{\exactU{t+\dt}{t}\ \rho\ \exactU{t+\dt}{t}^\dagger -  \mathcal{E}_{\lambda,t,\dt}(\rho)}_{1} 
\le&\ \frac{C\dt^2}{2} + \order{\dt^3},
\end{align*}
where
\begin{align*}
C = \sum_{k_1, k_2} \abs{\lambda_{k_2} - \delta_{k_1, k_2}} \frac{\abs{f_{k_1}(t+\dt) f_{k_2}(t+\dt)}}{ \lambda_{k_2}} \norm\Big{\comm\Big{\h_{k_1}}{\comm{\h_{k_2}}{\rho}}}_{1}.
\end{align*}
\end{lemma}
The detailed proof is postponed to \appref{proof::qDrift}.
With estimates of the upper bound of $C$, one can adaptively choose the time-step $\dt$, as well as finding the optimal probability distribution on-the-fly.

\subsection{Time-dependent qDrift using \SH{}'s clock (\rom{2})}
\label{sec::td_qDrift_v2}

Suppose that we consider the following probability space with events $\omega = (k, r) \in \{1,2,\cdots,\Lambda\}\times [0,1]$ following a non-degenerate discrete and continuous hybrid distribution $\mu$.
Let us consider a more complex decomposition,
\begin{align}
\label{eqn::H_omega}
\begin{aligned}
\bigH =&\ \sum_{k=1}^{\Lambda}  \int_{0}^{1}\dd r\ \mu(k,r) \Big(\hat{p}_s\otimes \id + \frac{H_k(\hat{s})}{\mu(k,r)}\Big) \\
=&\ \ee_{\omega} \underbrace{\big(\hat{p}_s\otimes \id + H_\omega(\hat{s})\big)}_{=: \wt{H}_\omega},
\end{aligned}
\end{align}
where $H_{\omega}(\hat{s}) = \frac{H_k(\hat{s})}{\mu(k,r)}$ for this example. When $\mu(k,r) \equiv \lambda_k$ is independent of $r$, then this decomposition simply reduces to the above \eqref{eqn::H_qdrift}.
By applying the qDrift algorithm \eqref{eqn::qDrift_inf} to the Hamiltonian decomposition \eqref{eqn::H_omega} in the augmented space, we have the following approximations 
\begin{align*}
&\ \ \exactU{t+\dt}{t}\ \rho\ \exactU{t+\dt}{t}^\dagger\\
 \myeq{\eqref{eqn::v2}}&\ \bra{t+\dt}_s e^{-i \bigH \dt}\ \big(\ket{\one}\bra{\one}\otimes \rho\big)\ e^{i \bigH \dt} \ket{t+\dt}_s \\
 \myeq{\eqref{eqn::qDrift_inf}}&\ \ee_{\omega} \bra{t+\dt}_s e^{-i \dt \wt{H}_\omega}\ \big(\ket{\one}\bra{\one}\otimes \rho\big)\ e^{i \dt \wt{H}_\omega} \ket{t+\dt}_s + \order{\dt^2}\\
=&\  \ee_{\omega} \bra{t+\dt}_s e^{-i \dt \big(\hat{p}_s\otimes \id + H_{\omega}(\hat{s})\big)}\ \big(\ket{\one}\bra{\one}\otimes \rho\big)\ e^{i \dt \big(\hat{p}_s \otimes \id + {H}_\omega(\hat{s})\big)} \ket{t+\dt}_s + \order{\dt^2}\\
=& \sum_{k} \int_{0}^{1} \dd r\ \mu(k,r)\ \eT{-i \frac{\int_{t}^{t+\dt} H_k(s)\dd s}{\mu(k,r)}}\ \rho\ \eT{i \frac{\int_{t}^{t+\dt} H_k(s)\dd s}{\mu(k,r)}} + \order{\dt^2}.
\end{align*}
This suggests the approximation
\begin{align*}
\mathcal{U}(t+\dt,t)\ \rho\ \mathcal{U}(t+\dt,t)^\dagger \approx \mathcal{E}_{\mu, t, \dt}(\rho),
\end{align*}
where
\begin{align}
\label{eqn::cqDrift_v2}
\begin{aligned}
\mathcal{E}_{\mu, t, \dt}(\rho) 
=&\  \sum_{k} \int_{0}^{1} \dd r\ \mu(k,r)\ \eT{-i \frac{\int_{t}^{t+\dt} H_k(s)\dd s}{\mu(k,r)}}\ \rho\ \eT{i \frac{\int_{t}^{t+\dt} H_k(s)\dd s}{\mu(k,r)}}.
\end{aligned}
\end{align}

\begin{lemma}
\label{lem::qdrift_err_2}
Suppose the time-dependent Hamiltonian has the form in \eqref{eqn::decompose} and \eqref{eqn::Hk}.
The error between approximated state and the exact state is 
\begin{align*}
 \norm{\exactU{t+\dt}{t} \rho\ \exactU{t+\dt}{t}^\dagger -  \mathcal{E}_{\mu,t,\dt}(\rho)}_{1} 
\le & \frac{C\dt^2}{2} + \order{\dt^3},
\end{align*}
where
\begin{align*}
& C = \iint\dd\omega_1\dd\omega_2\  \left(
\abs{\mu(\omega_2) - \delta(\omega_1 - \omega_2)} 
 \frac{\abs{f_{k_1}(t+\dt) f_{k_2}(t+\dt)}}{\mu(\omega_2)} 
 \norm{\comm\big{\h_{k_1}}{\comm{\h_{k_2}}{\cdot}}}_{1}
\right),
 \end{align*}
 and $\omega_1 = (k_1, r_1)$, $\omega_2 = (k_2, r_2)$, $\iint \dd \omega = \sum_{k} \int_{0}^{1}\dd r$.
\end{lemma}
See \appref{proof::qDrift} for detailed proof.

\subsection{Connections to {continuous qDrift}}
\label{sec::td_qDrift_v3}

The above noisy channels \eqref{eqn::cqDrift_v1} and \eqref{eqn::cqDrift_v2} can be regarded as reasonable ways to generalize the original time-independent qDrift to time-dependent case. Another possibility, known as the \enquote{continuous qDrift}, was proposed in \cite[Eq.~(77)]{berry_time_dependent_2020}:
\begin{align}
\label{eqn::c-qDrift}
\begin{aligned}
 \mathcal{E}^{\text{(c-qDrift)}}_{q,t,\dt}(\rho) 
:=&\ \sum_{k} \int_{0}^{1} \dd \tau\ q(k,\tau)\  e^{-i \dt \frac{H_{k}(t+\tau \dt)}{q(k,\tau)}}\ \rho\ e^{i \dt \frac{H_{k}(t+\tau \dt)}{q(k,\tau)}},
\end{aligned}
\end{align}
where $q$ is a  probability measure on the same discrete-continuous hybrid probability space.
As a remark, for consistency of notations, we have slightly changed notations in \cite[Eq.~(77)]{berry_time_dependent_2020} and rescaled the \enquote{probability distribution $p_{\ell}(\tau)$} in their equation from $[0,\dt]$ to the range $\tau \in [0,1]$ above.

In general, by comparing the formalism \eqref{eqn::cqDrift_v2} and \eqref{eqn::c-qDrift}, we can observe that the continuous qDrift can be obtained by performing the following approximation for any $\tau\in [0,1]$
\begin{align*}
\eT{-i \frac{\int_{t}^{t+\dt} H_k(s)\dd s}{\mu(k,r)}} = e^{-i \dt \frac{H_k(t+\tau\dt)}{\mu(k,\tau)}} + \order{\dt^2}.
\end{align*}
This is essentially  Lie's first-order approximation using the time grid point $t+\tau\dt$.

We are particularly interested in time-dependent Hamiltonian simulation for \DAS{} with the form in \eqref{eqn::decompose} and \eqref{eqn::Hk}.
In this setup, despite that \eqref{eqn::cqDrift_v2} uses the integration-based query and \eqref{eqn::c-qDrift} uses the pointwise query, the formalism \eqref{eqn::cqDrift_v2} and \eqref{eqn::c-qDrift} are  {equivalent} after rescaling.

\begin{lemma}
\label{lem::c_qDrift_equiv}
Suppose the time-dependent Hamiltonian has the form in \eqref{eqn::decompose} and \eqref{eqn::Hk}. We further assume that $f_k(s) > 0$ for any $k$ and time $s\in [t,t+\dt]$.
Then for any arbitrary probability measure $\mu$, there exists a measure $q$ such that 
\begin{align*}
\mathcal{E}_{\mu, t, \dt} =\mathcal{E}^{\text{(c-qDrift)}}_{q,t,\dt}.
\end{align*}
Namely, applying qDrift for \SH{}'s clock using decomposition \eqref{eqn::H_omega} leads into continuous qDrift in \cite{berry_time_dependent_2020}.
\end{lemma}

See \appref{proof::qDrift} for detailed proof.

\section{Taylor's expansion based algorithms}
\label{sec::taylor}

This section aims to illustrate the compatibility of \SH{}'s clock for LCU-Taylor based algorithms. We shall first discuss existing Dyson's expansion based methods, which might be hard to implement in the near future due to the black-box type of oracle query. As the Hamiltonian $\bigH$ is essentially time-independent, we consider a direct application of Taylor's expansion to \SH{}'s clock and present the resulting quantum algorithms based on LCU models. The latter one is not the same as applying  Dyson's expansion. We aim to illustrate that LCU-Taylor based quantum algorithms are compatible with \SH{}'s clock, which further confirms the versatility of this general approach in developing quantum algorithms for time-dependent dynamics.

\subsection{Related works}

Assuming that the Hamiltonian has the form of linear combination of unitaries (e.g., assume that $\h_j$ in \eqref{eqn::Hk} are unitary) and $H$ is time-independent, \cite{berry_simulating_2015} provided a simple method based on Taylor's expansion to achieve Hamiltonian simulation with $\order{\frac{T\log(\nicefrac{T}{\eps})}{\log \log(\nicefrac{T}{\eps})}}$ queries to controlled-$\h_{j}$ gates and other two-qubit gates, where $T$ is the length of simulation time.

The time-dependent generalization has been studied in \cite{kieferova_simulating_2019,low_hamiltonian_2019_paper_b} using Dyson's series, which is the generalization of Taylor's expansion by considering the time-ordering of Hamiltonians. If one discretizes the Hamiltonian $H(s)$ at grid points $t_j = j\frac{t}{M}$ ($0\le j\le M-1$) and for time-dependent Hamiltonian simulation, one has to use all information from $H(s)$ to simulate the dynamics, which motivated \cite{low_hamiltonian_2019_paper_b} to use oracle queries to HAM, defined as follows
\begin{align}
\label{eqn::HAM}
\bra{0}_{a} \text{HAM} \ket{0}_a = \frac{1}{\alpha}\sum_{m=0}^{M-1} \ketbra{m} \otimes H\left(\frac{m t}{M}\right) ,
\end{align}
where $\alpha := \sup_{s\in [0,T]} \norm{H(s)}$; see Definition 2 in \cite{low_hamiltonian_2019_paper_b}. For a given error tolerance $\eps$, it has been shown in e.g., \cite{kieferova_simulating_2019,low_hamiltonian_2019_paper_b} that one requires $M = \Theta\big(\nicefrac{1}{\eps}\big)$ \cite[Lemma 5]{low_hamiltonian_2019_paper_b}. \cite[Theorem 3]{low_hamiltonian_2019_paper_b} showed that the required number of queries to HAM logarithmically depends on the error $\eps$.  However, the less discussed issue is that during implementation, the gate cost for HAM may scale as $\order{\nicefrac{1}{\eps}}$ in the worst case. In \cite{kieferova_simulating_2019}, they showed similar results and the actual implementation of oracle queries {appear to be} not explicitly explained. Such a framework has been further generalized to develop first-order \cite{an_time_dependent_2022} and second-order  \cite{fang_time_dependent_2024} algorithms based on the Magnus expansion.

\subsection{The second-order Taylor-based algorithm for \SH{}'s clock}

In this part, we will discuss quantum algorithms based on \SH{}'s clock using second-order Taylor's expansion, which complements the aforementioned Dyson's expansion \cite{kieferova_simulating_2019,low_hamiltonian_2019_paper_b} and the recent works using Magnus' expansion \cite{an_time_dependent_2022,fang_time_dependent_2024}. 
Its generalization to high-order cases should merely be a routine task, if one only aims at finding a finite-order algorithm. As mentioned above, we present such an algorithm for illustration only to provide a more complete picture of how \SH{}'s clock can help with developing quantum algorithms. As a remark, we shall assume that $\h_j$ is unitary throughout this subsection.

By applying  Taylor's expansion directly to \SH{}'s clock, we have the following
\begin{align*}
e^{-i \dt \bigH} =& \id - i \dt \Big(\hat{p}_s\otimes \id  + \sum_{j=1}^{\Lambda} f_j(\hat{s}) \h_j\Big) \\
&\ + \frac{-\dt^2}{2} \Big(\hat{p}_s\otimes \id  + \sum_{j=1}^{\Lambda} f_j(\hat{s}) \h_j\Big)^2 + \order{\dt^3}.
\end{align*}
By applying the delocalized version \eqref{eqn::v2},
\begin{align}
\label{eqn::taylor_LCU}
\begin{aligned}
&\ \exactU{t+\dt}{t} =\ \bra{t+\dt} e^{-i \bigH \dt} \ket{\one} \\
= &\ \id - i\dt \Big(\sum_{j=1}^{\Lambda} f_j(t+\dt) \h_j\Big) \\
&\qquad + \frac{-\dt^2}{2} \Big(-i \sum_{j=1}^{\Lambda} f_j'(t+\dt) \h_j  + \sum_{j,k} f_j(t+\dt) f_k(t+\dt) \h_j \h_k\Big)  + \order{\dt^3}\\
=&\ \id + \sum_{j=1}^{\Lambda} \big(\dt f_j(t+\dt) -\frac{\dt^2}{2} f_j'(t+\dt)\big) (-i \h_j) \\
&\qquad + \sum_{j,k=1}^{\Lambda} \frac{\dt^2 f_j(t+\dt) f_k(t+\dt)}{2} (-i \h_j) (-i \h_k)  + \order{\dt^3} \\
=:&\ \mathcal{U}_{\text{Taylor}}(t+\dt,t)+ \order{\dt^3}.
\end{aligned}
\end{align}
Since $\h_j$ are Hermitian and unitary, it is clear that $-i \h_j$ are also unitary, and the above equation approximates the unitary evolution $\exactU{t+\dt}{t}$ in the form of LCU. This expansion can also been obtained by Dyson's expansion, with further approximation of the integration at the time $t+\dt$. 

For simplicity, let us consider the case where $f_j \ge 0$ and $\dt$ is small enough such that $\dt f_j(t+\dt) - \frac{\dt^2}{2} f_j'(t+\dt) \ge 0$. Then the success probability for constructing the gate using LCU directly \cite{berry_simulating_2015} is $\nicefrac{1}{\mathsf{s}^2}$ where $\mathsf{s}$ is simply the summation of coefficients
\begin{align*}
\mathsf{s} &=  1 + \sum_{j=1}^{\Lambda} \big(\dt f_j(t+\dt) -\frac{\dt^2}{2} f_j'(t+\dt)\big)  + \sum_{j,k=1}^{\Lambda} \frac{\dt^2 f_j(t+\dt) f_k(t+\dt)}{2} \\
&= e^{\sum_{j=1}^{\Lambda} \int_{t}^{t+\dt} f_j(s)\ \dd s} + \order{\dt^3}. 
\end{align*}

{\noindent \emph{The challenge of vanishing probability.}} Suppose we consider the time-interval $[\mathsf{T}_i, \mathsf{T}_f]$, by discretizing the time into $N$ subintervals $\mathsf{T}_i = t_0 <  t_1 < \cdots t_N = \mathsf{T}_f$, and along multiple iterations, one may prepare the state proportional to 
\begin{align*}
&\ \ \mathcal{U}_{\text{Taylor}}(t_N, t_{N-1}) \mathcal{U}_{\text{Taylor}}(t_{N-1}, t_{N-2}) \cdots \mathcal{U}_{\text{Taylor}}(t_1, t_{0}) \ket{\psi} \\
\approx &\ \ \exactU{\mathsf{T}_f}{\mathsf{T}_i} \ket{\psi} + \order{\dt^2},
\end{align*}
with success probability asymptotically
\begin{align*}
\exp\Big(-2\sum_{j=1}^{\Lambda} \int_{\mathsf{T}_i}^{\mathsf{T}_f} f_j(s)\ \dd s\Big).
\end{align*}
This probability vanishes exponentially fast with respect to $\mathsf{T}_f - \mathsf{T}_i$ which is also well discussed in various literatures e.g., \cite{fang_time-marching_2023}.\\

{\noindent \emph{Application of oblivious amplitude amplification}.}
Suppose that we consider the simulation on the time interval $[0,T]$ where $T$ is large (or equivalently the magnitude of $f_j$ is large while $T=\order{1}$), we can always find some special time grid pints $\mathsf{T}_1 \equiv 0,\ \mathsf{T}_1, \cdots, \mathsf{T}_{\beta}$ such that 
\begin{align*}
e^{-2\sum_{j=1}^{\Lambda} \int_{\mathsf{T}_{k}}^{\mathsf{T}_{k+1}} f_j(s)\ \dd s} = \nicefrac{1}{4}.
\end{align*}
For simplicity, let us suppose that $\mathsf{T}_{\beta} = T$.
Then with at most three walk operators (constructed via controlled-$\h_j^{(\dagger)}$ gates) and some reflection operators, one can boost the probability of the above LCU-Taylor scheme to one for the time interval $[\mathsf{T}_k, \mathsf{T}_{k+1}]$ where $0\le k < \beta-1$. With the relation
\begin{align*}
\sum_{j=1}^{\Lambda} \int_{0}^{T} f_j(s)\ \dd s = \beta \ln(2),
\end{align*}
the number of times needed to apply oblivious amplitude amplification is $\order{T}$.

To obtain a heuristic estimate of the resources, within each sub-interval $[\mathsf{T}_k, \mathsf{T}_{k+1}]$, the error from the LCU-Taylor should be confined to $\order{\eps/T}$ and therefore, the time-interval $\dt = \order{\sqrt{\eps/T}}$. For each construction of LCU in \eqref{eqn::taylor_LCU}, the application of controlled-$\h_j$ gates (as well as controlled-$\h_j^\dagger$) scales at most like $\order{\Lambda^2}$ and therefore, the total gate complexity scales like $\order{\Lambda^2 \frac{T}{\dt}} = \order{\Lambda^2 T \sqrt{T/\eps}}$. If one keeps  Taylor's expansion up to the order $p$ (without further optimization), then the scaling is expected to be at most $\order{T^{1+\frac{1}{p}} \Lambda^p \eps^{-1/p}}$  in a way similar to product formulas.

The above arguments suggest that one can avoid the vanishing probability problem and there is no need to encode the whole time-dependent Hamiltonian into a single large unitary matrix as in \eqref{eqn::HAM}. The above discussion illustrates how time-dependent Hamiltonian simulation based on the LCU-Taylor \emph{could be achieved} using \SH{}'s clock and with a designing principle similar to time-independent LCU-Taylor \cite{berry_simulating_2015}.

\section{Numerical experiments}
\label{sec::example}

In this section, we will validate the practical efficacy of product-based formula derived based on \SH{}'s clock in handling the time-dependence. We consider two \DAS{} problems  (adiabatic Grover's algorithm and adiabatic Google's PageRank algorithm). An additional experiment will be  presented in \appref{Ising} for time-dependent quantum Ising chain.
Numerical experiments suggest that the generalized high-order HDR algorithms can overall achieve the best performance or near best performance; for some parameter regions, it can achieve significantly improved performance than a Magnus-based algorithm recently proposed in \cite{ikeda_minimum_2023}. 
We use four different weights \enquote{FRS}, \enquote{FRO}, \enquote{Suz4}, and \enquote{Ost4} from time-independent product formulas (all being fourth-order schemes) \cite{ostmeyer_optimised_2023}; see \cite{ostmeyer_optimised_2023} or \appref{app::weights} for details. 
For these examples, we consider $H(t) = f_1(t) \h_1 + f_2(t) \h_2$ for $t\in [0,1]$, and assume access to $e^{-i \alpha \h_k}$ for any $\alpha\in \Real$ and $k=1,2$. The gate count refers to the number of such unitaries that we need. We remark that this is still a fair comparison since all algorithms considered below need to use such resources; a further decomposition into one-qubit and two-qubit gates are not considered herein for simplicity, as we only aim at comparing the relative performance of algorithms in handling the time-dependence.
As mentioned earlier in the introduction, broader benchmark tests are undoubtedly necessary for more {thorough} validation. However, the existing experiments have already demonstrated the potential of Sambe-Howland’s clock in aiding the development of efficient quantum algorithms.

\begin{figure}
\captionsetup[subfigure]{labelformat=empty}
\subfloat[]{
\includegraphics[width=0.95\textwidth]{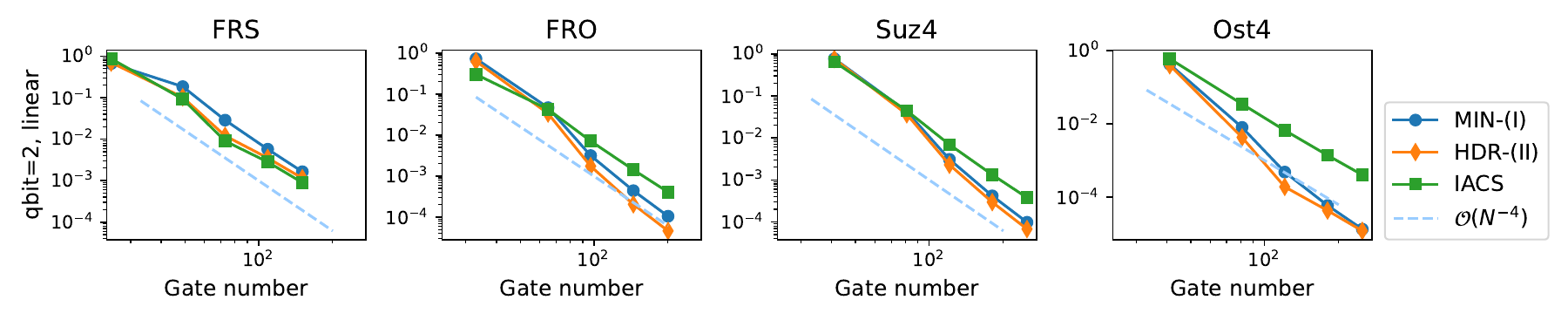}
} \\
\subfloat[]{
\includegraphics[width=0.95\textwidth]{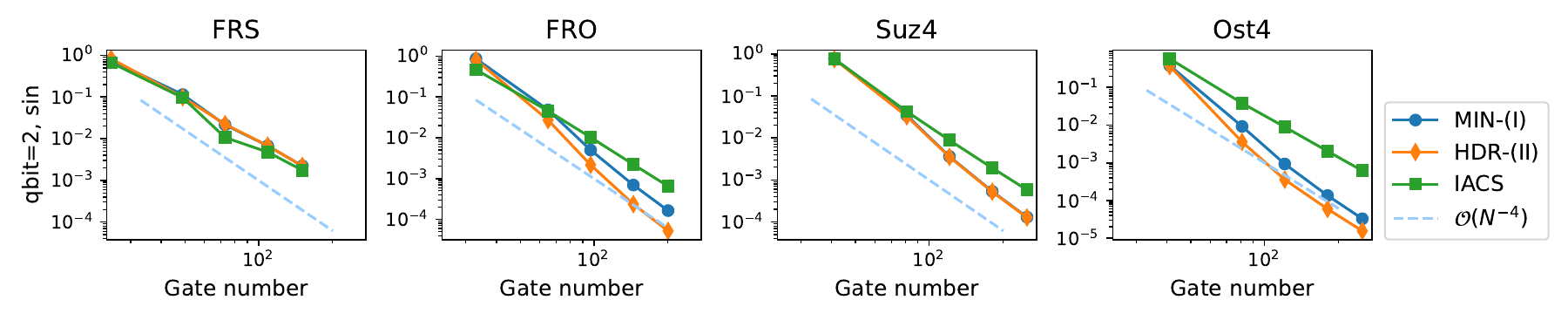}
} \\
\subfloat[]{
\includegraphics[width=0.95\textwidth]{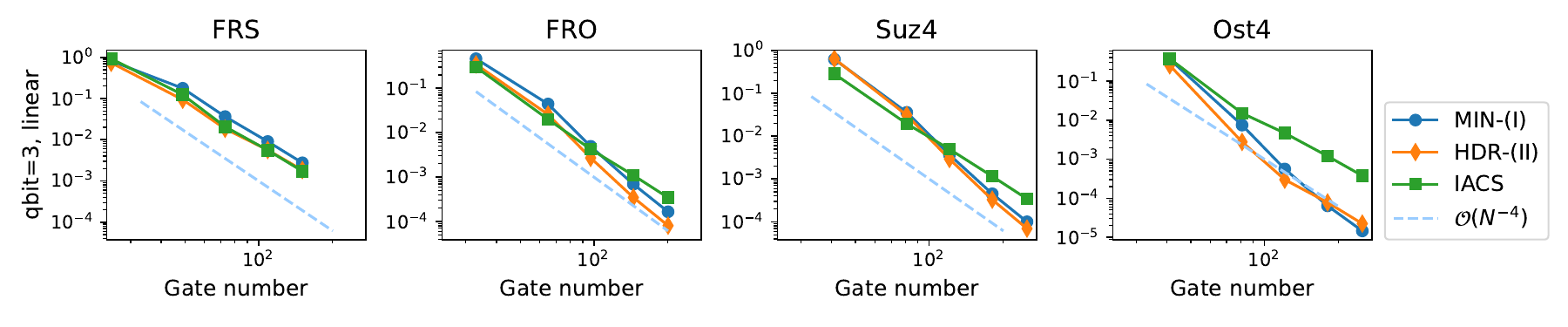}
} \\
\subfloat[]{
\includegraphics[width=0.95\textwidth]{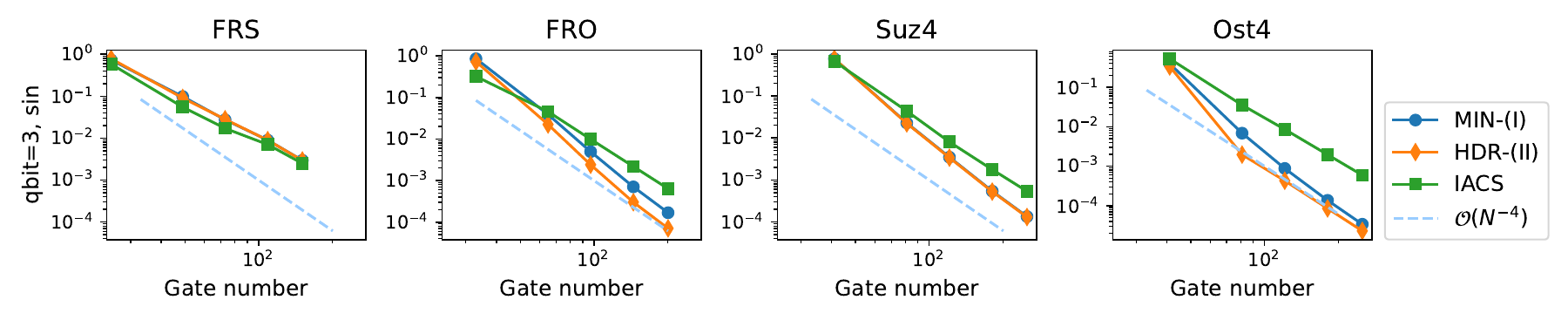}
}
\caption{Averaged simulation error (in trace distance) for  Grover's search algorithm with respect to the gate counts $N$ for various time-dependent schemes. \enquote{MIN-(I)} refers to the  scheme in \eqref{eqn::product_td} with the ordering $H_1, \hat{p}, H_2$. HDR refers to the scheme in \eqref{eqn::product_td_v2}. IACS refers to the scheme in \cite{ikeda_minimum_2023}; see \eqref{eqn::FRS_Magnus}.  The blue dashed line indicates the scaling $\order{N^{-4}}$ and is the same for all subplots.}
\label{fig::Grover}
\end{figure}

\begin{figure}
\captionsetup[subfigure]{labelformat=empty}
\subfloat[]{
\includegraphics[width=0.95\textwidth]{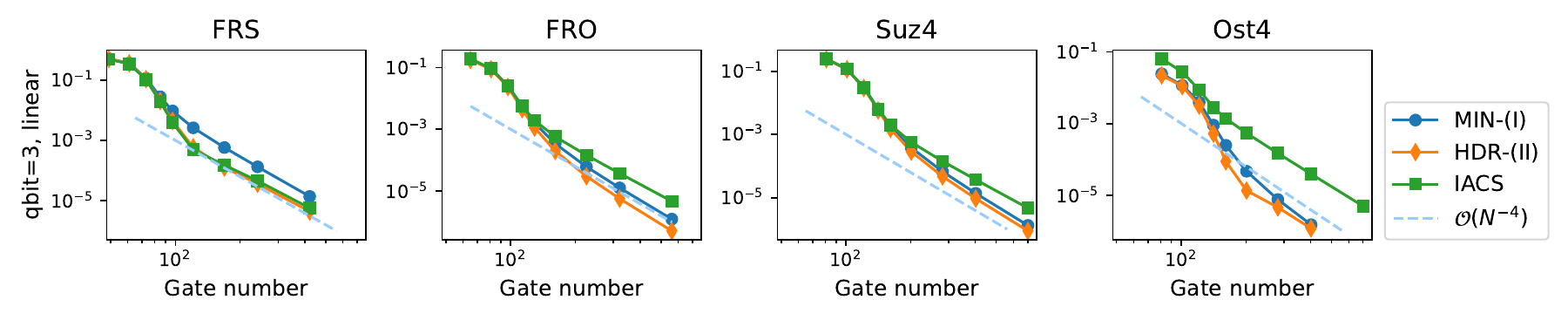}
}\\
\subfloat[]{
\includegraphics[width=0.95\textwidth]{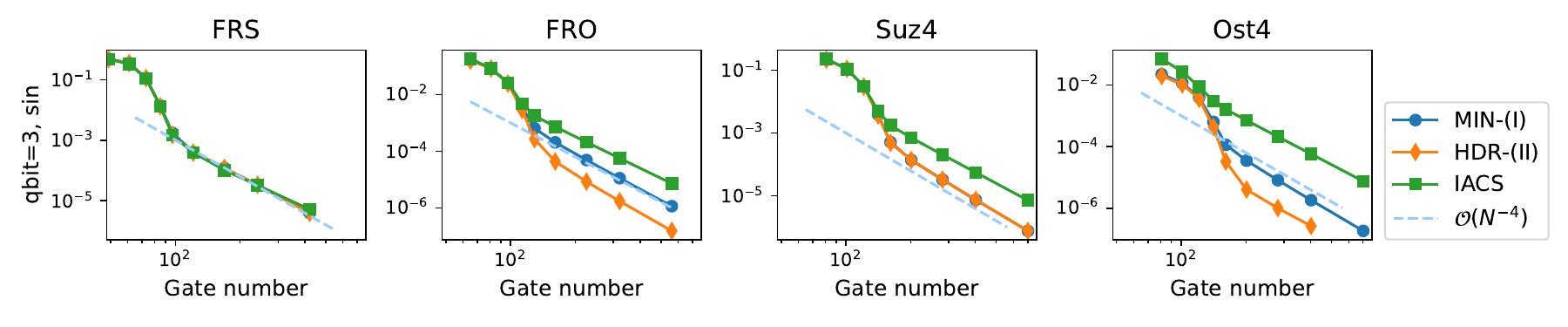}
}\\
\subfloat[]{
\includegraphics[width=0.95\textwidth]{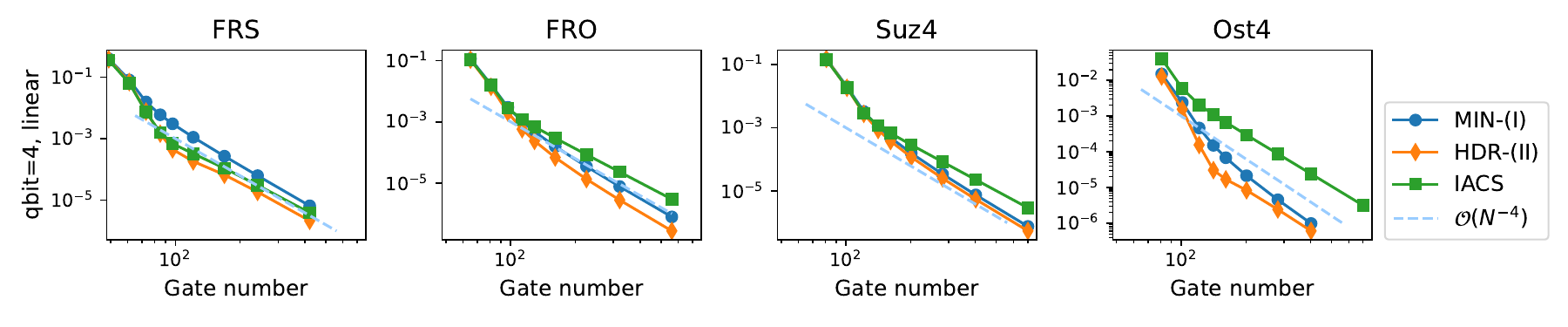}
}\\
\subfloat[]{
\includegraphics[width=0.95\textwidth]{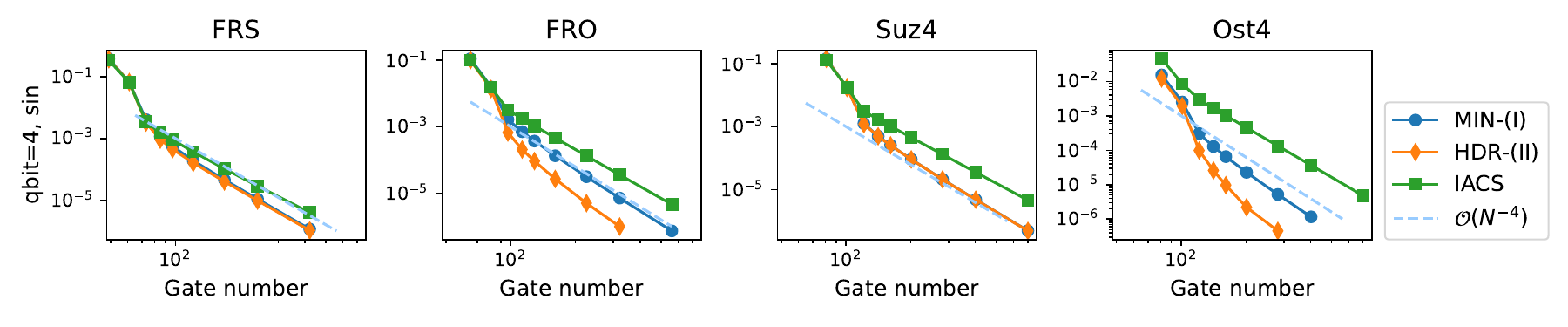}
}
\caption{Averaged simulation error in trace distance for adiabatic Google's PageRank algorithm with respect to the gate counts $N$ for various time-dependent schemes. \enquote{MIN-(I)} refers to the scheme \eqref{eqn::product_td} with the ordering $H_1, \hat{p}, H_2$. HDR refers to the scheme in \eqref{eqn::product_td_v2}. IACS refers to the scheme in \cite{ikeda_minimum_2023}; see also \eqref{eqn::FRS_Magnus}. The blue dashed line indicates the scaling $\order{N^{-4}}$ and is the same for all subplots.}
\label{fig::page_rank}
\end{figure}

\subsection{Adiabatic Grover's algorithm}
\label{sec::grover}

Suppose that we want to prepare the $n$-qubit state $\ket{\varphi} = \ket{\varphi_1}\ket{\varphi_2}\cdots\ket{\varphi_n}$ via the following Hamiltonians:
\begin{align*}
\h_1 = \id - \ket{+}^{\otimes n}\bra{+}^{\otimes n}, \qquad \h_2 = \id - \ket{\varphi}\bra{\varphi},
\end{align*}
where we choose $\ket{\varphi}$ as a product of $n$ local pure states in numerical experiments and $\ket{+}$ is the maximally entangled state $\ket{+} = \frac{1}{\sqrt{2}} \big(\ket{0} + \ket{1}\big)$.
We consider the following adiabatic path:
\begin{align*}
H(t) = T \Big(\big(1-f(t)\big) \h_1 + f(t) \h_2\Big), \qquad t\in [0,1],
\end{align*}
where we choose $T = 40$ in experiments, and the schedule is simply chosen as 
\begin{align}
\label{eqn::ft}
\begin{aligned}
f(t) =\left\{ \begin{aligned}t, & \qquad \text{ (linear)}; \\
\sin\big(\frac{\pi}{2} t\big), & \qquad \text{ (sin)}.\\
\end{aligned}\right.
\end{aligned}
\end{align}

The averaged simulation error in trace distance for  Grover's search algorithm (without considering the adiabatic approximation error) is visualized in \figref{fig::Grover}. The error has been averaged over $17, 16, 13, 5$ randomly generated target states $\ket{\varphi}$ such that $\abs{\bra{\varphi(1)}\ket{\varphi}}^2 \ge 99\%$, where $\varphi(1)$ is the exact state generated by the digital adiabatic simulation. As can be seen in \figref{fig::Grover}, the IACS (a Magnus-motivated product formula) may perform better if one uses a less efficient weight scheme like FRS. However, it does not improve much using a more efficient weight scheme. This can be possibly explained as follows: since the total error of Magnus-based product formula comes from both the Magnus expansion and the additional application of the product formula, a more efficient weight scheme can provide little help with reducing the error of Magnus' expansion itself. However, this is likely not an issue for a pure product-based formula (using \SH{}'s clock); with optimized time-independent product formula, one can generally expect a more efficient time-dependent product formula (as demonstrated in this example). We can observe that the HDR algorithm generally provides the best performance in most parameter regions, and it can offer roughly an order of magnitude improvement compared with IACS scheme using \enquote{Ost4} weights.

\subsection{Adiabatic Google's PageRank algorithm}

The second example is the adiabatic Google's PageRank algorithm.
Suppose we have a directed graph of nodes (e.g., web pages) with adjacency matrix $A$ (whose dimension is $2^n$ for simplicity), the transition probability for the random walk on the graph is defined as 
\begin{align*}
P(i,j) = \left\{
\begin{aligned}
\nicefrac{1}{d(i)}, &\ \  \text{ if } d(i) > 0\ \text{ and }  (i,j) \text{ is an edge};\\
0, &\ \  \text{ if } d(i) > 0\ \text{ and }  (i,j) \text{ is not an edge};\\
\nicefrac{1}{2^n}, &\ \  \text{ if } d(i)  = 0;
\end{aligned}\right.
\end{align*}
where $d(i)$ is the out degree for the node $i = 0, 1, \cdots, 2^{n}-1$.
The Google matrix $G$ is defined as $
G = \alpha P^{T} + (1-\alpha) E
$
where $\alpha = 0.85$ typically and $E$ is a matrix filled with elements ${1}/{2^n}$ \cite{garnerone_adiabatic_2012,albash_adiabatic_2018}. The stationary distribution of $G$, denoted as $\ket{\varphi}$, satisfies $G \ket{\varphi} = \ket{\varphi}$, and it characterizes the importance of these nodes. Similar to previous examples, it can be solved adiabatically, e.g., using the following two Hamiltonians:
\begin{align*}
\h_1 = \id - \ket{+}^{\otimes n} \bra{+}^{\otimes n}, \qquad \h_2 = (\id - G)^\dagger (\id - G).
\end{align*}

We consider the same schedule as in Grover's case with time schedule chosen as in \eqref{eqn::ft} and time scale $T=40$ for experiments, and  $n = 3, 4$ (namely, $8$ nodes and $16$ nodes respectively). The simulation error for various cases (without adiabatic approximation error) is summarized in \figref{fig::page_rank}. The error has been averaged over $15$ randomly generated graphs; we have validated the high fidelity for all such cases $\abs{\bra{\varphi(1)}\ket{\varphi}}^2 \ge 99\%$, where $\ket{\varphi}$ is the ground state of $\h_2$ (namely, the page-rank vector) and $\varphi(1)$ is the approximated state from the adiabatic simulation.
In \figref{fig::page_rank}, one can observe that the generalized HDR algorithm consistently has comparable performance compared with IACS scheme, and it outperforms significantly when equipped with a very efficient time-independent scheme \enquote{Ost4}. Across a variety of precision-level and different weights (from time-independent schemes), the 4th-order HDR algorithm \eqref{eqn::product_td_v2} seems to be the most reliable and efficient one in general.

\section{Summary}

{In this work, we have used \enquote{\SH{}'s continuous clock} to develop quantum algorithms for analog quantum computing and digital quantum computing, which complements recent advances using Sambe's space in \cite{mizuta_optimal_2023,mizuta_optimal_multiperiodic_2023,mizuta_explicit_2024}.}
In particular, for digital quantum computing, we have demonstrated  that \SH{}'s clock offers a way to unify the approach by Suzuki and the approach by \HDR{}.
Moreover, this extra continuous clock helps to resolve  several open questions in literature: (i) the question of developing high-order schemes for \HDR{}'s approach, (ii) the recently discussed problem of finding minimum gate implementation, and (iii) a conjecture in deriving multi-product formulas. We have also discussed how this overarching framework can be used for randomized quantum algorithms (qDrift in particular), and LCU-Taylor based algorithms.
Numerically, our experiments also  demonstrate the efficiency of newly developed high-order HDR algorithms for \DAS{}, which is at least comparable as some Magnus-based algorithms.

\section*{Acknowledgment}
The authors  are supported by NSFC grant No. 12341104, the Shanghai Jiao Tong University 2030 Initiative, the Shanghai Science and Technology Innovation
 Action Plan (24LZ1401200), and the Fundamental Research Funds for the Central Universities. YC is sponsored by Shanghai Pujiang Program (23PJ1404600).
SJ was
 also partially supported by the NSFC grant No. 12426637. NL also acknowledges funding from the Science and Technology Commission of Shanghai Municipality (STCSM) grant no. 24LZ1401200
 (21JC1402900) and the  NSFC grant No.12471411.

\bibliography{SH_Clock_AQC_Ref.bib}

\begin{thebibliography}{75}%
\makeatletter
\providecommand \@ifxundefined [1]{%
 \@ifx{#1\undefined}
}%
\providecommand \@ifnum [1]{%
 \ifnum #1\expandafter \@firstoftwo
 \else \expandafter \@secondoftwo
 \fi
}%
\providecommand \@ifx [1]{%
 \ifx #1\expandafter \@firstoftwo
 \else \expandafter \@secondoftwo
 \fi
}%
\providecommand \natexlab [1]{#1}%
\providecommand \enquote  [1]{``#1''}%
\providecommand \bibnamefont  [1]{#1}%
\providecommand \bibfnamefont [1]{#1}%
\providecommand \citenamefont [1]{#1}%
\providecommand \href@noop [0]{\@secondoftwo}%
\providecommand \href [0]{\begingroup \@sanitize@url \@href}%
\providecommand \@href[1]{\@@startlink{#1}\@@href}%
\providecommand \@@href[1]{\endgroup#1\@@endlink}%
\providecommand \@sanitize@url [0]{\catcode `\\12\catcode `\$12\catcode
  `\&12\catcode `\#12\catcode `\^12\catcode `\_12\catcode `\%12\relax}%
\providecommand \@@startlink[1]{}%
\providecommand \@@endlink[0]{}%
\providecommand \url  [0]{\begingroup\@sanitize@url \@url }%
\providecommand \@url [1]{\endgroup\@href {#1}{\urlprefix }}%
\providecommand \urlprefix  [0]{URL }%
\providecommand \Eprint [0]{\href }%
\providecommand \doibase [0]{https://doi.org/}%
\providecommand \selectlanguage [0]{\@gobble}%
\providecommand \bibinfo  [0]{\@secondoftwo}%
\providecommand \bibfield  [0]{\@secondoftwo}%
\providecommand \translation [1]{[#1]}%
\providecommand \BibitemOpen [0]{}%
\providecommand \bibitemStop [0]{}%
\providecommand \bibitemNoStop [0]{.\EOS\space}%
\providecommand \EOS [0]{\spacefactor3000\relax}%
\providecommand \BibitemShut  [1]{\csname bibitem#1\endcsname}%
\let\auto@bib@innerbib\@empty
%</preamble>
\bibitem [{\citenamefont {Feynman}(1982)}]{feynman_simulating_1982}%
  \BibitemOpen
  \bibfield  {author} {\bibinfo {author} {\bibfnamefont {R.~P.}\ \bibnamefont
  {Feynman}},\ }\bibfield  {title} {\bibinfo {title} {Simulating physics with
  computers},\ }\href {https://doi.org/10.1007/BF02650179} {\bibfield
  {journal} {\bibinfo  {journal} {International Journal of Theoretical
  Physics}\ }\textbf {\bibinfo {volume} {21}},\ \bibinfo {pages} {467}
  (\bibinfo {year} {1982})}\BibitemShut {NoStop}%
\bibitem [{\citenamefont {Lloyd}(1996)}]{lloyd_universal_1996}%
  \BibitemOpen
  \bibfield  {author} {\bibinfo {author} {\bibfnamefont {S.}~\bibnamefont
  {Lloyd}},\ }\bibfield  {title} {\bibinfo {title} {Universal {Quantum}
  {Simulators}},\ }\href {https://doi.org/10.1126/science.273.5278.1073}
  {\bibfield  {journal} {\bibinfo  {journal} {Science}\ }\textbf {\bibinfo
  {volume} {273}},\ \bibinfo {pages} {1073} (\bibinfo {year}
  {1996})}\BibitemShut {NoStop}%
\bibitem [{\citenamefont {Trotter}(1959)}]{trotter_product_1959}%
  \BibitemOpen
  \bibfield  {author} {\bibinfo {author} {\bibfnamefont {H.~F.}\ \bibnamefont
  {Trotter}},\ }\bibfield  {title} {\bibinfo {title} {On the {Product} of
  {Semi}-{Groups} of {Operators}},\ }\href {https://doi.org/10.2307/2033649}
  {\bibfield  {journal} {\bibinfo  {journal} {Proceedings of the American
  Mathematical Society}\ }\textbf {\bibinfo {volume} {10}},\ \bibinfo {pages}
  {545} (\bibinfo {year} {1959})}\BibitemShut {NoStop}%
\bibitem [{\citenamefont {Suzuki}(1976)}]{suzuki_generalized_1976}%
  \BibitemOpen
  \bibfield  {author} {\bibinfo {author} {\bibfnamefont {M.}~\bibnamefont
  {Suzuki}},\ }\bibfield  {title} {\bibinfo {title} {Generalized {Trotter}'s
  formula and systematic approximants of exponential operators and inner
  derivations with applications to many-body problems},\ }\href
  {https://doi.org/10.1007/BF01609348} {\bibfield  {journal} {\bibinfo
  {journal} {Communications in Mathematical Physics}\ }\textbf {\bibinfo
  {volume} {51}},\ \bibinfo {pages} {183} (\bibinfo {year} {1976})}\BibitemShut
  {NoStop}%
\bibitem [{\citenamefont {Suzuki}(1990)}]{suzuki_fractal_1990}%
  \BibitemOpen
  \bibfield  {author} {\bibinfo {author} {\bibfnamefont {M.}~\bibnamefont
  {Suzuki}},\ }\bibfield  {title} {\bibinfo {title} {Fractal decomposition of
  exponential operators with applications to many-body theories and {Monte}
  {Carlo} simulations},\ }\href {https://doi.org/10.1016/0375-9601(90)90962-N}
  {\bibfield  {journal} {\bibinfo  {journal} {Physics Letters A}\ }\textbf
  {\bibinfo {volume} {146}},\ \bibinfo {pages} {319} (\bibinfo {year}
  {1990})}\BibitemShut {NoStop}%
\bibitem [{\citenamefont {Suzuki}(1991)}]{suzuki_general_1991}%
  \BibitemOpen
  \bibfield  {author} {\bibinfo {author} {\bibfnamefont {M.}~\bibnamefont
  {Suzuki}},\ }\bibfield  {title} {\bibinfo {title} {General theory of fractal
  path integrals with applications to many‐body theories and statistical
  physics},\ }\href {https://doi.org/10.1063/1.529425} {\bibfield  {journal}
  {\bibinfo  {journal} {Journal of Mathematical Physics}\ }\textbf {\bibinfo
  {volume} {32}},\ \bibinfo {pages} {400} (\bibinfo {year} {1991})}\BibitemShut
  {NoStop}%
\bibitem [{\citenamefont {Chin}\ and\ \citenamefont
  {Geiser}(2011)}]{chin_multi_product_2011}%
  \BibitemOpen
  \bibfield  {author} {\bibinfo {author} {\bibfnamefont {S.~A.}\ \bibnamefont
  {Chin}}\ and\ \bibinfo {author} {\bibfnamefont {J.}~\bibnamefont {Geiser}},\
  }\bibfield  {title} {\bibinfo {title} {Multi-product operator splitting as a
  general method of solving autonomous and nonautonomous equations},\ }\href
  {https://doi.org/10.1093/imanum/drq022} {\bibfield  {journal} {\bibinfo
  {journal} {IMA Journal of Numerical Analysis}\ }\textbf {\bibinfo {volume}
  {31}},\ \bibinfo {pages} {1552} (\bibinfo {year} {2011})}\BibitemShut
  {NoStop}%
\bibitem [{\citenamefont {Childs}\ and\ \citenamefont
  {Wiebe}(2012)}]{childs_hamiltonian_2012}%
  \BibitemOpen
  \bibfield  {author} {\bibinfo {author} {\bibfnamefont {A.~M.}\ \bibnamefont
  {Childs}}\ and\ \bibinfo {author} {\bibfnamefont {N.}~\bibnamefont {Wiebe}},\
  }\bibfield  {title} {\bibinfo {title} {Hamiltonian simulation using linear
  combinations of unitary operations},\ }\href
  {https://doi.org/10.26421/QIC12.11-12-1} {\bibfield  {journal} {\bibinfo
  {journal} {Quantum Information and Computation}\ }\textbf {\bibinfo {volume}
  {12}},\ \bibinfo {pages} {901} (\bibinfo {year} {2012})}\BibitemShut
  {NoStop}%
\bibitem [{\citenamefont {Low}\ \emph {et~al.}(2019)\citenamefont {Low},
  \citenamefont {Kliuchnikov},\ and\ \citenamefont
  {Wiebe}}]{low_well_conditioned_2019}%
  \BibitemOpen
  \bibfield  {author} {\bibinfo {author} {\bibfnamefont {G.~H.}\ \bibnamefont
  {Low}}, \bibinfo {author} {\bibfnamefont {V.}~\bibnamefont {Kliuchnikov}},\
  and\ \bibinfo {author} {\bibfnamefont {N.}~\bibnamefont {Wiebe}},\
  }\href@noop {} {\bibinfo {title} {Well-conditioned multiproduct {Hamiltonian}
  simulation}} (\bibinfo {year} {2019}),\ \bibinfo {note}
  {arXiv:1907.11679}\BibitemShut {NoStop}%
\bibitem [{\citenamefont {Endo}\ \emph {et~al.}(2019)\citenamefont {Endo},
  \citenamefont {Zhao}, \citenamefont {Li}, \citenamefont {Benjamin},\ and\
  \citenamefont {Yuan}}]{endo_mitigating_2019}%
  \BibitemOpen
  \bibfield  {author} {\bibinfo {author} {\bibfnamefont {S.}~\bibnamefont
  {Endo}}, \bibinfo {author} {\bibfnamefont {Q.}~\bibnamefont {Zhao}}, \bibinfo
  {author} {\bibfnamefont {Y.}~\bibnamefont {Li}}, \bibinfo {author}
  {\bibfnamefont {S.}~\bibnamefont {Benjamin}},\ and\ \bibinfo {author}
  {\bibfnamefont {X.}~\bibnamefont {Yuan}},\ }\bibfield  {title} {\bibinfo
  {title} {Mitigating algorithmic errors in a {Hamiltonian} simulation},\
  }\href {https://doi.org/10.1103/PhysRevA.99.012334} {\bibfield  {journal}
  {\bibinfo  {journal} {Physical Review A}\ }\textbf {\bibinfo {volume} {99}},\
  \bibinfo {pages} {012334} (\bibinfo {year} {2019})}\BibitemShut {NoStop}%
\bibitem [{\citenamefont {Poulin}\ \emph {et~al.}(2011)\citenamefont {Poulin},
  \citenamefont {Qarry}, \citenamefont {Somma},\ and\ \citenamefont
  {Verstraete}}]{poulin_quantum_2011}%
  \BibitemOpen
  \bibfield  {author} {\bibinfo {author} {\bibfnamefont {D.}~\bibnamefont
  {Poulin}}, \bibinfo {author} {\bibfnamefont {A.}~\bibnamefont {Qarry}},
  \bibinfo {author} {\bibfnamefont {R.}~\bibnamefont {Somma}},\ and\ \bibinfo
  {author} {\bibfnamefont {F.}~\bibnamefont {Verstraete}},\ }\bibfield  {title}
  {\bibinfo {title} {Quantum {Simulation} of {Time}-{Dependent} {Hamiltonians}
  and the {Convenient} {Illusion} of {Hilbert} {Space}},\ }\href
  {https://doi.org/10.1103/PhysRevLett.106.170501} {\bibfield  {journal}
  {\bibinfo  {journal} {Physical Review Letters}\ }\textbf {\bibinfo {volume}
  {106}},\ \bibinfo {pages} {170501} (\bibinfo {year} {2011})}\BibitemShut
  {NoStop}%
\bibitem [{\citenamefont {Zhang}(2012)}]{zhang_randomized_2012}%
  \BibitemOpen
  \bibfield  {author} {\bibinfo {author} {\bibfnamefont {C.}~\bibnamefont
  {Zhang}},\ }\bibfield  {title} {\bibinfo {title} {Randomized {Algorithms} for
  {Hamiltonian} {Simulation}},\ }in\ \href
  {https://doi.org/10.1007/978-3-642-27440-4_42} {\emph {\bibinfo {booktitle}
  {Monte {Carlo} and {Quasi}-{Monte} {Carlo} {Methods} 2010}}},\ \bibinfo
  {editor} {edited by\ \bibinfo {editor} {\bibfnamefont {L.}~\bibnamefont
  {Plaskota}}\ and\ \bibinfo {editor} {\bibfnamefont {H.}~\bibnamefont
  {Woźniakowski}}}\ (\bibinfo  {publisher} {Springer},\ \bibinfo {address}
  {Berlin, Heidelberg},\ \bibinfo {year} {2012})\ pp.\ \bibinfo {pages}
  {709--719}\BibitemShut {NoStop}%
\bibitem [{\citenamefont {Childs}\ \emph {et~al.}(2019)\citenamefont {Childs},
  \citenamefont {Ostrander},\ and\ \citenamefont {Su}}]{childs_faster_2019}%
  \BibitemOpen
  \bibfield  {author} {\bibinfo {author} {\bibfnamefont {A.~M.}\ \bibnamefont
  {Childs}}, \bibinfo {author} {\bibfnamefont {A.}~\bibnamefont {Ostrander}},\
  and\ \bibinfo {author} {\bibfnamefont {Y.}~\bibnamefont {Su}},\ }\bibfield
  {title} {\bibinfo {title} {Faster quantum simulation by randomization},\
  }\href {https://doi.org/10.22331/q-2019-09-02-182} {\bibfield  {journal}
  {\bibinfo  {journal} {Quantum}\ }\textbf {\bibinfo {volume} {3}},\ \bibinfo
  {pages} {182} (\bibinfo {year} {2019})}\BibitemShut {NoStop}%
\bibitem [{\citenamefont {Campbell}(2019)}]{campbell_random_2019}%
  \BibitemOpen
  \bibfield  {author} {\bibinfo {author} {\bibfnamefont {E.}~\bibnamefont
  {Campbell}},\ }\bibfield  {title} {\bibinfo {title} {Random {Compiler} for
  {Fast} {Hamiltonian} {Simulation}},\ }\href
  {https://doi.org/10.1103/PhysRevLett.123.070503} {\bibfield  {journal}
  {\bibinfo  {journal} {Physical Review Letters}\ }\textbf {\bibinfo {volume}
  {123}},\ \bibinfo {pages} {070503} (\bibinfo {year} {2019})}\BibitemShut
  {NoStop}%
\bibitem [{\citenamefont {Berry}\ \emph {et~al.}(2015)\citenamefont {Berry},
  \citenamefont {Childs}, \citenamefont {Cleve}, \citenamefont {Kothari},\ and\
  \citenamefont {Somma}}]{berry_simulating_2015}%
  \BibitemOpen
  \bibfield  {author} {\bibinfo {author} {\bibfnamefont {D.~W.}\ \bibnamefont
  {Berry}}, \bibinfo {author} {\bibfnamefont {A.~M.}\ \bibnamefont {Childs}},
  \bibinfo {author} {\bibfnamefont {R.}~\bibnamefont {Cleve}}, \bibinfo
  {author} {\bibfnamefont {R.}~\bibnamefont {Kothari}},\ and\ \bibinfo {author}
  {\bibfnamefont {R.~D.}\ \bibnamefont {Somma}},\ }\bibfield  {title} {\bibinfo
  {title} {Simulating {Hamiltonian} {Dynamics} with a {Truncated} {Taylor}
  {Series}},\ }\href {https://doi.org/10.1103/PhysRevLett.114.090502}
  {\bibfield  {journal} {\bibinfo  {journal} {Physical Review Letters}\
  }\textbf {\bibinfo {volume} {114}},\ \bibinfo {pages} {090502} (\bibinfo
  {year} {2015})}\BibitemShut {NoStop}%
\bibitem [{\citenamefont {Low}\ and\ \citenamefont
  {Chuang}(2019)}]{low_hamiltonian_2019}%
  \BibitemOpen
  \bibfield  {author} {\bibinfo {author} {\bibfnamefont {G.~H.}\ \bibnamefont
  {Low}}\ and\ \bibinfo {author} {\bibfnamefont {I.~L.}\ \bibnamefont
  {Chuang}},\ }\bibfield  {title} {\bibinfo {title} {Hamiltonian {Simulation}
  by {Qubitization}},\ }\href {https://doi.org/10.22331/q-2019-07-12-163}
  {\bibfield  {journal} {\bibinfo  {journal} {Quantum}\ }\textbf {\bibinfo
  {volume} {3}},\ \bibinfo {pages} {163} (\bibinfo {year} {2019})}\BibitemShut
  {NoStop}%
\bibitem [{\citenamefont {Bosse}\ \emph {et~al.}(2025)\citenamefont {Bosse},
  \citenamefont {Childs}, \citenamefont {Derby}, \citenamefont {Gambetta},
  \citenamefont {Montanaro},\ and\ \citenamefont
  {Santos}}]{bosse_efficient_2024}%
  \BibitemOpen
  \bibfield  {author} {\bibinfo {author} {\bibfnamefont {J.~L.}\ \bibnamefont
  {Bosse}}, \bibinfo {author} {\bibfnamefont {A.~M.}\ \bibnamefont {Childs}},
  \bibinfo {author} {\bibfnamefont {C.}~\bibnamefont {Derby}}, \bibinfo
  {author} {\bibfnamefont {F.~M.}\ \bibnamefont {Gambetta}}, \bibinfo {author}
  {\bibfnamefont {A.}~\bibnamefont {Montanaro}},\ and\ \bibinfo {author}
  {\bibfnamefont {R.~A.}\ \bibnamefont {Santos}},\ }\bibfield  {title}
  {\bibinfo {title} {Efficient and practical {Hamiltonian} simulation from
  time-dependent product formulas},\ }\href
  {https://doi.org/10.1038/s41467-025-57580-5} {\bibfield  {journal} {\bibinfo
  {journal} {Nature Communications}\ }\textbf {\bibinfo {volume} {16}},\
  \bibinfo {pages} {2673} (\bibinfo {year} {2025})}\BibitemShut {NoStop}%
\bibitem [{\citenamefont {Sharma}\ and\ \citenamefont
  {Tran}(2024)}]{sharma_hamiltonian_2024}%
  \BibitemOpen
  \bibfield  {author} {\bibinfo {author} {\bibfnamefont {K.}~\bibnamefont
  {Sharma}}\ and\ \bibinfo {author} {\bibfnamefont {M.~C.}\ \bibnamefont
  {Tran}},\ }\href {https://doi.org/10.48550/arXiv.2404.02966} {\bibinfo
  {title} {Hamiltonian {Simulation} in the {Interaction} {Picture} {Using} the
  {Magnus} {Expansion}}} (\bibinfo {year} {2024}),\ \bibinfo {note}
  {arXiv:2404.02966}\BibitemShut {NoStop}%
\bibitem [{\citenamefont {Suzuki}(1995)}]{suzuki_hybrid_1995}%
  \BibitemOpen
  \bibfield  {author} {\bibinfo {author} {\bibfnamefont {M.}~\bibnamefont
  {Suzuki}},\ }\bibfield  {title} {\bibinfo {title} {Hybrid exponential product
  formulas for unbounded operators with possible applications to {Monte}
  {Carlo} simulations},\ }\href {https://doi.org/10.1016/0375-9601(95)00266-6}
  {\bibfield  {journal} {\bibinfo  {journal} {Physics Letters A}\ }\textbf
  {\bibinfo {volume} {201}},\ \bibinfo {pages} {425} (\bibinfo {year}
  {1995})}\BibitemShut {NoStop}%
\bibitem [{\citenamefont {Huyghebaert}\ and\ \citenamefont
  {Raedt}(1990)}]{huyghebaert_product_1990}%
  \BibitemOpen
  \bibfield  {author} {\bibinfo {author} {\bibfnamefont {J.}~\bibnamefont
  {Huyghebaert}}\ and\ \bibinfo {author} {\bibfnamefont {H.~D.}\ \bibnamefont
  {Raedt}},\ }\bibfield  {title} {\bibinfo {title} {Product formula methods for
  time-dependent {Schrodinger} problems},\ }\href
  {https://doi.org/10.1088/0305-4470/23/24/019} {\bibfield  {journal} {\bibinfo
   {journal} {J. Phys. A: Math. Gen.}\ }\textbf {\bibinfo {volume} {23}},\
  \bibinfo {pages} {5777} (\bibinfo {year} {1990})}\BibitemShut {NoStop}%
\bibitem [{\citenamefont {Suzuki}(1993)}]{suzuki_general_1993}%
  \BibitemOpen
  \bibfield  {author} {\bibinfo {author} {\bibfnamefont {M.}~\bibnamefont
  {Suzuki}},\ }\bibfield  {title} {\bibinfo {title} {General {Decomposition}
  {Theory} of {Ordered} {Exponentials}},\ }\href
  {https://doi.org/10.2183/pjab.69.161} {\bibfield  {journal} {\bibinfo
  {journal} {Proceedings of the Japan Academy, Series B}\ }\textbf {\bibinfo
  {volume} {69}},\ \bibinfo {pages} {161} (\bibinfo {year} {1993})}\BibitemShut
  {NoStop}%
\bibitem [{\citenamefont {Blanes}\ \emph {et~al.}(2009)\citenamefont {Blanes},
  \citenamefont {Casas}, \citenamefont {Oteo},\ and\ \citenamefont
  {Ros}}]{blanes_magnus_2009}%
  \BibitemOpen
  \bibfield  {author} {\bibinfo {author} {\bibfnamefont {S.}~\bibnamefont
  {Blanes}}, \bibinfo {author} {\bibfnamefont {F.}~\bibnamefont {Casas}},
  \bibinfo {author} {\bibfnamefont {J.~A.}\ \bibnamefont {Oteo}},\ and\
  \bibinfo {author} {\bibfnamefont {J.}~\bibnamefont {Ros}},\ }\bibfield
  {title} {\bibinfo {title} {The {Magnus} expansion and some of its
  applications},\ }\href {https://doi.org/10.1016/j.physrep.2008.11.001}
  {\bibfield  {journal} {\bibinfo  {journal} {Physics Reports}\ }\textbf
  {\bibinfo {volume} {470}},\ \bibinfo {pages} {151} (\bibinfo {year}
  {2009})}\BibitemShut {NoStop}%
\bibitem [{\citenamefont {Watkins}\ \emph {et~al.}(2024)\citenamefont
  {Watkins}, \citenamefont {Wiebe}, \citenamefont {Roggero},\ and\
  \citenamefont {Lee}}]{watkins_time_2024}%
  \BibitemOpen
  \bibfield  {author} {\bibinfo {author} {\bibfnamefont {J.}~\bibnamefont
  {Watkins}}, \bibinfo {author} {\bibfnamefont {N.}~\bibnamefont {Wiebe}},
  \bibinfo {author} {\bibfnamefont {A.}~\bibnamefont {Roggero}},\ and\ \bibinfo
  {author} {\bibfnamefont {D.}~\bibnamefont {Lee}},\ }\bibfield  {title}
  {\bibinfo {title} {Time-{Dependent} {Hamiltonian} {Simulation} {Using}
  {Discrete}-{Clock} {Constructions}},\ }\href
  {https://doi.org/10.1103/PRXQuantum.5.040316} {\bibfield  {journal} {\bibinfo
   {journal} {PRX Quantum}\ }\textbf {\bibinfo {volume} {5}},\ \bibinfo {pages}
  {040316} (\bibinfo {year} {2024})}\BibitemShut {NoStop}%
\bibitem [{\citenamefont {Mizuta}\ and\ \citenamefont
  {Fujii}(2023)}]{mizuta_optimal_2023}%
  \BibitemOpen
  \bibfield  {author} {\bibinfo {author} {\bibfnamefont {K.}~\bibnamefont
  {Mizuta}}\ and\ \bibinfo {author} {\bibfnamefont {K.}~\bibnamefont {Fujii}},\
  }\bibfield  {title} {\bibinfo {title} {Optimal {Hamiltonian} simulation for
  time-periodic systems},\ }\href {https://doi.org/10.22331/q-2023-03-28-962}
  {\bibfield  {journal} {\bibinfo  {journal} {Quantum}\ }\textbf {\bibinfo
  {volume} {7}},\ \bibinfo {pages} {962} (\bibinfo {year} {2023})}\BibitemShut
  {NoStop}%
\bibitem [{\citenamefont {Mizuta}(2023)}]{mizuta_optimal_multiperiodic_2023}%
  \BibitemOpen
  \bibfield  {author} {\bibinfo {author} {\bibfnamefont {K.}~\bibnamefont
  {Mizuta}},\ }\bibfield  {title} {\bibinfo {title} {Optimal and nearly optimal
  simulation of multiperiodic time-dependent {Hamiltonians}},\ }\href
  {https://doi.org/10.1103/PhysRevResearch.5.033067} {\bibfield  {journal}
  {\bibinfo  {journal} {Physical Review Research}\ }\textbf {\bibinfo {volume}
  {5}},\ \bibinfo {pages} {033067} (\bibinfo {year} {2023})}\BibitemShut
  {NoStop}%
\bibitem [{\citenamefont {Berry}\ \emph {et~al.}(2020)\citenamefont {Berry},
  \citenamefont {Childs}, \citenamefont {Su}, \citenamefont {Wang},\ and\
  \citenamefont {Wiebe}}]{berry_time_dependent_2020}%
  \BibitemOpen
  \bibfield  {author} {\bibinfo {author} {\bibfnamefont {D.~W.}\ \bibnamefont
  {Berry}}, \bibinfo {author} {\bibfnamefont {A.~M.}\ \bibnamefont {Childs}},
  \bibinfo {author} {\bibfnamefont {Y.}~\bibnamefont {Su}}, \bibinfo {author}
  {\bibfnamefont {X.}~\bibnamefont {Wang}},\ and\ \bibinfo {author}
  {\bibfnamefont {N.}~\bibnamefont {Wiebe}},\ }\bibfield  {title} {\bibinfo
  {title} {Time-dependent {Hamiltonian} simulation with ${L}^1$-norm scaling},\
  }\href {https://doi.org/10.22331/q-2020-04-20-254} {\bibfield  {journal}
  {\bibinfo  {journal} {Quantum}\ }\textbf {\bibinfo {volume} {4}},\ \bibinfo
  {pages} {254} (\bibinfo {year} {2020})}\BibitemShut {NoStop}%
\bibitem [{\citenamefont {Ikeda}\ \emph {et~al.}(2023)\citenamefont {Ikeda},
  \citenamefont {Abrar}, \citenamefont {Chuang},\ and\ \citenamefont
  {Sugiura}}]{ikeda_minimum_2023}%
  \BibitemOpen
  \bibfield  {author} {\bibinfo {author} {\bibfnamefont {T.~N.}\ \bibnamefont
  {Ikeda}}, \bibinfo {author} {\bibfnamefont {A.}~\bibnamefont {Abrar}},
  \bibinfo {author} {\bibfnamefont {I.~L.}\ \bibnamefont {Chuang}},\ and\
  \bibinfo {author} {\bibfnamefont {S.}~\bibnamefont {Sugiura}},\ }\bibfield
  {title} {\bibinfo {title} {Minimum {Trotterization} {Formulas} for a
  {Time}-{Dependent} {Hamiltonian}},\ }\href
  {https://doi.org/10.22331/q-2023-11-06-1168} {\bibfield  {journal} {\bibinfo
  {journal} {Quantum}\ }\textbf {\bibinfo {volume} {7}},\ \bibinfo {pages}
  {1168} (\bibinfo {year} {2023})}\BibitemShut {NoStop}%
\bibitem [{\citenamefont {Albash}\ and\ \citenamefont
  {Lidar}(2018)}]{albash_adiabatic_2018}%
  \BibitemOpen
  \bibfield  {author} {\bibinfo {author} {\bibfnamefont {T.}~\bibnamefont
  {Albash}}\ and\ \bibinfo {author} {\bibfnamefont {D.~A.}\ \bibnamefont
  {Lidar}},\ }\bibfield  {title} {\bibinfo {title} {Adiabatic quantum
  computation},\ }\href {https://doi.org/10.1103/RevModPhys.90.015002}
  {\bibfield  {journal} {\bibinfo  {journal} {Rev. Mod. Phys.}\ }\textbf
  {\bibinfo {volume} {90}},\ \bibinfo {pages} {015002} (\bibinfo {year}
  {2018})}\BibitemShut {NoStop}%
\bibitem [{\citenamefont {Low}\ and\ \citenamefont
  {Wiebe}(2019)}]{low_hamiltonian_2019_paper_b}%
  \BibitemOpen
  \bibfield  {author} {\bibinfo {author} {\bibfnamefont {G.~H.}\ \bibnamefont
  {Low}}\ and\ \bibinfo {author} {\bibfnamefont {N.}~\bibnamefont {Wiebe}},\
  }\href {https://doi.org/10.48550/arXiv.1805.00675} {\bibinfo {title}
  {Hamiltonian {Simulation} in the {Interaction} {Picture}}} (\bibinfo {year}
  {2019}),\ \bibinfo {note} {arXiv:1805.00675}\BibitemShut {NoStop}%
\bibitem [{\citenamefont {Jin}\ \emph {et~al.}(2024)\citenamefont {Jin},
  \citenamefont {Liu},\ and\ \citenamefont {Yu}}]{jin2212quantum}%
  \BibitemOpen
  \bibfield  {author} {\bibinfo {author} {\bibfnamefont {S.}~\bibnamefont
  {Jin}}, \bibinfo {author} {\bibfnamefont {N.}~\bibnamefont {Liu}},\ and\
  \bibinfo {author} {\bibfnamefont {Y.}~\bibnamefont {Yu}},\ }\bibfield
  {title} {\bibinfo {title} {Quantum simulation of partial differential
  equations via {S}chr\"odingerization},\ }\href
  {https://doi.org/10.1103/PhysRevLett.133.230602} {\bibfield  {journal}
  {\bibinfo  {journal} {Phys. Rev. Lett.}\ }\textbf {\bibinfo {volume} {133}},\
  \bibinfo {pages} {230602} (\bibinfo {year} {2024})}\BibitemShut {NoStop}%
\bibitem [{\citenamefont {Jin}\ \emph {et~al.}(2023)\citenamefont {Jin},
  \citenamefont {Liu},\ and\ \citenamefont {Yu}}]{jin2022quantum}%
  \BibitemOpen
  \bibfield  {author} {\bibinfo {author} {\bibfnamefont {S.}~\bibnamefont
  {Jin}}, \bibinfo {author} {\bibfnamefont {N.}~\bibnamefont {Liu}},\ and\
  \bibinfo {author} {\bibfnamefont {Y.}~\bibnamefont {Yu}},\ }\bibfield
  {title} {\bibinfo {title} {Quantum simulation of partial differential
  equations: Applications and detailed analysis},\ }\href
  {https://doi.org/10.1103/PhysRevA.108.032603} {\bibfield  {journal} {\bibinfo
   {journal} {Physical Review A}\ }\textbf {\bibinfo {volume} {108}},\ \bibinfo
  {pages} {032603} (\bibinfo {year} {2023})}\BibitemShut {NoStop}%
\bibitem [{\citenamefont {Jin}\ and\ \citenamefont
  {Liu}(2024)}]{jin_analog_2024}%
  \BibitemOpen
  \bibfield  {author} {\bibinfo {author} {\bibfnamefont {S.}~\bibnamefont
  {Jin}}\ and\ \bibinfo {author} {\bibfnamefont {N.}~\bibnamefont {Liu}},\
  }\bibfield  {title} {\bibinfo {title} {Analog quantum simulation of partial
  differential equations},\ }\href {https://doi.org/10.1088/2058-9565/ad49cf}
  {\bibfield  {journal} {\bibinfo  {journal} {Quantum Science and Technology}\
  }\textbf {\bibinfo {volume} {9}},\ \bibinfo {pages} {035047} (\bibinfo {year}
  {2024})}\BibitemShut {NoStop}%
\bibitem [{\citenamefont {Cao}\ \emph {et~al.}(2025)\citenamefont {Cao},
  \citenamefont {Jin},\ and\ \citenamefont {Liu}}]{time_dilation_2023}%
  \BibitemOpen
  \bibfield  {author} {\bibinfo {author} {\bibfnamefont {Y.}~\bibnamefont
  {Cao}}, \bibinfo {author} {\bibfnamefont {S.}~\bibnamefont {Jin}},\ and\
  \bibinfo {author} {\bibfnamefont {N.}~\bibnamefont {Liu}},\ }\bibfield
  {title} {\bibinfo {title} {Quantum simulation for time-dependent
  {Hamiltonians}—with applications to non-autonomous ordinary and partial
  differential equations},\ }\href {https://doi.org/10.1088/1751-8121/adb3fe}
  {\bibfield  {journal} {\bibinfo  {journal} {Journal of Physics A:
  Mathematical and Theoretical}\ }\textbf {\bibinfo {volume} {58}},\ \bibinfo
  {pages} {155304} (\bibinfo {year} {2025})}\BibitemShut {NoStop}%
\bibitem [{\citenamefont {Berry}\ \emph {et~al.}(2007)\citenamefont {Berry},
  \citenamefont {Ahokas}, \citenamefont {Cleve},\ and\ \citenamefont
  {Sanders}}]{berry_efficient_2007}%
  \BibitemOpen
  \bibfield  {author} {\bibinfo {author} {\bibfnamefont {D.~W.}\ \bibnamefont
  {Berry}}, \bibinfo {author} {\bibfnamefont {G.}~\bibnamefont {Ahokas}},
  \bibinfo {author} {\bibfnamefont {R.}~\bibnamefont {Cleve}},\ and\ \bibinfo
  {author} {\bibfnamefont {B.~C.}\ \bibnamefont {Sanders}},\ }\bibfield
  {title} {\bibinfo {title} {Efficient {Quantum} {Algorithms} for {Simulating}
  {Sparse} {Hamiltonians}},\ }\href {https://doi.org/10.1007/s00220-006-0150-x}
  {\bibfield  {journal} {\bibinfo  {journal} {Communications in Mathematical
  Physics}\ }\textbf {\bibinfo {volume} {270}},\ \bibinfo {pages} {359}
  (\bibinfo {year} {2007})}\BibitemShut {NoStop}%
\bibitem [{\citenamefont {Sambe}(1973)}]{sambe_steady_1973}%
  \BibitemOpen
  \bibfield  {author} {\bibinfo {author} {\bibfnamefont {H.}~\bibnamefont
  {Sambe}},\ }\bibfield  {title} {\bibinfo {title} {Steady {States} and
  {Quasienergies} of a {Quantum}-{Mechanical} {System} in an {Oscillating}
  {Field}},\ }\href {https://doi.org/10.1103/PhysRevA.7.2203} {\bibfield
  {journal} {\bibinfo  {journal} {Physical Review A}\ }\textbf {\bibinfo
  {volume} {7}},\ \bibinfo {pages} {2203} (\bibinfo {year} {1973})}\BibitemShut
  {NoStop}%
\bibitem [{\citenamefont {Howland}(1974)}]{howland1974stationary}%
  \BibitemOpen
  \bibfield  {author} {\bibinfo {author} {\bibfnamefont {J.~S.}\ \bibnamefont
  {Howland}},\ }\bibfield  {title} {\bibinfo {title} {Stationary scattering
  theory for time-dependent {Hamiltonians}},\ }\href
  {https://doi.org/10.1007/BF01351346} {\bibfield  {journal} {\bibinfo
  {journal} {Mathematische Annalen}\ }\textbf {\bibinfo {volume} {207}},\
  \bibinfo {pages} {315} (\bibinfo {year} {1974})}\BibitemShut {NoStop}%
\bibitem [{\citenamefont {Peskin}\ and\ \citenamefont
  {Moiseyev}(1993)}]{peskin_solution_1993}%
  \BibitemOpen
  \bibfield  {author} {\bibinfo {author} {\bibfnamefont {U.}~\bibnamefont
  {Peskin}}\ and\ \bibinfo {author} {\bibfnamefont {N.}~\bibnamefont
  {Moiseyev}},\ }\bibfield  {title} {\bibinfo {title} {The solution of the
  time‐dependent {Schr{\"o}dinger} equation by the $(t,t')$ method: {Theory},
  computational algorithm and applications},\ }\href
  {https://doi.org/10.1063/1.466058} {\bibfield  {journal} {\bibinfo  {journal}
  {The Journal of Chemical Physics}\ }\textbf {\bibinfo {volume} {99}},\
  \bibinfo {pages} {4590} (\bibinfo {year} {1993})}\BibitemShut {NoStop}%
\bibitem [{\citenamefont {Burgarth}\ \emph {et~al.}(2023)\citenamefont
  {Burgarth}, \citenamefont {Facchi},\ and\ \citenamefont
  {Hillier}}]{burgarth_control_2023}%
  \BibitemOpen
  \bibfield  {author} {\bibinfo {author} {\bibfnamefont {D.}~\bibnamefont
  {Burgarth}}, \bibinfo {author} {\bibfnamefont {P.}~\bibnamefont {Facchi}},\
  and\ \bibinfo {author} {\bibfnamefont {R.}~\bibnamefont {Hillier}},\
  }\bibfield  {title} {\bibinfo {title} {Control of quantum noise: On the role
  of dilations},\ }\href {https://doi.org/10.1007/s00023-022-01211-y}
  {\bibfield  {journal} {\bibinfo  {journal} {Annales Henri Poincaré}\
  }\textbf {\bibinfo {volume} {24}},\ \bibinfo {pages} {325} (\bibinfo {year}
  {2023})}\BibitemShut {NoStop}%
\bibitem [{\citenamefont {Reed}\ and\ \citenamefont
  {Simon}(1975)}]{reed1975methods}%
  \BibitemOpen
  \bibfield  {author} {\bibinfo {author} {\bibfnamefont {M.}~\bibnamefont
  {Reed}}\ and\ \bibinfo {author} {\bibfnamefont {B.}~\bibnamefont {Simon}},\
  }\href@noop {} {\emph {\bibinfo {title} {Methods of mathematical physics.
  {Fourier} analysis, self-adjointness}}}\ (\bibinfo  {publisher} {Academic
  Press},\ \bibinfo {year} {1975})\BibitemShut {NoStop}%
\bibitem [{\citenamefont {McClean}\ \emph {et~al.}(2013)\citenamefont
  {McClean}, \citenamefont {Parkhill},\ and\ \citenamefont
  {Aspuru-Guzik}}]{mcclean_feynmans_2013}%
  \BibitemOpen
  \bibfield  {author} {\bibinfo {author} {\bibfnamefont {J.~R.}\ \bibnamefont
  {McClean}}, \bibinfo {author} {\bibfnamefont {J.~A.}\ \bibnamefont
  {Parkhill}},\ and\ \bibinfo {author} {\bibfnamefont {A.}~\bibnamefont
  {Aspuru-Guzik}},\ }\bibfield  {title} {\bibinfo {title} {Feynman’s clock, a
  new variational principle, and parallel-in-time quantum dynamics},\ }\href
  {https://doi.org/10.1073/pnas.1308069110} {\bibfield  {journal} {\bibinfo
  {journal} {Proceedings of the National Academy of Sciences}\ }\textbf
  {\bibinfo {volume} {110}},\ \bibinfo {pages} {E3901} (\bibinfo {year}
  {2013})}\BibitemShut {NoStop}%
\bibitem [{\citenamefont {Miessen}\ \emph {et~al.}(2023)\citenamefont
  {Miessen}, \citenamefont {Ollitrault}, \citenamefont {Tacchino},\ and\
  \citenamefont {Tavernelli}}]{miessen_quantum_2023}%
  \BibitemOpen
  \bibfield  {author} {\bibinfo {author} {\bibfnamefont {A.}~\bibnamefont
  {Miessen}}, \bibinfo {author} {\bibfnamefont {P.~J.}\ \bibnamefont
  {Ollitrault}}, \bibinfo {author} {\bibfnamefont {F.}~\bibnamefont
  {Tacchino}},\ and\ \bibinfo {author} {\bibfnamefont {I.}~\bibnamefont
  {Tavernelli}},\ }\bibfield  {title} {\bibinfo {title} {Quantum algorithms for
  quantum dynamics},\ }\href {https://doi.org/10.1038/s43588-022-00374-2}
  {\bibfield  {journal} {\bibinfo  {journal} {Nature Computational Science}\
  }\textbf {\bibinfo {volume} {3}},\ \bibinfo {pages} {25} (\bibinfo {year}
  {2023})}\BibitemShut {NoStop}%
\bibitem [{\citenamefont {De~Raedt}(1987)}]{de_raedt_product_1987}%
  \BibitemOpen
  \bibfield  {author} {\bibinfo {author} {\bibfnamefont {H.}~\bibnamefont
  {De~Raedt}},\ }\bibfield  {title} {\bibinfo {title} {Product formula
  algorithms for solving the time dependent {Schrödinger} equation},\ }\href
  {https://doi.org/10.1016/0167-7977(87)90002-5} {\bibfield  {journal}
  {\bibinfo  {journal} {Computer Physics Reports}\ }\textbf {\bibinfo {volume}
  {7}},\ \bibinfo {pages} {1} (\bibinfo {year} {1987})}\BibitemShut {NoStop}%
\bibitem [{\citenamefont {Childs}\ \emph {et~al.}(2021)\citenamefont {Childs},
  \citenamefont {Su}, \citenamefont {Tran}, \citenamefont {Wiebe},\ and\
  \citenamefont {Zhu}}]{childs_theory_2021}%
  \BibitemOpen
  \bibfield  {author} {\bibinfo {author} {\bibfnamefont {A.~M.}\ \bibnamefont
  {Childs}}, \bibinfo {author} {\bibfnamefont {Y.}~\bibnamefont {Su}}, \bibinfo
  {author} {\bibfnamefont {M.~C.}\ \bibnamefont {Tran}}, \bibinfo {author}
  {\bibfnamefont {N.}~\bibnamefont {Wiebe}},\ and\ \bibinfo {author}
  {\bibfnamefont {S.}~\bibnamefont {Zhu}},\ }\bibfield  {title} {\bibinfo
  {title} {Theory of {Trotter} {Error} with {Commutator} {Scaling}},\ }\href
  {https://doi.org/10.1103/PhysRevX.11.011020} {\bibfield  {journal} {\bibinfo
  {journal} {Physical Review X}\ }\textbf {\bibinfo {volume} {11}},\ \bibinfo
  {pages} {011020} (\bibinfo {year} {2021})}\BibitemShut {NoStop}%
\bibitem [{\citenamefont {Hatano}\ and\ \citenamefont
  {Suzuki}(2005)}]{hatano_finding_2005}%
  \BibitemOpen
  \bibfield  {author} {\bibinfo {author} {\bibfnamefont {N.}~\bibnamefont
  {Hatano}}\ and\ \bibinfo {author} {\bibfnamefont {M.}~\bibnamefont
  {Suzuki}},\ }\bibfield  {title} {\bibinfo {title} {Finding {Exponential}
  {Product} {Formulas} of {Higher} {Orders}},\ }in\ \href
  {https://doi.org/10.1007/11526216_2} {\emph {\bibinfo {booktitle} {Quantum
  {Annealing} and {Other} {Optimization} {Methods}}}},\ \bibinfo {editor}
  {edited by\ \bibinfo {editor} {\bibfnamefont {A.}~\bibnamefont {Das}}\ and\
  \bibinfo {editor} {\bibfnamefont {B.}~\bibnamefont {K.~Chakrabarti}}}\
  (\bibinfo  {publisher} {Springer},\ \bibinfo {address} {Berlin, Heidelberg},\
  \bibinfo {year} {2005})\ pp.\ \bibinfo {pages} {37--68}\BibitemShut {NoStop}%
\bibitem [{\citenamefont {Wiebe}\ \emph {et~al.}(2010)\citenamefont {Wiebe},
  \citenamefont {Berry}, \citenamefont {Høyer},\ and\ \citenamefont
  {Sanders}}]{wiebe_higher_2010}%
  \BibitemOpen
  \bibfield  {author} {\bibinfo {author} {\bibfnamefont {N.}~\bibnamefont
  {Wiebe}}, \bibinfo {author} {\bibfnamefont {D.}~\bibnamefont {Berry}},
  \bibinfo {author} {\bibfnamefont {P.}~\bibnamefont {Høyer}},\ and\ \bibinfo
  {author} {\bibfnamefont {B.~C.}\ \bibnamefont {Sanders}},\ }\bibfield
  {title} {\bibinfo {title} {Higher order decompositions of ordered operator
  exponentials},\ }\href {https://doi.org/10.1088/1751-8113/43/6/065203}
  {\bibfield  {journal} {\bibinfo  {journal} {J. Phys. A: Math. Theor.}\
  }\textbf {\bibinfo {volume} {43}},\ \bibinfo {pages} {065203} (\bibinfo
  {year} {2010})}\BibitemShut {NoStop}%
\bibitem [{\citenamefont {Wiebe}\ \emph {et~al.}(2011)\citenamefont {Wiebe},
  \citenamefont {Berry}, \citenamefont {Høyer},\ and\ \citenamefont
  {Sanders}}]{wiebe_simulating_2011}%
  \BibitemOpen
  \bibfield  {author} {\bibinfo {author} {\bibfnamefont {N.}~\bibnamefont
  {Wiebe}}, \bibinfo {author} {\bibfnamefont {D.~W.}\ \bibnamefont {Berry}},
  \bibinfo {author} {\bibfnamefont {P.}~\bibnamefont {Høyer}},\ and\ \bibinfo
  {author} {\bibfnamefont {B.~C.}\ \bibnamefont {Sanders}},\ }\bibfield
  {title} {\bibinfo {title} {Simulating quantum dynamics on a quantum
  computer},\ }\href {https://doi.org/10.1088/1751-8113/44/44/445308}
  {\bibfield  {journal} {\bibinfo  {journal} {Journal of Physics A:
  Mathematical and Theoretical}\ }\textbf {\bibinfo {volume} {44}},\ \bibinfo
  {pages} {445308} (\bibinfo {year} {2011})}\BibitemShut {NoStop}%
\bibitem [{\citenamefont {Zhao}\ \emph {et~al.}(2024)\citenamefont {Zhao},
  \citenamefont {Bukov}, \citenamefont {Heyl},\ and\ \citenamefont
  {Moessner}}]{zhao_adaptive_2024}%
  \BibitemOpen
  \bibfield  {author} {\bibinfo {author} {\bibfnamefont {H.}~\bibnamefont
  {Zhao}}, \bibinfo {author} {\bibfnamefont {M.}~\bibnamefont {Bukov}},
  \bibinfo {author} {\bibfnamefont {M.}~\bibnamefont {Heyl}},\ and\ \bibinfo
  {author} {\bibfnamefont {R.}~\bibnamefont {Moessner}},\ }\bibfield  {title}
  {\bibinfo {title} {Adaptive {Trotterization} for {Time}-{Dependent}
  {Hamiltonian} {Quantum} {Dynamics} {Using} {Piecewise} {Conservation}
  {Laws}},\ }\href {https://doi.org/10.1103/PhysRevLett.133.010603} {\bibfield
  {journal} {\bibinfo  {journal} {Physical Review Letters}\ }\textbf {\bibinfo
  {volume} {133}},\ \bibinfo {pages} {010603} (\bibinfo {year}
  {2024})}\BibitemShut {NoStop}%
\bibitem [{\citenamefont {Aftab}\ \emph {et~al.}(2024)\citenamefont {Aftab},
  \citenamefont {An},\ and\ \citenamefont
  {Trivisa}}]{aftab_multi_product_2024}%
  \BibitemOpen
  \bibfield  {author} {\bibinfo {author} {\bibfnamefont {J.}~\bibnamefont
  {Aftab}}, \bibinfo {author} {\bibfnamefont {D.}~\bibnamefont {An}},\ and\
  \bibinfo {author} {\bibfnamefont {K.}~\bibnamefont {Trivisa}},\ }\href
  {https://doi.org/10.48550/arXiv.2403.08922} {\bibinfo {title} {Multi-product
  {Hamiltonian} simulation with explicit commutator scaling}} (\bibinfo {year}
  {2024}),\ \bibinfo {note} {arXiv:2403.08922}\BibitemShut {NoStop}%
\bibitem [{\citenamefont {Geiser}(2011)}]{geiser_multi_product_2011}%
  \BibitemOpen
  \bibfield  {author} {\bibinfo {author} {\bibfnamefont {J.}~\bibnamefont
  {Geiser}},\ }\bibfield  {title} {\bibinfo {title} {Multi-{Product}
  {Expansion} with {Suzuki}'s {Method}: {Generalization}},\ }\href
  {https://doi.org/10.4208/nmtma.2011.m9010} {\bibfield  {journal} {\bibinfo
  {journal} {Numerical Mathematics: Theory, Methods and Applications}\ }\textbf
  {\bibinfo {volume} {4}},\ \bibinfo {pages} {68} (\bibinfo {year}
  {2011})}\BibitemShut {NoStop}%
\bibitem [{\citenamefont {Ouyang}\ \emph {et~al.}(2020)\citenamefont {Ouyang},
  \citenamefont {White},\ and\ \citenamefont
  {Campbell}}]{ouyang_compilation_2020}%
  \BibitemOpen
  \bibfield  {author} {\bibinfo {author} {\bibfnamefont {Y.}~\bibnamefont
  {Ouyang}}, \bibinfo {author} {\bibfnamefont {D.~R.}\ \bibnamefont {White}},\
  and\ \bibinfo {author} {\bibfnamefont {E.~T.}\ \bibnamefont {Campbell}},\
  }\bibfield  {title} {\bibinfo {title} {Compilation by stochastic
  {Hamiltonian} sparsification},\ }\href
  {https://doi.org/10.22331/q-2020-02-27-235} {\bibfield  {journal} {\bibinfo
  {journal} {Quantum}\ }\textbf {\bibinfo {volume} {4}},\ \bibinfo {pages}
  {235} (\bibinfo {year} {2020})}\BibitemShut {NoStop}%
\bibitem [{\citenamefont {Jin}\ and\ \citenamefont
  {Li}(2023)}]{jin_partially_2023}%
  \BibitemOpen
  \bibfield  {author} {\bibinfo {author} {\bibfnamefont {S.}~\bibnamefont
  {Jin}}\ and\ \bibinfo {author} {\bibfnamefont {X.}~\bibnamefont {Li}},\
  }\bibfield  {title} {\bibinfo {title} {A {Partially} {Random} {Trotter}
  {Algorithm} for {Quantum} {Hamiltonian} {Simulations}},\ }\href
  {https://doi.org/10.1007/s42967-023-00336-z} {\bibfield  {journal} {\bibinfo
  {journal} {Communications on Applied Mathematics and Computation}\ }\textbf
  {\bibinfo {volume} {7}},\ \bibinfo {pages} {442–469} (\bibinfo {year}
  {2023})}\BibitemShut {NoStop}%
\bibitem [{\citenamefont {Hagan}\ and\ \citenamefont
  {Wiebe}(2023)}]{hagan_composite_2023}%
  \BibitemOpen
  \bibfield  {author} {\bibinfo {author} {\bibfnamefont {M.}~\bibnamefont
  {Hagan}}\ and\ \bibinfo {author} {\bibfnamefont {N.}~\bibnamefont {Wiebe}},\
  }\bibfield  {title} {\bibinfo {title} {Composite {Quantum} {Simulations}},\
  }\href {https://doi.org/10.22331/q-2023-11-14-1181} {\bibfield  {journal}
  {\bibinfo  {journal} {Quantum}\ }\textbf {\bibinfo {volume} {7}},\ \bibinfo
  {pages} {1181} (\bibinfo {year} {2023})}\BibitemShut {NoStop}%
\bibitem [{\citenamefont {Nakaji}\ \emph {et~al.}(2024)\citenamefont {Nakaji},
  \citenamefont {Bagherimehrab},\ and\ \citenamefont
  {Aspuru-Guzik}}]{nakaji_high_order_2024}%
  \BibitemOpen
  \bibfield  {author} {\bibinfo {author} {\bibfnamefont {K.}~\bibnamefont
  {Nakaji}}, \bibinfo {author} {\bibfnamefont {M.}~\bibnamefont
  {Bagherimehrab}},\ and\ \bibinfo {author} {\bibfnamefont {A.}~\bibnamefont
  {Aspuru-Guzik}},\ }\bibfield  {title} {\bibinfo {title} {High-{Order}
  {Randomized} {Compiler} for {Hamiltonian} {Simulation}},\ }\href
  {https://doi.org/10.1103/PRXQuantum.5.020330} {\bibfield  {journal} {\bibinfo
   {journal} {PRX Quantum}\ }\textbf {\bibinfo {volume} {5}},\ \bibinfo {pages}
  {020330} (\bibinfo {year} {2024})}\BibitemShut {NoStop}%
\bibitem [{\citenamefont {Kiss}\ \emph {et~al.}(2023)\citenamefont {Kiss},
  \citenamefont {Grossi},\ and\ \citenamefont
  {Roggero}}]{kiss_importance_2023}%
  \BibitemOpen
  \bibfield  {author} {\bibinfo {author} {\bibfnamefont {O.}~\bibnamefont
  {Kiss}}, \bibinfo {author} {\bibfnamefont {M.}~\bibnamefont {Grossi}},\ and\
  \bibinfo {author} {\bibfnamefont {A.}~\bibnamefont {Roggero}},\ }\bibfield
  {title} {\bibinfo {title} {Importance sampling for stochastic quantum
  simulations},\ }\href {https://doi.org/10.22331/q-2023-04-13-977} {\bibfield
  {journal} {\bibinfo  {journal} {Quantum}\ }\textbf {\bibinfo {volume} {7}},\
  \bibinfo {pages} {977} (\bibinfo {year} {2023})}\BibitemShut {NoStop}%
\bibitem [{\citenamefont {Childs}\ \emph {et~al.}(2018)\citenamefont {Childs},
  \citenamefont {Maslov}, \citenamefont {Nam}, \citenamefont {Ross},\ and\
  \citenamefont {Su}}]{childs_toward_2018}%
  \BibitemOpen
  \bibfield  {author} {\bibinfo {author} {\bibfnamefont {A.~M.}\ \bibnamefont
  {Childs}}, \bibinfo {author} {\bibfnamefont {D.}~\bibnamefont {Maslov}},
  \bibinfo {author} {\bibfnamefont {Y.}~\bibnamefont {Nam}}, \bibinfo {author}
  {\bibfnamefont {N.~J.}\ \bibnamefont {Ross}},\ and\ \bibinfo {author}
  {\bibfnamefont {Y.}~\bibnamefont {Su}},\ }\bibfield  {title} {\bibinfo
  {title} {Toward the first quantum simulation with quantum speedup},\ }\href
  {https://doi.org/10.1073/pnas.1801723115} {\bibfield  {journal} {\bibinfo
  {journal} {Proceedings of the National Academy of Sciences}\ }\textbf
  {\bibinfo {volume} {115}},\ \bibinfo {pages} {9456} (\bibinfo {year}
  {2018})}\BibitemShut {NoStop}%
\bibitem [{\citenamefont {Kieferová}\ \emph {et~al.}(2019)\citenamefont
  {Kieferová}, \citenamefont {Scherer},\ and\ \citenamefont
  {Berry}}]{kieferova_simulating_2019}%
  \BibitemOpen
  \bibfield  {author} {\bibinfo {author} {\bibfnamefont {M.}~\bibnamefont
  {Kieferová}}, \bibinfo {author} {\bibfnamefont {A.}~\bibnamefont
  {Scherer}},\ and\ \bibinfo {author} {\bibfnamefont {D.~W.}\ \bibnamefont
  {Berry}},\ }\bibfield  {title} {\bibinfo {title} {Simulating the dynamics of
  time-dependent {Hamiltonians} with a truncated {Dyson} series},\ }\href
  {https://doi.org/10.1103/PhysRevA.99.042314} {\bibfield  {journal} {\bibinfo
  {journal} {Physical Review A}\ }\textbf {\bibinfo {volume} {99}},\ \bibinfo
  {pages} {042314} (\bibinfo {year} {2019})}\BibitemShut {NoStop}%
\bibitem [{\citenamefont {Chen}\ \emph
  {et~al.}(2021{\natexlab{a}})\citenamefont {Chen}, \citenamefont {Kalev},\
  and\ \citenamefont {Hen}}]{chen_quantum_2021}%
  \BibitemOpen
  \bibfield  {author} {\bibinfo {author} {\bibfnamefont {Y.-H.}\ \bibnamefont
  {Chen}}, \bibinfo {author} {\bibfnamefont {A.}~\bibnamefont {Kalev}},\ and\
  \bibinfo {author} {\bibfnamefont {I.}~\bibnamefont {Hen}},\ }\bibfield
  {title} {\bibinfo {title} {Quantum {Algorithm} for {Time}-{Dependent}
  {Hamiltonian} {Simulation} by {Permutation} {Expansion}},\ }\href
  {https://doi.org/10.1103/PRXQuantum.2.030342} {\bibfield  {journal} {\bibinfo
   {journal} {PRX Quantum}\ }\textbf {\bibinfo {volume} {2}},\ \bibinfo {pages}
  {030342} (\bibinfo {year} {2021}{\natexlab{a}})}\BibitemShut {NoStop}%
\bibitem [{\citenamefont {An}\ \emph {et~al.}(2022)\citenamefont {An},
  \citenamefont {Fang},\ and\ \citenamefont {Lin}}]{an_time_dependent_2022}%
  \BibitemOpen
  \bibfield  {author} {\bibinfo {author} {\bibfnamefont {D.}~\bibnamefont
  {An}}, \bibinfo {author} {\bibfnamefont {D.}~\bibnamefont {Fang}},\ and\
  \bibinfo {author} {\bibfnamefont {L.}~\bibnamefont {Lin}},\ }\bibfield
  {title} {\bibinfo {title} {Time-dependent {Hamiltonian} {Simulation} of
  {Highly} {Oscillatory} {Dynamics} and {Superconvergence} for {Schr\"odinger}
  {Equation}},\ }\href {https://doi.org/10.22331/q-2022-04-15-690} {\bibfield
  {journal} {\bibinfo  {journal} {Quantum}\ }\textbf {\bibinfo {volume} {6}},\
  \bibinfo {pages} {690} (\bibinfo {year} {2022})}\BibitemShut {NoStop}%
\bibitem [{\citenamefont {Fang}\ \emph {et~al.}(2025)\citenamefont {Fang},
  \citenamefont {Liu},\ and\ \citenamefont
  {Sarkar}}]{fang_time_dependent_2024}%
  \BibitemOpen
  \bibfield  {author} {\bibinfo {author} {\bibfnamefont {D.}~\bibnamefont
  {Fang}}, \bibinfo {author} {\bibfnamefont {D.}~\bibnamefont {Liu}},\ and\
  \bibinfo {author} {\bibfnamefont {R.}~\bibnamefont {Sarkar}},\ }\bibfield
  {title} {\bibinfo {title} {Time-{Dependent} {Hamiltonian} {Simulation} via
  {Magnus} {Expansion}: {Algorithm} and {Superconvergence}},\ }\href
  {https://doi.org/10.1007/s00220-025-05314-5} {\bibfield  {journal} {\bibinfo
  {journal} {Communications in Mathematical Physics}\ }\textbf {\bibinfo
  {volume} {406}},\ \bibinfo {pages} {128} (\bibinfo {year}
  {2025})}\BibitemShut {NoStop}%
\bibitem [{\citenamefont {Bandrauk}\ and\ \citenamefont
  {Lu}(2013)}]{bandrauk_exponential_2013}%
  \BibitemOpen
  \bibfield  {author} {\bibinfo {author} {\bibfnamefont {A.~D.}\ \bibnamefont
  {Bandrauk}}\ and\ \bibinfo {author} {\bibfnamefont {H.}~\bibnamefont {Lu}},\
  }\bibfield  {title} {\bibinfo {title} {Exponential propagators (integrators)
  for the time-dependent {S}chrödinger equation},\ }\href
  {https://doi.org/10.1142/S0219633613400014} {\bibfield  {journal} {\bibinfo
  {journal} {Journal of Theoretical and Computational Chemistry}\ }\textbf
  {\bibinfo {volume} {12}},\ \bibinfo {pages} {1340001} (\bibinfo {year}
  {2013})}\BibitemShut {NoStop}%
\bibitem [{\citenamefont {Chen}\ \emph {et~al.}(2023)\citenamefont {Chen},
  \citenamefont {Foroozandeh}, \citenamefont {Budd},\ and\ \citenamefont
  {Singh}}]{chen_quantum_2023}%
  \BibitemOpen
  \bibfield  {author} {\bibinfo {author} {\bibfnamefont {G.}~\bibnamefont
  {Chen}}, \bibinfo {author} {\bibfnamefont {M.}~\bibnamefont {Foroozandeh}},
  \bibinfo {author} {\bibfnamefont {C.}~\bibnamefont {Budd}},\ and\ \bibinfo
  {author} {\bibfnamefont {P.}~\bibnamefont {Singh}},\ }\href@noop {} {\bibinfo
  {title} {Quantum simulation of highly-oscillatory many-body {Hamiltonians}
  for near-term devices}} (\bibinfo {year} {2023}),\ \bibinfo {note}
  {arXiv:2312.08310}\BibitemShut {NoStop}%
\bibitem [{\citenamefont {Casares}\ \emph {et~al.}(2024)\citenamefont
  {Casares}, \citenamefont {Zini},\ and\ \citenamefont
  {Arrazola}}]{casares_quantum_2024}%
  \BibitemOpen
  \bibfield  {author} {\bibinfo {author} {\bibfnamefont {P.~A.~M.}\
  \bibnamefont {Casares}}, \bibinfo {author} {\bibfnamefont {M.~S.}\
  \bibnamefont {Zini}},\ and\ \bibinfo {author} {\bibfnamefont {J.~M.}\
  \bibnamefont {Arrazola}},\ }\bibfield  {title} {\bibinfo {title} {Quantum
  simulation of time-dependent {Hamiltonians} via commutator-free
  quasi-{Magnus} operators},\ }\href
  {https://doi.org/10.22331/q-2024-12-17-1567} {\bibfield  {journal} {\bibinfo
  {journal} {Quantum}\ }\textbf {\bibinfo {volume} {8}},\ \bibinfo {pages}
  {1567} (\bibinfo {year} {2024})}\BibitemShut {NoStop}%
\bibitem [{\citenamefont {Blanes}\ and\ \citenamefont
  {Moan}(2006)}]{blanes_fourth_2006}%
  \BibitemOpen
  \bibfield  {author} {\bibinfo {author} {\bibfnamefont {S.}~\bibnamefont
  {Blanes}}\ and\ \bibinfo {author} {\bibfnamefont {P.~C.}\ \bibnamefont
  {Moan}},\ }\bibfield  {title} {\bibinfo {title} {Fourth- and sixth-order
  commutator-free {Magnus} integrators for linear and non-linear dynamical
  systems},\ }\href
  {https://doi.org/https://doi.org/10.1016/j.apnum.2005.11.004} {\bibfield
  {journal} {\bibinfo  {journal} {Applied Numerical Mathematics}\ }\textbf
  {\bibinfo {volume} {56}},\ \bibinfo {pages} {1519} (\bibinfo {year}
  {2006})}\BibitemShut {NoStop}%
\bibitem [{\citenamefont {Alvermann}\ and\ \citenamefont
  {Fehske}(2011)}]{alvermann_high_order_2011}%
  \BibitemOpen
  \bibfield  {author} {\bibinfo {author} {\bibfnamefont {A.}~\bibnamefont
  {Alvermann}}\ and\ \bibinfo {author} {\bibfnamefont {H.}~\bibnamefont
  {Fehske}},\ }\bibfield  {title} {\bibinfo {title} {High-order commutator-free
  exponential time-propagation of driven quantum systems},\ }\href
  {https://doi.org/https://doi.org/10.1016/j.jcp.2011.04.006} {\bibfield
  {journal} {\bibinfo  {journal} {Journal of Computational Physics}\ }\textbf
  {\bibinfo {volume} {230}},\ \bibinfo {pages} {5930} (\bibinfo {year}
  {2011})}\BibitemShut {NoStop}%
\bibitem [{\citenamefont {Barthel}\ and\ \citenamefont
  {Zhang}(2020)}]{barthel_optimized_2020}%
  \BibitemOpen
  \bibfield  {author} {\bibinfo {author} {\bibfnamefont {T.}~\bibnamefont
  {Barthel}}\ and\ \bibinfo {author} {\bibfnamefont {Y.}~\bibnamefont
  {Zhang}},\ }\bibfield  {title} {\bibinfo {title} {Optimized
  {Lie}–{Trotter}–{Suzuki} decompositions for two and three non-commuting
  terms},\ }\href {https://doi.org/https://doi.org/10.1016/j.aop.2020.168165}
  {\bibfield  {journal} {\bibinfo  {journal} {Annals of Physics}\ }\textbf
  {\bibinfo {volume} {418}},\ \bibinfo {pages} {168165} (\bibinfo {year}
  {2020})}\BibitemShut {NoStop}%
\bibitem [{\citenamefont {Ostmeyer}(2023)}]{ostmeyer_optimised_2023}%
  \BibitemOpen
  \bibfield  {author} {\bibinfo {author} {\bibfnamefont {J.}~\bibnamefont
  {Ostmeyer}},\ }\bibfield  {title} {\bibinfo {title} {Optimised {Trotter}
  decompositions for classical and quantum computing},\ }\href
  {https://doi.org/10.1088/1751-8121/acde7a} {\bibfield  {journal} {\bibinfo
  {journal} {J. Phys. A: Math. Theor.}\ }\textbf {\bibinfo {volume} {56}},\
  \bibinfo {pages} {285303} (\bibinfo {year} {2023})}\BibitemShut {NoStop}%
\bibitem [{\citenamefont {Mizuta}\ \emph {et~al.}(2024)\citenamefont {Mizuta},
  \citenamefont {Ikeda},\ and\ \citenamefont {Fujii}}]{mizuta_explicit_2024}%
  \BibitemOpen
  \bibfield  {author} {\bibinfo {author} {\bibfnamefont {K.}~\bibnamefont
  {Mizuta}}, \bibinfo {author} {\bibfnamefont {T.~N.}\ \bibnamefont {Ikeda}},\
  and\ \bibinfo {author} {\bibfnamefont {K.}~\bibnamefont {Fujii}},\
  }\href@noop {} {\bibinfo {title} {Explicit error bounds with commutator
  scaling for time-dependent product and multi-product formulas}} (\bibinfo
  {year} {2024}),\ \bibinfo {note} {arXiv:2410.14243}\BibitemShut {NoStop}%
\bibitem [{\citenamefont {Zeng}\ \emph {et~al.}(2016)\citenamefont {Zeng},
  \citenamefont {Zhang},\ and\ \citenamefont {Sarovar}}]{zeng_schedule_2016}%
  \BibitemOpen
  \bibfield  {author} {\bibinfo {author} {\bibfnamefont {L.}~\bibnamefont
  {Zeng}}, \bibinfo {author} {\bibfnamefont {J.}~\bibnamefont {Zhang}},\ and\
  \bibinfo {author} {\bibfnamefont {M.}~\bibnamefont {Sarovar}},\ }\bibfield
  {title} {\bibinfo {title} {Schedule path optimization for adiabatic quantum
  computing and optimization},\ }\href
  {https://doi.org/10.1088/1751-8113/49/16/165305} {\bibfield  {journal}
  {\bibinfo  {journal} {J. Phys. A: Math. Theor.}\ }\textbf {\bibinfo {volume}
  {49}},\ \bibinfo {pages} {165305} (\bibinfo {year} {2016})}\BibitemShut
  {NoStop}%
\bibitem [{\citenamefont {An}\ and\ \citenamefont
  {Lin}(2022)}]{an_quantum_2022}%
  \BibitemOpen
  \bibfield  {author} {\bibinfo {author} {\bibfnamefont {D.}~\bibnamefont
  {An}}\ and\ \bibinfo {author} {\bibfnamefont {L.}~\bibnamefont {Lin}},\
  }\bibfield  {title} {\bibinfo {title} {Quantum {Linear} {System} {Solver}
  {Based} on {Time}-optimal {Adiabatic} {Quantum} {Computing} and {Quantum}
  {Approximate} {Optimization} {Algorithm}},\ }\href
  {https://doi.org/10.1145/3498331} {\bibfield  {journal} {\bibinfo  {journal}
  {ACM Transactions on Quantum Computing}\ }\textbf {\bibinfo {volume} {3}},\
  \bibinfo {pages} {1} (\bibinfo {year} {2022})}\BibitemShut {NoStop}%
\bibitem [{\citenamefont {Kendon}\ \emph {et~al.}(2010)\citenamefont {Kendon},
  \citenamefont {Nemoto},\ and\ \citenamefont {Munro}}]{kendon_quantum_2010}%
  \BibitemOpen
  \bibfield  {author} {\bibinfo {author} {\bibfnamefont {V.~M.}\ \bibnamefont
  {Kendon}}, \bibinfo {author} {\bibfnamefont {K.}~\bibnamefont {Nemoto}},\
  and\ \bibinfo {author} {\bibfnamefont {W.~J.}\ \bibnamefont {Munro}},\
  }\bibfield  {title} {\bibinfo {title} {Quantum analogue computing},\ }\href
  {https://doi.org/10.1098/rsta.2010.0017} {\bibfield  {journal} {\bibinfo
  {journal} {Philosophical Transactions of the Royal Society A: Mathematical,
  Physical and Engineering Sciences}\ }\textbf {\bibinfo {volume} {368}},\
  \bibinfo {pages} {3609} (\bibinfo {year} {2010})}\BibitemShut {NoStop}%
\bibitem [{\citenamefont {Forest}\ and\ \citenamefont
  {Ruth}(1990)}]{forest_fourth-order_1990}%
  \BibitemOpen
  \bibfield  {author} {\bibinfo {author} {\bibfnamefont {E.}~\bibnamefont
  {Forest}}\ and\ \bibinfo {author} {\bibfnamefont {R.~D.}\ \bibnamefont
  {Ruth}},\ }\bibfield  {title} {\bibinfo {title} {Fourth-order symplectic
  integration},\ }\href {https://doi.org/10.1016/0167-2789(90)90019-L}
  {\bibfield  {journal} {\bibinfo  {journal} {Physica D: Nonlinear Phenomena}\
  }\textbf {\bibinfo {volume} {43}},\ \bibinfo {pages} {105} (\bibinfo {year}
  {1990})}\BibitemShut {NoStop}%
\bibitem [{\citenamefont {Blanes}\ \emph {et~al.}(1999)\citenamefont {Blanes},
  \citenamefont {Casas},\ and\ \citenamefont
  {Ros}}]{blanes_extrapolation_1999}%
  \BibitemOpen
  \bibfield  {author} {\bibinfo {author} {\bibfnamefont {S.}~\bibnamefont
  {Blanes}}, \bibinfo {author} {\bibfnamefont {F.}~\bibnamefont {Casas}},\ and\
  \bibinfo {author} {\bibfnamefont {J.}~\bibnamefont {Ros}},\ }\bibfield
  {title} {\bibinfo {title} {Extrapolation of symplectic {Integrators}},\
  }\href {https://doi.org/10.1023/A:1008364504014} {\bibfield  {journal}
  {\bibinfo  {journal} {Celestial Mechanics and Dynamical Astronomy}\ }\textbf
  {\bibinfo {volume} {75}},\ \bibinfo {pages} {149} (\bibinfo {year}
  {1999})}\BibitemShut {NoStop}%
\bibitem [{\citenamefont {Chen}\ \emph
  {et~al.}(2021{\natexlab{b}})\citenamefont {Chen}, \citenamefont {Huang},
  \citenamefont {Kueng},\ and\ \citenamefont
  {Tropp}}]{chen_concentration_2021}%
  \BibitemOpen
  \bibfield  {author} {\bibinfo {author} {\bibfnamefont {C.-F.}\ \bibnamefont
  {Chen}}, \bibinfo {author} {\bibfnamefont {H.-Y.}\ \bibnamefont {Huang}},
  \bibinfo {author} {\bibfnamefont {R.}~\bibnamefont {Kueng}},\ and\ \bibinfo
  {author} {\bibfnamefont {J.~A.}\ \bibnamefont {Tropp}},\ }\bibfield  {title}
  {\bibinfo {title} {Concentration for {Random} {Product} {Formulas}},\ }\href
  {https://doi.org/10.1103/PRXQuantum.2.040305} {\bibfield  {journal} {\bibinfo
   {journal} {PRX Quantum}\ }\textbf {\bibinfo {volume} {2}},\ \bibinfo {pages}
  {040305} (\bibinfo {year} {2021}{\natexlab{b}})}\BibitemShut {NoStop}%
\bibitem [{\citenamefont {Fang}\ \emph {et~al.}(2023)\citenamefont {Fang},
  \citenamefont {Lin},\ and\ \citenamefont {Tong}}]{fang_time-marching_2023}%
  \BibitemOpen
  \bibfield  {author} {\bibinfo {author} {\bibfnamefont {D.}~\bibnamefont
  {Fang}}, \bibinfo {author} {\bibfnamefont {L.}~\bibnamefont {Lin}},\ and\
  \bibinfo {author} {\bibfnamefont {Y.}~\bibnamefont {Tong}},\ }\bibfield
  {title} {\bibinfo {title} {Time-marching based quantum solvers for
  time-dependent linear differential equations},\ }\href
  {https://doi.org/10.22331/q-2023-03-20-955} {\bibfield  {journal} {\bibinfo
  {journal} {Quantum}\ }\textbf {\bibinfo {volume} {7}},\ \bibinfo {pages}
  {955} (\bibinfo {year} {2023})}\BibitemShut {NoStop}%
\bibitem [{\citenamefont {Garnerone}\ \emph {et~al.}(2012)\citenamefont
  {Garnerone}, \citenamefont {Zanardi},\ and\ \citenamefont
  {Lidar}}]{garnerone_adiabatic_2012}%
  \BibitemOpen
  \bibfield  {author} {\bibinfo {author} {\bibfnamefont {S.}~\bibnamefont
  {Garnerone}}, \bibinfo {author} {\bibfnamefont {P.}~\bibnamefont {Zanardi}},\
  and\ \bibinfo {author} {\bibfnamefont {D.~A.}\ \bibnamefont {Lidar}},\
  }\bibfield  {title} {\bibinfo {title} {Adiabatic {Quantum} {Algorithm} for
  {Search} {Engine} {Ranking}},\ }\href
  {https://doi.org/10.1103/PhysRevLett.108.230506} {\bibfield  {journal}
  {\bibinfo  {journal} {Physical Review Letters}\ }\textbf {\bibinfo {volume}
  {108}},\ \bibinfo {pages} {230506} (\bibinfo {year} {2012})}\BibitemShut
  {NoStop}%
\end{thebibliography}%

\appendix

\section{Additional proofs}

\subsection{Proof of Theorem \ref{thm::mpf_analog}}
\label{proof::thm::mpf_analog}

The fidelity of $\rho_\omega$ with respect to the target state $\ket{\psi(t)} = \exactU{t}{0} \ket{\psi(0)}$ is given as 
\begin{align*}
\bra{\psi(t)} \rho_\omega(t) \ket{\psi(t)} &= \int\ \dd s\ \abs{G(s-t)}^2\  \bra{\psi(t)} \exactU{s}{s-t} \ketbra{\psi(0)} \exactU{s}{s-t}^{\dagger} \ket{\psi(t)}\\
&= \int\ \dd s\ \abs{G(s-t)}^2\ \abs\big{\bra{\psi(0)} \exactU{0}{t} \exactU{s}{s-t} \ket{\psi(0)}}^2\\
&= \int\ \dd s\ \abs{G(s)}^2\ \underbrace{\abs\big{ \bra{\psi(0)} \exactU{0}{t} \exactU{s+t}{s} \ket{\psi(0)}}^2}_{=:\theta(s)}.
\end{align*}
By Taylor's expansion of $\theta$ near $s\approx 0$, 
\begin{align*}
\bra{\psi(t)} \rho_\omega(t) \ket{\psi(t)} =& \int\ \dd s\ \abs{G(s)}^2 \sum_{\ell=0}^{\infty} \frac{\theta^{(\ell)}(0)}{\ell!} s^\ell \\
=& \sum_{\ell=0}^{\infty} \frac{\theta^{(\ell)}(0)}{\ell!} \big(\int\ \dd s\ \abs{G(s)}^2 s^\ell\Big).
\end{align*}
Since we choose $\abs{G}^2$ to be a Gaussian $\mathcal{N}(0,\omega^2)$, we have 
\begin{align*}
\bra{\psi(t)} \rho_\omega(t) \ket{\psi(t)} =& \sum_{\ell=0,2,4,\cdots} \frac{\theta^{(\ell)}(0)}{\ell!} \frac{2^{\ell/2}\Gamma(\nicefrac{\ell+1}{2})}{\sqrt{\pi}} \omega^\ell \\
=& 1 + \sum_{\ell=1,2,\cdots} \frac{\theta^{(2\ell)}(0)}{(2\ell)!} \frac{2^{\ell}\Gamma(\ell+\nicefrac{1}{2})}{\sqrt{\pi}} \omega^{2\ell}.
\end{align*}
The case $\theta^{(0)}(0) = 1$ is easy to verify by definition.
By choosing various $\alpha_j$ and $k_j$ that satisfies \eqref{eqn::mpf_order}, we have
\begin{align*}
\bra{\psi(t)} \sum_{j=1}^{M}\ \alpha_j \rho_{\omega/k_j}(t)\ \ket{\psi(t)} =& \sum_{j=1}^{M} \alpha_j \Big(1 + \sum_{\ell=1,2,\cdots} \frac{\theta^{(2\ell)}(0)}{(2\ell)!} \frac{2^{\ell}\Gamma(\ell+\nicefrac{1}{2})}{\sqrt{\pi}} \big(\frac{\omega}{k_j}\big)^{2\ell}\Big) \\
=& 1 + \sum_{\ell=1,2,\cdots} \frac{1}{(2\ell)!} \frac{2^{\ell}\Gamma(\ell+\nicefrac{1}{2})}{\sqrt{\pi}} \omega^{2\ell}\  \sum_{j=1}^{M} \alpha_j \big(\frac{1}{k_j}\big)^{2\ell} \theta^{(2\ell)}(0) \\
=& 1 + \frac{\theta^{(2m)}(0)}{(2 m)!} \frac{2^{m}\Gamma(m+\nicefrac{1}{2})}{\sqrt{\pi}} \omega^{2 m}  \sum_{j=1}^{M} \alpha_j \big(\frac{1}{k_j}\big)^{2 m}\  + \order{\omega^{2m+2}} \\
=& 1 - C_m \omega^{2m} + \order{\omega^{2m+2}},
\end{align*}
where 
\begin{align}
\label{eqn::Cm}
C_m = -\frac{2^{m}\Gamma(m+\nicefrac{1}{2})}{(2 m)!\sqrt{\pi}} \big(\sum_{j=1}^{M} \alpha_j \big(\frac{1}{k_j}\big)^{2 m}\big)\ \theta^{(2\ell)}(0) > 0.
\end{align}
Since $\rho_{\omega/k_j}(t)$ is a density matrix, it holds  that $\bra{\psi(t)} \sum_{j=1}^{M} \alpha_j \rho_{\omega/k_j}(t)\ket{\psi(t)}\le 1$ is always true and thus $C_m \ge 0$.
Then
\begin{align*}
&\ \abs\Big{\ \sum_{j=1}^{M} \alpha_j \tr\big(\widehat{O} \rho_{\omega/k_j}(t)\big) - \tr(\widehat{O} \ket{\psi(t)}\bra{\psi(t)})\ } \\
\le&\ \norm\bigg{\ \sum_{j=1}^{M} \alpha_j \rho_{\omega/k_j}(t) - \ket{\psi(t)}\bra{\psi(t)}\ }_1  \\
\le&\ 2 \sqrt{1 - \bra{\psi(t)} \sum_{j=1}^{M} \alpha_j \rho_{\omega/k_j}(t) \ket{\psi(t)} } \\
=&\ 2 \sqrt{C_m}\ \omega^{m} + \order{\omega^{m+2}}.
\end{align*}

\subsection{Proof of  Theorem \ref{thm::omega_analog}}
\label{proof::aqc}

Recall from the error estimates in \cite[Lemma 11]{time_dilation_2023} that 
\begin{align*}
\norm\Big{\ \ketbra{\psi(1)} - \tr_s\big(\ketbra{\Psi(t)}\big)\ }_{\tr} \le \sqrt{C}\omega.
\end{align*}
By the triangle inequality of trace norm, one needs 
\begin{align}
\label{eqn::c_omega_epsilon}
\sqrt{C} \omega \le \sqrt{\epsilon}.
\end{align}
By \cite[Lemma 11]{time_dilation_2023}, the constant $C$ is
\begin{align*}
C &= \expval{H^2(1)} - 2 \text{Re}\Big(\bra{\psi(1)} H(1) \exactU{t}{0} H(0) \ket{\psi(0)}\Big) + \expval{H^2(0)} - \Big(\expval{H(1)} - \expval{H(0)}\Big)^2\\
&= \expval{H^2(1)} - 2 T \lambda_1 \expval{H(1)} + T^2 \lambda_1^2 - \expval{H(1)}^2 + 2 T \lambda_1 \expval{H(1)} - T^2\lambda_1^2\\
&= \expval{H^2(1)} - \expval{H(1)}^2,
\end{align*}
where $\expval{O(t)} := \bra{\psi(t)} O(t) \ket{\psi(t)}$ for any time-dependent observable $O$ and time $t\in [0,1]$. To get the above, we have used the fact that the state $\psi(0)$ is prepared as a ground state of $\h_1$ (with $\h_1 \ket{\psi(0)} = \lambda_1 \ket{\psi(0)}$ and thus $H(0) \ket{\psi(0)} = T \lambda_1 \ket{\psi(0)}$.

Suppose that $\ket{\psi(1)} = \alpha \ket{\exact} + \beta \ket{\exactorth}$ where $\ket{\exactorth}$ is perpendicular to $\ket{\exact}$. By the error assumption in \eqref{eqn::err_aqc}, and normalization of $\ket{\psi(1)}$, 
\begin{align*}
\alpha = 1 - \order{\eps}, \qquad \beta = \order{\sqrt{\eps}}.
\end{align*}
With such a decomposition, one can compute that 
\begin{align*}
 \expval{H^2(1)}  &=  T^2 \bra{\alpha \exact + \beta \exactorth} \h_2^2 \ket{\alpha \exact + \beta \exactorth} \\
 &= T^2\big(\abs{\alpha}^2 \lambda_2^2 + \abs{\beta}^2 \bra{\exactorth} \h_2^2 \ket{\exactorth}\big);\\
  \expval{H(1)} &= T \bra{\alpha \exact + \beta \exactorth} \h_2 \ket{\alpha \exact + \beta \exactorth} \\
  &= T\big(\abs{\alpha}^2 \lambda_2 + \abs{\beta}^2 \bra{\exactorth} \h_2 \ket{\exactorth}\big),
\end{align*}
where $\lambda_2$ is the ground energy of $\h_2$ and $\h_2 \ket{\exact} = \lambda_2 \ket{\exact}$.
Hence,
\begin{align*}
 & \frac{1}{T^2}\big(\expval{H^2(1)}  -   \expval{H(1)}^2\big) \\
 =& (\abs{\alpha}^2 -\abs{\alpha}^4)\lambda_2^2 + \abs{\beta}^2 \bra{\exactorth} \h_2^2 \ket{\exactorth} - \abs{\beta}^4 \bra{\exactorth} \h_2 \ket{\exactorth}^2 - 2 \abs{\alpha}^2 \abs{\beta}^2 \lambda_2 \bra{\exactorth} \h_2 \ket{\exactorth}\\
 =& \order{\eps \lambda_2^2 + \eps \bra{\exactorth} \h_2^2 \ket{\exactorth} + \eps^2 \bra{\exactorth} \h_2 \ket{\exactorth}^2 + \eps \lambda_2 \bra{\exactorth} \h_2 \ket{\exactorth}}\\
=& \order{\eps (\lambda_2^2 + \norm{\h_2}^2) + \eps^2 \norm{\h_2}^2}\\
=& \order{\eps \norm{\h_2}^2}.
\end{align*}
To maintain the same order of error, it is sufficient to pick 
\begin{align*}
\sqrt{C} \omega = \order{T \sqrt{\eps} \norm{\h_2} \omega} \le \sqrt{\eps},
\end{align*}
which leads to \eqref{eqn::omega}.

\subsection{Proof of \lemref{lem::product_1}}
\label{subsec::lem::product_1}

We list without proof some facts, which can be easily validated by definition:
\begin{lemma}
        If the quantum state $\ket{\Psi}$ comes from the augmented space $\augspace$, then for any $s'\in \timedom$, $\tau\in \Real$, any index $k$ and $\varphi\in \hbt$, we have
        \begin{align}
        \label{eqn::v2_basic}
        \begin{aligned}
        \bra{s'} \bra{\varphi} e^{-i \tau H_k(\hat{s})} \ket{\Psi} &= \bra{\varphi} e^{-i \tau H_k(s')}\ket{\Psi(s',\cdot)} = \bra{s'}\bra{\varphi} \id \otimes e^{-i \tau H_k(s')} \ket{\Psi}, \\
        \bra{s'} \bra{\varphi} e^{-i  \tau \hat{p}_s\otimes \id} \ket{\Psi} &= \bra{s'-\tau}\bra{\varphi} \ket{\Psi} \equiv \bra{\varphi}\ket{\Psi(s'-\tau,\cdot)}.
        \end{aligned}
        \end{align}
\end{lemma}

Next we shall prove the first equation in Lemma \ref{lem::product_1} (also copied below)
\begin{align*}
 &\ \bra{s'} \bra{\varphi}\prod_{k=1}^{\Lambda+1} e^{A_k \theta \dt}\ket{\Psi}\ =\ \bra{s'-\theta\dt} \bra{\varphi} \Big(\id \otimes U_F(s', s' - \theta\dt)\Big) \ket{\Psi},
\end{align*}
and the second one can be similarly validated.
\begin{align*}
&\ \bra{s'} \bra{\varphi}\prod_{k=1}^{\Lambda+1} e^{A_k \theta \dt}\ket{\Psi}\\
=&\ \bra{s'}\bra{\varphi} \prod_{k=1}^{\Lambda'} e^{-i \theta \dt H_k(\hat{s})} e^{-i \theta\dt \hat{p}_s\otimes \id } \prod_{k=\Lambda'+1}^{\Lambda} e^{-i \theta \dt H_k(\hat{s})}\ket{\Psi}\\
\myeq{\eqref{eqn::v2_basic}} &\ \bra{s'}  \big(\bra{\varphi} \prod_{k=1}^{\Lambda'} e^{-i \theta\dt H_k(s')}\big) e^{-i \theta\dt \hat{p}_s\otimes \id } \prod_{k=\Lambda'+1}^{\Lambda} e^{-i \theta \dt H_k(\hat{s})}\ket{\Psi}\\
\myeq{\eqref{eqn::v2_basic}} &\ \bra{s' - \theta\dt}  \big(\bra{\varphi} \prod_{k=1}^{\Lambda'} e^{-i \theta\dt H_k(s')}\big) \prod_{k=\Lambda'+1}^{\Lambda} e^{-i \theta \dt H_k(\hat{s})}\ket{\Psi}\\
\myeq{\eqref{eqn::v2_basic}} &\ \bra{s' - \theta\dt}  \big(\bra{\varphi} \prod_{k=1}^{\Lambda'} e^{-i \theta\dt H_k(s')} \prod_{k=\Lambda'+1}^{\Lambda} e^{-i \theta \dt H_k(s'-\theta\dt)}\big) \ket{\Psi}\\ 
=&\  \big(\bra{\varphi} \prod_{k=1}^{\Lambda'} e^{-i \theta\dt H_k(s')} \prod_{k=\Lambda'+1}^{\Lambda} e^{-i \theta \dt H_k(s'-\theta\dt)}\big) \ket{\Psi(s'-\theta\dt, \cdot)}\\ 
=&\ \bra{s'-\theta\dt}\bra{\varphi} (\id \otimes U_F(s',s'-\theta\dt)) \ket{\Psi}. 
\end{align*}
Thus completes the proof.

\subsection{Explicit form of the HDR scheme with the FRS weight}

When one applies  the Forest-Ruth-Suzki's formula to HDR scheme, one has
\begin{align}
        \label{eqn::HDR_FRS}
        \begin{aligned}
        &\ \mathcal{T}e^{-i\int_{t}^{t+\dt} H_1(s) + H_2(s)\ \dd s} \\
        \approx &\  \hdu_F\big(t+\dt, t+(1-\frac{\gamma}{2})\dt\big) \hdu_B\big(t+(1-\frac{\gamma}{2})\dt, t+(1-\gamma)\dt\big)\\
        &\qquad \hdu_F\big(t+(1-\gamma)\dt, t+\frac{1}{2}\dt\big) \hdu_B\big(t+\frac{1}{2}\dt, t + \gamma\dt\big) \\
        &\qquad \hdu_F\big(t+\gamma\dt, t+\frac{\gamma}{2} \dt\big) \hdu_B\big(t+\frac{\gamma}{2}\dt, t\big) \\
        =&\  \eT{-i \int_{(1-\frac{\gamma}{2})\dt}^{\dt} H_1(t+r)\ \dd r} \eT{-i \int_{(1-\gamma)\dt}^{\dt} H_2(t+r)\ \dd r} \\
        &\qquad \eT{-i \int_{\frac{1}{2}\dt}^{(1-\frac{\gamma}{2})\dt} H_1(t+r)\dd r} 
         \eT{-i \int_{\gamma\dt}^{(1-\gamma)\dt} H_2(t+r)\dd r} \\
         &\qquad \eT{-i\int_{\frac{\gamma}{2}\dt}^{\frac{1}{2}\dt} H_1(t+r)\dd r } 
         \eT{-i \int_{0}^{\gamma\dt} H_2(t+r)\dd r} \\
         &\qquad \eT{-i \int_{0}^{\frac{\gamma}{2}\dt} H_1(t+r)\dd r}.
        \end{aligned}
\end{align}

\subsection{Proofs for \secref{sec::qDrift}}
\label{proof::qDrift}

\begin{proof}[Proof of Lemma \ref{lem::qdrift_indep}]
We shall prove the following result: for any density matrix $\rho$, 
\begin{align}
\label{eqn::error_qDrift}
\begin{aligned}
&\ e^{-i \dt H} \rho\ e^{i \dt H} - \int \dd \omega\ \lambda(\omega) e^{-i \dt H_\omega}\ \rho\  e^{i \dt H_\omega} \\
=&\ -\frac{\dt^2}{2} \begin{aligned} \iint \dd\omega_1\dd\omega_2\ \lambda(\omega_1) \big(\lambda(\omega_2) - \delta(\omega_1 - \omega_2)\big) \comm\big{H_{\omega_1}}{\comm{H_{\omega_2}}{\rho}}\end{aligned}  + \order{\dt^3}.
\end{aligned}
\end{align}
The conclusion of lemma can be readily obtained by applying the triangle inequality to this equation.

One can prove the last equation by direct expansion:
\begin{align*}
&\ e^{-i \dt H} \rho\ e^{i \dt H} - \int \dd \omega\ \lambda(\omega) e^{-i \dt H_\omega}\ \rho\  e^{i \dt H_\omega} \\
=&\ \dt^2 \Big(-\frac{1}{2} \comm\big{H}{\comm{H}{\rho}} + \frac{1}{2} \int\dd\omega\ \lambda(\omega) \comm\big{H_{\omega}}{\comm{H_{\omega}}{\rho}}\Big) + \order{\dt^3} \\
=&\ -\frac{1}{2}\dt^2 \Big(\iint \dd\omega_1\dd\omega_2\ \lambda(\omega_1)\lambda(\omega_2) \comm\big{H_{\omega_1}}{\comm{H_{\omega_2}}{\rho}} - \int\dd\omega\ \lambda(\omega) \comm\big{H_{\omega}}{\comm{H_{\omega}}{\rho}}\Big) + \order{\dt^3}\\
=&\ -\frac{1}{2}\dt^2 \Big(\iint \dd\omega_1\dd\omega_2\ \lambda(\omega_1) \big(\lambda(\omega_2) - \delta(\omega_1 - \omega_2)\big)\ \comm\big{H_{\omega_1}}{\comm{H_{\omega_2}}{\rho}}\Big) + \order{\dt^3}.
\end{align*}
\end{proof}

\begin{proof}[Proof of Lemma \lemref{lem::qdrift_err_1}]
As can be seen, the only error comes from the application of time-independent qDrift algorithm, thus by the proof of Lemma \ref{lem::qdrift_indep} or more specifically by Eq.~\eqref{eqn::error_qDrift},
\begin{align*}
&\ \exactU{t+\dt}{t}\ \rho\ \exactU{t+\dt}{t}^\dagger -  \mathcal{E}_{\lambda,t,\dt}(\rho)\\ 
=&\ \bra{t+\dt}_{s} \Bigg(- \frac{1}{2}\dt^2 \sum_{k_1, k_2} \lambda_{k_1} \big(\lambda_{k_2} - \delta_{k_1, k_2}) \comm\bigg{\wt{H}_{k_1}}{\comm{\wt{H}_{k_2}}{\ket{\one}\bra{\one}\otimes \rho}}\Bigg)\ket{t+\dt}_{s} + \order{\dt^3}\\
=&\ -\frac{1}{2}\dt^2 \sum_{k_1, k_2} \lambda_{k_1} \big(\lambda_{k_2} - \delta_{k_1, k_2}) \bra{t+\dt}_{s} \comm\bigg{\wt{H}_{k_1}}{\comm{\wt{H}_{k_2}}{\ket{\one}\bra{\one}\otimes \rho}} \ket{t+\dt}_{s} + \order{\dt^3}.
\end{align*}

Next, let us estimate
\begin{align*}
&\ \comm\bigg{\wt{H}_{k_1}}{\comm{\wt{H}_{k_2}}{\ket{\one}\bra{\one}\otimes \rho}} \\
=&\ \comm\bigg{\hat{p}_s \otimes \id + \frac{H_{k_1}(\hat{s})}{\lambda_{k_1}}}{\comm{\hat{p}_s \otimes \id + \frac{H_{k_2}(\hat{s})}{\lambda_{k_2}}}{\ket{\one}\bra{\one}\otimes \rho}} \\
=&\ \frac{1}{\lambda_{k_2}}\ \comm\bigg{\hat{p}_s \otimes \id + \frac{H_{k_1}(\hat{s})}{\lambda_{k_1}}}{\ket{f_{k_2}(s)}\bra{\one} \otimes \h_2 \rho - \ket{\one}\bra{f_{k_2}(s)} \otimes \rho \h_2} \\
=&\ \frac{1}{\lambda_{k_2}} \left(\begin{aligned}
& -i \ket{f_{k_2}'(s)}\bra{\one} \otimes \h_2 \rho + \frac{1}{\lambda_{k_1}} \ket{f_{k_1}(s)f_{k_2}(s)}\bra{\one} \otimes \h_1 \h_2 \rho \\
& - \frac{1}{\lambda_{k_1}} \ket{f_{k_2}(s)}\bra{f_{k_1}(s)} \otimes \h_2 \rho \h_1
+ \ket{\one} \bra{ (-i) f_{k_2}'(s)} \otimes \rho \h_2 \\
& - \frac{1}{\lambda_{k_1}} \ket{f_{k_1}(s)}\bra{f_{k_2}(s)} \otimes \h_1 \rho \h_2  + \frac{1}{\lambda_{k_1}} \ket{\one}\bra{f_{k_2}(s)f_{k_1}(s)}\otimes \rho \h_2 \h_1
\end{aligned}\right),
\end{align*}
where the notation $\ket{f_{k}(s)}\equiv \int \dd s\ f_{k}(s) \ket{s}$ and similar conventions apply to other places.
Given such an expansion, and with expressions \eqref{eqn::decompose} and \eqref{eqn::Hk}, one can straightforwardly validate that
\begin{align*}
&\ \exactU{t+\dt}{t} \rho \exactU{t+\dt}{t}^\dagger -  \mathcal{E}_{\lambda,t,\dt}(\rho) \\
=&\ -\frac{\dt^2}{2} \sum_{k_1, k_2} \lambda_{k_1} \big(\lambda_{k_2} - \delta_{k_1, k_2}) \Big(\frac{-i f_{k_2}'(t+\dt)}{\lambda_{k_2}} \comm{\h_{k_2}}{\rho} + \frac{f_{k_1}(t+\dt) f_{k_2}(t+\dt)}{\lambda_{k_1} \lambda_{k_2}} \comm\big{\h_{k_1}}{\comm{\h_{k_2}}{\rho}} \Big) \\
&\qquad + \order{\dt^3}\\
=&\ -\frac{\dt^2}{2} \sum_{k_1, k_2} \lambda_{k_1} \big(\lambda_{k_2} - \delta_{k_1, k_2}) \frac{f_{k_1}(t+\dt) f_{k_2}(t+\dt)}{\lambda_{k_1} \lambda_{k_2}} \comm\big{\h_{k_1}}{\comm{\h_{k_2}}{\rho}} + \order{\dt^3}.
\end{align*}
The third line is obtained by observing that $\sum_{k_1} \lambda_{k_1} (\lambda_{k_2} - \delta_{k_1,k_2}) = \lambda_{k_2} - \lambda_{k_2} = 0$.
Finally, by applying the triangle inequality, one can immediately obtain the conclusion.
\end{proof}

\begin{proof}[Proof of Lemma \ref{lem::qdrift_err_2}]
The proof is not much different from Lemma \ref{lem::qdrift_err_1}. We can similarly verify that for $\omega_1 = (k_1, r_1)$, $\omega_2 = (k_2, r_2)$, by Lemma \ref{lem::qdrift_indep}, or more specifically by Eq.~\eqref{eqn::error_qDrift}, the error term is 
\begin{align*}
&\ \exactU{t+\dt}{t}\ \rho\ \exactU{t+\dt}{t}^\dagger -  \mathcal{E}_{\mu,t,\dt}(\rho) \\
=&\ -\frac{\dt^2}{2} \iint\dd\omega_1\dd\omega_2\ \mu(\omega_1) \big(\mu(\omega_2) - \delta(\omega_1 - \omega_2)\big) 
\left(\begin{aligned} & \frac{-i f_{k_2}'(t+\dt)}{\mu(\omega_2)} \comm{\h_{k_2}}{\rho} \\
&+ \frac{f_{k_1}(t+\dt) f_{k_2}(t+\dt)}{\mu(\omega_1)\mu(\omega_2)} \comm\big{\h_{k_1}}{\comm{\h_{k_2}}{\rho}}\end{aligned}\right) \\
&\qquad + \order{\dt^3}\\
=&\ -\frac{\dt^2}{2} \iint\dd\omega_1\dd\omega_2\ \mu(\omega_1) \big(\mu(\omega_2) - \delta(\omega_1 - \omega_2)\big) 
\left(\begin{aligned} \frac{f_{k_1}(t+\dt) f_{k_2}(t+\dt)}{\mu(\omega_1)\mu(\omega_2)} \comm\big{\h_{k_1}}{\comm{\h_{k_2}}{\rho}}\end{aligned}\right) \\
&\qquad + \order{\dt^3}.
\end{align*}
Finally, one can easily get the upper bound of $C$ via the triangle inequality.
\end{proof}

\begin{proof}[Proof of Lemma \ref{lem::c_qDrift_equiv}]
As we choose Hamiltonians with the form \eqref{eqn::decompose} and \eqref{eqn::Hk},
\begin{align*}
\mathcal{E}_{\mu, t, \dt}(\rho)\ \myeq{\eqref{eqn::cqDrift_v2}}&\ \ \sum_{k} \int_{0}^{1} \dd r\ \mu(k,r)\ \eT{-i \frac{\int_{t}^{t+\dt} H_k(s)\dd s}{\mu(k,r)}}\ \rho\ \eT{i \frac{\int_{t}^{t+\dt} H_k(s)\dd s}{\mu(k,r)}} \\
=&\  \sum_{k} \int_{0}^{1} \dd r\ \mu(k,r)\ e^{-i \dt \frac{\frac{1}{\dt}\int_{t}^{t+\dt} f_k(s)\dd s}{\mu(k,r)} \h_k}\ \rho\ e^{i \dt \frac{\frac{1}{\dt}\int_{t}^{t+\dt} f_k(s)\dd s}{\mu(k,r)} \h_k}.
\end{align*}
For each $k$, define $F_k: [0,1] \to \Real^{+}$ and $\tau_k: [0,1]\to [0,1]$ by
\begin{align}
\label{eqn::tau_k}
F_k(r) = \int_{0}^{r} f_k(t+\dt s)\ \dd s, \qquad \tau_k(r) := F_k^{-1}\big(F_k(1) r\big).
\end{align}
It is easy to verify that $F_k$ and $\tau_k$ are strictly monotone increasing. 
Let us further define a hybrid probability measure $q$ via
\begin{align}
\label{eqn::q}
q\big(k, \tau_k(r)\big) := \frac{\mu(k,r) f_k\big(t+\dt \tau_k(r)\big)}{F_k(1)}.
\end{align}
Then it is easy to verify that
\begin{align*}
q\big(k, \tau_k(r)\big) \tau_k'(r) 
=&\ \frac{\mu(k,r) f_k\big(t+\dt \tau_k(r)\big)}{F_k(1)} \tau_k'(r) \\
=&\ \frac{\mu(k,r) f_k\big(t+\dt \tau_k(r)\big)}{F_k(1)} \frac{F_k(1)}{F_k'\big(F_k^{-1}\big(F_k(1) r\big)\big)} \\
=&\ \frac{\mu(k,r) f_k\big(t+\dt \tau_k(r)\big)}{f_k\big(t+\dt F_k^{-1}\big(F_k(1) r\big)\big)} 
\myeq{\eqref{eqn::tau_k}}\ \mu(k,r).
\end{align*}
Therefore, by change of variables, 
\begin{align*}
 \mathcal{E}_{\mu, t, \dt}(\rho) =&\ \sum_{k} \int_{0}^{1} \dd r\ \mu(k,r)\ e^{-i \dt \frac{F_k(1)}{\mu(k,r)} \h_k}\ \rho\ e^{i \dt \frac{F_k(1)}{\mu(k,r)} \h_k} \\
=&\ \sum_{k} \int_{0}^{1} \dd \tau_k\ q\big(k, \tau_k\big) e^{-i\dt \frac{F_k(1) f_k(t+\dt\tau_k)}{\mu(k,r) f_k(t+\dt\tau_k)} \h_k} \rho e^{i\dt \frac{F_k(1)f_k(t+\dt\tau_k)}{\mu(k,r) f_k(t+\dt\tau_k)} \h_k} \\
\myeq{\eqref{eqn::q}}&\ \sum_{k} \int_{0}^{1} \dd \tau_k\ q\big(k, \tau_k\big) e^{-i\dt \frac{f_k(t+\dt\tau_k)}{q(k,\tau_k)} \h_k} \rho e^{i\dt \frac{f_k(t+\dt\tau_k)}{q(k,\tau_k)} \h_k} \\
\myeq{\eqref{eqn::Hk}}&\ \sum_{k} \int_{0}^{1} \dd \tau_k\ q\big(k, \tau_k\big) e^{-i\dt \frac{H_k(t+\dt\tau_k)}{q(k,\tau_k)}} \rho e^{i\dt \frac{H_k(t+\dt\tau_k)}{q(k,\tau_k)}}.
\end{align*}
Finally, by removing the subscript $k$ in $\tau$ for simplicity of notations,
\begin{align*}
\mathcal{E}_{\mu, t, \dt}(\rho)  = \sum_{k} \int_{0}^{1} \dd \tau\ q\big(k, \tau\big) e^{-i\dt \frac{H_k(t+\dt\tau)}{q(k,\tau)}}\ \rho\ e^{i\dt \frac{H_k(t+\dt\tau)}{q(k,\tau)}}\  \myeq{\eqref{eqn::c-qDrift}}\  \mathcal{E}^{\text{(c-qDrift)}}_{q,t,\dt}(\rho).
\end{align*}
\end{proof}

\section{Supplementary details for numerical experiments}

\subsection{Weights for time-independent splitting schemes}
\label{app::weights}
All the following weights for fourth-order splitting schemes come from a summary in \cite{ostmeyer_optimised_2023}:
\begin{itemize}[leftmargin=\leftmarginnew{}]
\item {\bf Forest-Ruth-Suzuki (FRS)} \cite[Sect.~3.2.3]{ostmeyer_optimised_2023} refers to the splitting scheme in \eqref{eqn::product_2_body} with cycle $q = 3$, where
\begin{align*}
\begin{array}{llllllll}
a_1 &= \frac{\gamma}{2}, & \qquad
b_1 &= \gamma, & \qquad
a_2 &= \frac{1-\gamma}{2},&\qquad
b_2 &= 1-2\gamma,\\
a_3 &= a_2, &\qquad
b_3 &= b_1, &\qquad
a_4 &= a_1,
\end{array}
\end{align*}
and $\gamma = (2-2^{1/3})^{-1} \approx 1.3512071919596578$.

\item {\bf Omelyan’s Forest–Ruth-Type (FRO)} \cite[Sect.~3.2.4]{ostmeyer_optimised_2023} refers to the splitting scheme in \eqref{eqn::product_2_body} with cycle $q = 4$, where
\begin{align*}
\begin{array}{llllllll}
a_1 &= 0.172 086 559 029 5143, & \qquad
b_1 &= 0.591 562 030 755 1568, & \\
a_2 &=-0.161 621 762 210 7222 ,&\qquad
b_2 &= 1/2 - b_1 ,\\
a_3 &= 1-2(a_1+a_2), &\qquad
b_3 &= b_2, &\\
a_4 &=a_2 , &\qquad
b_4 &=b_1 , &\\
a_5 &= a_1
\end{array}
\end{align*}

\item {\bf Suzuki's fourth-order scheme (Suz4)} \cite[Sect.~3.2.7]{ostmeyer_optimised_2023} refers to the splitting scheme in \eqref{eqn::product_2_body} with cycle $q = 5$, where
\begin{align*}
\begin{array}{llllllll}
a_1 &= 0.207 245 385 897 1879, & \qquad
b_1 &= 0.414 490 771 794 3757 , & \\
a_2 &=0.414 490 771 794 3757 ,&\qquad
b_2 &=0.414 490 771 794 3757 ,\\
a_3 &=1/2 - (a_1+a_2), &\qquad
b_3 &= 1 - 2 (b_1 + b_2), &\\
a_4 &=1/2 - (a_1+a_2), &\qquad
b_4 &= b_2, &\\
a_5 &= a_2 &\qquad
b_5 &= b_1\\
a_6 &= a_1 &\qquad
\end{array}
\end{align*}

\item {\bf Ostmeyer's optimized 4th order (Ost4)} \cite[Sect.~3.2.8]{ostmeyer_optimised_2023} refers to the splitting scheme in \eqref{eqn::product_2_body} with cycle $q = 5$, where
\begin{align*}
\begin{array}{llllllll}
a_1 &= 0.092 575 474 731 957 87, & \qquad
b_1 &= 0.254 099 631 552 9392 , & \\
a_2 &= 0.462 716 031 021 0738 ,&\qquad
b_2 &= -0.167 651 724 011 9692 ,\\
a_3 &=1/2 - (a_1+a_2), &\qquad
b_3 &= 1 - 2 (b_1 + b_2), &\\
a_4 &=1/2 - (a_1+a_2), &\qquad
b_4 &= b_2, &\\
a_5 &= a_2 &\qquad
b_5 &= b_1\\
a_6 &= a_1 &\qquad
\end{array}
\end{align*}

\end{itemize}

\begin{figure}
\captionsetup[subfigure]{labelformat=empty}
\subfloat[]{\includegraphics[width=0.9\textwidth]{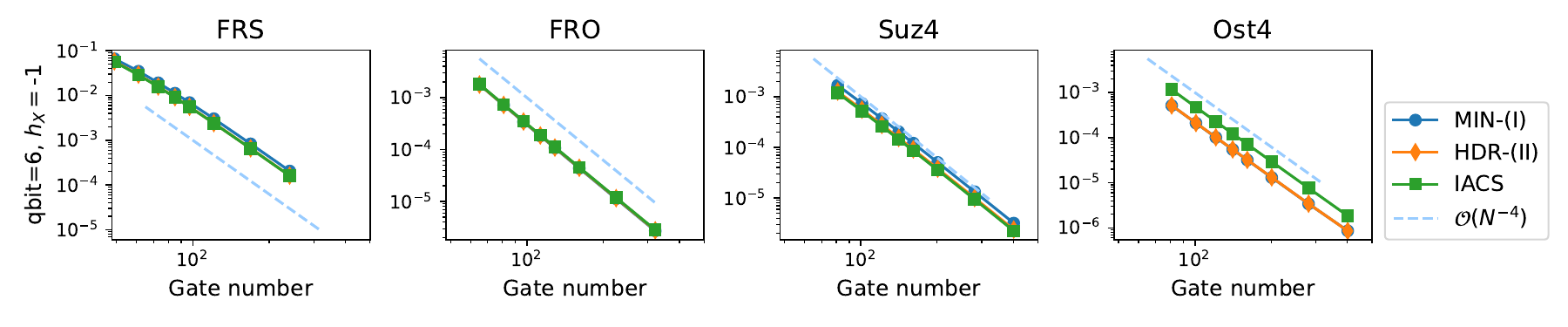}}\\
\subfloat[]{\includegraphics[width=0.9\textwidth]{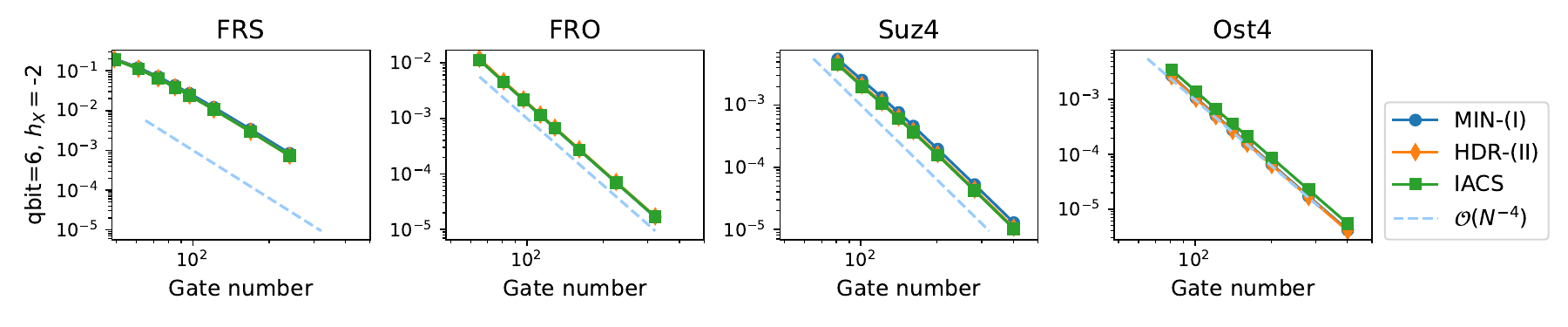}}\\
\subfloat[]{\includegraphics[width=0.9\textwidth]{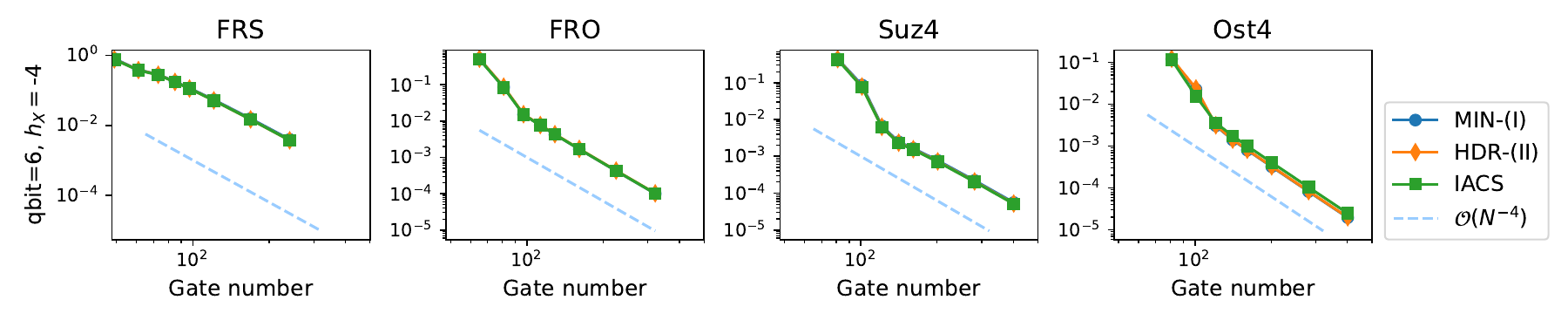}
}
\caption{Averaged simulation error in trace distance for  the quantum Ising model with respect to the gate counts $N$ for various time-dependent schemes. 
\enquote{MIN-(I)} refers to the scheme \eqref{eqn::product_td} with the ordering $H_1, \hat{p}, H_2$. HDR refers to the scheme in \eqref{eqn::product_td_v2}. IACS refers to the scheme in \cite{ikeda_minimum_2023}; see also \eqref{eqn::FRS_Magnus}. 
The blue dashed line indicates the scaling $\order{N^{-4}}$ and is the same for all subplots.}
\label{fig::Ising}
\end{figure}

\subsection{The quantum Ising chain}
\label{Ising}

We also consider the same problem in \cite[Sect.~5]{ikeda_minimum_2023} for $H(t) = f_1(t) \h_1 + f_2(t) \h_2$ where 
\begin{align*}
& \h_1 = \sum_{j=1}^{L} h_X \sigma_X^{(j)} \qquad
\h_2 = \sum_{j=1}^{L} J \sigma_Z^{(j)}\otimes \sigma_Z^{(j+1)} + h_Z \sigma_Z^{(j)}\qquad
 f_1(t) = \pi \sin(\pi t)\qquad f_2(t) = \pi,
\end{align*}
and $J = -1$, $h_Z = 0.2$. The periodic boundary condition is imposed for the  1D Ising chain, and $\sigma_{X,Y,Z}^{(j)}$ denotes the Pauli matrices for the site $j$. It can be easily verified that the maximally entangled state $\ket{+}^{\otimes L}$ is the ground state of $\h_1$ and we will use it as the initial condition. We consider $L = 6$ qubits, and various $h_X = -1, -2, -4$. We notice from \figref{fig::Ising} that when the magnitude of $h_X$ is smaller, the performance of \eqref{eqn::product_td_v2} is slightly better than IACS; overall, it is almost indistinguishable from IACS scheme for this example.

\end{document}